\documentclass[superscriptaddress,amsmath,amssymb,longbibliography,twocolumn]{revtex4-1}

\usepackage{mathrsfs}
\usepackage{MnSymbol}
\usepackage{complexity}
\usepackage{graphicx}
\usepackage{amsmath}
\usepackage{fullpage}
\usepackage{amsfonts}
\usepackage{amssymb}
\usepackage{dsfont}
\usepackage{amsthm}
\usepackage{tensor}
\usepackage{enumitem}
\usepackage{amsthm}
\usepackage{xcolor}
\usepackage{comment}
\usepackage{mathtools}
\usepackage{physics}
\usepackage{bbold}

\usepackage{appendix}
\usepackage[ruled,vlined]{algorithm2e}
\usepackage[colorlinks]{hyperref}
\hypersetup{
pdfstartview={FitH},
pdfnewwindow=true,
colorlinks=true,
linkcolor=blue,
citecolor=blue,
filecolor=blue,
urlcolor=blue}
\usepackage[capitalize]{cleveref}
\usepackage{enumitem}

\newtheorem{problem}{Problem}
\newtheorem*{problem*}{Problem}

\def\bH{{\bf H}}
\def\bA{{\bf A}}

\def\bD{{\bf D}}
\def\bX{{\bf X}}
\def\bY{{\bf Y}}
\def\bB{{\bf B}}

\def\Be{\textsc{Be}}

\def\bGamma{{\bf \Gamma}}

\newcommand{\wA}{\tilde{A}}
\newcommand{\wbA}{\tilde{\mathbf{A}}}
\newcommand{\mA}{\mathcal{A}}
\newcommand{\mX}{\mathcal{X}}
\newcommand{\mB}{\mathcal{B}}

\newcommand{\wphi}{\tilde{\phi}}
\newcommand{\wpsi}{\tilde{\psi}}
\newcommand{\nn}{\nonumber \\}

\def\ri{{\rm i}}

\def\rd{{\rm d}}

\def\cA{\mathcal{A}}
\def\cB{\mathcal{B}}

\def\cI{\mathcal{I}}

\def\cL{\mathcal{L}}

\def\cN{\mathcal{N}}
\def\cO{\mathcal{O}}
\def\cP{\mathcal{P}}

\def\cT{\mathcal{T}}

\def\cX{\mathcal{X}}
\def\cY{\mathcal{Y}}

\def\bC{\mathbf{C}}

\def\bL{{\bf L}}

\def\poly{\mathrm{poly}}

\def\one{{\mathchoice {\rm 1\mskip-4mu l} {\rm 1\mskip-4mu l} {\rm
1\mskip-4.5mu l} {\rm 1\mskip-5mu l}}}
\def\dqc1{\textsc{DQC1}}

\def\epshist{{\epsilon_{\rm hist}}}
\def\epstrunc{{\epsilon_{\rm trunc}}}
\def\epsbe{{\epsilon_{\rm be}}}

 \newcommand{\sket}[1]{| #1 \rangle}        
 \newcommand{\sbra}[1]{\langle #1|}  

\newcommand{\init}{\mathrm{in}}

\newcommand{\lb}{\left(}
\newcommand{\rb}{\right)}
\newcommand{\hist}{\mathrm{hist}}

\newcommand{\QLSA}{\mathrm{QLSA}}
\newcommand{\LCHS}{\mathrm{LCHS}}
\newcommand{\tcL}{\tilde{\cL}}

\makeatletter
\newtheorem*{rep@theorem}{\rep@title}
\newcommand{\newreptheorem}[2]{%
\newenvironment{rep#1}[1]{%
 \def\rep@title{#2 \ref{##1}}%
 \begin{rep@theorem}}%
 {\end{rep@theorem}}}
\makeatother

\theoremstyle{mytheorem}

\newtheorem{theorem}{Theorem}
\newreptheorem{theorem}{Theorem}
\newtheorem{lemma}[theorem]{Lemma}
\newreptheorem{lemma}{Lemma}
\newtheorem{claim}[theorem]{Claim}
\newreptheorem{claim}{Claim}

\newreptheorem{definition}{Definition}

\newreptheorem{remark}{Remark}
\newtheorem{corollary}[theorem]{Corollary}
\newreptheorem{corollary}{Corollary}

\newreptheorem{observation}{Observation}

\begin{document}

\title{Efficient quantum algorithm for linear matrix differential equations and applications to open quantum systems}

\author{Sophia Simon}
\email{sophia.simon@mail.utoronto.ca}
\affiliation{Google Quantum AI, Venice, CA 90291, United States}
\affiliation{Department of Physics, University of Toronto, Ontario M5S 1A7, Canada}

\author{Dominic W. Berry}
\affiliation{School of Mathematical and Physical Sciences, Macquarie University, Sydney, NSW 2109, Australia}

\author{Rolando D. Somma}
\affiliation{Google Quantum AI, Venice, CA 90291, United States}

\date{\today}

\begin{abstract}
We present an efficient, nearly optimal quantum algorithm for solving linear matrix differential equations, with applications to the simulation of  open quantum systems and beyond. For unitary or dissipative dynamics, the algorithm computes an entry of the solution matrix with query complexity $\widetilde{\mathcal{O}}(\nu  \mathcal{L} t/\epsilon)$, where the constant $\nu$ depends on the problem parameters, $\mathcal{L}$ involves a time integral of upper bounds on the norms of evolution operators, and $\epsilon$ is the error. In particular, $\nu \mathcal{L}$ is linear in $t$ for unitary dynamics and can be a constant for dissipative dynamics. Our result contrasts prior quantum approaches for differential equations that typically require exponential time for this problem due to the encoding in a quantum state, which can lead to exponentially small amplitudes. We demonstrate the utility of the algorithm through an end-to-end application, namely the simulation of dissipative dynamics for non-interacting fermions, which can be extended to other quantum and classical systems. We compare with classical algorithms and give evidence of polynomial quantum speedups for systems in a lattice, which become more pronounced for systems with long-range interactions and can be shown to be exponential in general. We also provide a lower bound of $\Omega(\nu \mathcal{L} t/\epsilon)$ for unitary or dissipative dynamics that proves our algorithm is optimal up to logarithmic factors.
\end{abstract}

\maketitle

\section{Introduction}
\label{sec:intro}

While quantum computers are anticipated to tackle problems currently intractable for classical computers -- most notably the simulation of quantum dynamics~\cite{berry2007efficient,berry2015truncated,low2017optimal} for applications in chemistry~\cite{wecker2014gate,babbush2015chemical,low2025fast} and materials science~\cite{Llo96,SOGKL02,bassman2020towards} -- understanding their utility to broader areas like linear algebra and general differential equations remains a significant challenge.
In particular, practical applications for known
quantum differential equation solvers~\cite{berry2014high,Berry2017lineardifferential,krovi2023improved,liu2021efficient,jin2022quantum,an2023lchs,low2025optimal} 
and related quantum algorithms for linear algebra~\cite{harrow2009quantum,childs2017quantum} have been difficult to establish due to the stringent constraints required for these algorithms
to be efficient. 
Among these, a key hurdle is due to the encoding,
where the solution to the computation
is encoded in a quantum state that has to be normalized, and
the quantum algorithms can only output restricted computations,  such as expectation values of normalized vectors. This makes specific tasks, like extracting individual entries of a solution, impractical because normalization factors often decrease exponentially with the number of qubits~\cite{aaronson2015read}.

Motivated by these practical limitations of known quantum methods,
 we present a quantum algorithm  for the matrix version of linear differential equations.
Such matrix equations have broad applications in control theory~\cite{nguyen2009solving}, machine learning~\cite{sarkka2019applied}, and physical simulations~\cite{purkayastha2022lyapunov}.
Rather than encoding the solution matrix $\bX(t)$
in a quantum state via vectorization,
we consider an approach
designed to output the entries of $\bX(t)$
directly, using an idea based on   ``history states''
to perform certain integrals. This
avoids the issue due to the normalization of states
and unlocks applications for quantum differential equation solvers, such as the simulation of  large open quantum systems~\cite{purkayastha2022lyapunov}, as we will explain.

Specifically, we start from an equation of the form $\frac{\rd}{\rd t}\bX(t)=\bA^\dagger \bX(t) + \bX(t) \bB + \bC$,
where $\bA$, $\bB$, $\bC$ and $\bX(t)$ are possibly time-dependent matrices of dimension $N = 2^n$. 
The problem is to compute an entry of $\bX(t)$ in some specified basis, like the computational basis. For quantum algorithms, a common assumption is that the matrices 
are accessed efficiently via so-called block-encodings. 

We present a quantum algorithm
that solves this problem efficiently and nearly optimally, in time polynomial in $n$ and $t$ 
for some relevant cases. When $\bA$ and $\bB$ are skewHermitian or negative semidefinite, 
the query complexity is $\widetilde \cO(\nu \cL t/\epsilon)$, where $\cL$ involves a time integral of 
upper bounds on the norms of evolution operators, $\nu$ is a constant, and $\epsilon$ is the error. These cases correspond to unitary or dissipative dynamics, for which the query complexities can be quadratic or even linear in $t$. The gate complexities contain additional factors that are $\poly(n)$. More generally, our algorithm is efficient if the norms of the evolution operators remain bounded.

For comparison, we note that matrix differential equations can be reduced to vector differential equations via vectorization, where the dimension is $N^2$.
This could be solved via 
Refs.~\cite{berry2014high,Berry2017lineardifferential,krovi2023improved,liu2021efficient,jin2022quantum,an2023lchs,low2025optimal}; however, doing so would again suffer from the state normalization issues and the complexity could be exponential.
Indeed,   Ref.~\cite{somma2025quantum} proves an exponential separation between the state- and matrix-access models.

One application of our algorithm is the simulation of dissipative dynamics in certain open quantum systems. We illustrate this through a free-fermion system interacting with a bath of free fermions. 
The covariance matrix satisfies a linear matrix differential equation of the form that our algorithm can solve.
For this example, the query complexity of our algorithm is $\widetilde \cO(\mu t_e)$
under realistic assumptions and constant error, where $t_e = \frac 1 \gamma + \beta$ is an effective time, $\mu$ a constant,
and $\beta \ge 0$
the inverse temperature.  Since the relevant matrices are now
associated with physical Hamiltonians, 
they might be constructed efficiently with gate cost $\poly(n)$ under broad assumptions. Similar results would apply to other systems, including non-interacting bosons and coupled classical oscillators, generalizing the problem in Ref.~\cite{babbush2023exponential} and allowing for dissipation.

Our result contrasts classical algorithms, 
whose complexity is polynomial in $N$ in the worst case; indeed, our algorithm can solve $\BQP$-complete problems efficiently due to Ref.~\cite{childs2009universal}. 
This gives the strongest evidence for general exponential quantum advantage. Nevertheless, we note that
classical algorithms can also leverage structures of the system in physics simulations, such as the locality of the interactions.
For dissipative dynamics in a lattice in $D \ge 1$ spatial dimensions, classical algorithms 
of complexity $\widetilde \cO((\mu t_e)^{1+D})$ might be possible.
This suggests a polynomial quantum speedup -- e.g., a quartic speedup for $D=3$ -- that becomes more pronounced   for long-range interactions.
The quantum memory is $\cO(\log(N))$ 
while efficient classical algorithms require exponentially larger memory $\Omega(N)$.

We conclude by showing that our algorithm is almost optimal in query complexity in the assumed access model. To this end, we consider a large class of instances with 
slight or no dissipation and prove a lower bound $\Omega(\nu \cL t/\epsilon)$ for the query complexity
that matches the upper bound  up to logarithmic factors. Further improvements
could only follow from using other structures or more powerful access models.

\section{Problem and main result}
\label{sec:LMDE}

We consider the matrix differential equation
\begin{align}
\label{eq:main}
    \frac{\rd}{\rd t} \bX(t) = \bA^\dagger \bX(t) + \bX(t) \bB + \bC \;,
\end{align}
subject to the initial condition $\bX(0)=\bD$.
Here,  $\bA$, $\bB$, $\bC$, $\bX(t)$
are complex matrices of dimension $N \times N$, and $t \ge 0$ is the time parameter. 
This is a type of Sylvester differential equation. When $\bA$ and $\bB$ are time independent, the unique solution is
\begin{align}
\label{eq:mainsolution}
    \bX(t) = \int_0^t \rd s \;e^{(t-s) \bA^\dagger} \bC e^{(t-s) \bB} +  e^{t \bA^\dagger} \bD e^{t\bB} \;.
\end{align}
For time-dependent matrices, the exponentials are replaced as
$e^{(t-s) \bA^\dagger} \mapsto \cT e^{\int_s^t \rd \tau \bA^\dagger(\tau)}$
and $e^{(t-s) \bB} \mapsto \tilde \cT e^{\int_s^t \rd \tau \bB(\tau)}$,
where $\cT$ 
is time ordering and $\tilde \cT$ is anti-time ordering.

We seek a quantum algorithm that computes a given entry of $\bX(t)$. This requires making assumptions
on how the matrices can be accessed. We  consider the
 block-encoding model~\cite{berry2015truncated,LC19,gilyen2019qsvt},
for which we assume access to unitary quantum circuits $U_\bA,U_\bB, U_\bC, U_\bD \in \mathbb C^{M\times M}$, $M \ge N$,
where the first blocks of these matrices are $\bA/a$, $\bB/b$, $\bC/c$, and $\bD/d$, respectively; that is $\lb \bra 0 \otimes \one \rb U_\bA \lb \ket 0 \otimes \one \rb =\bA/a$ for $a \geq \norm{\bA}$ and similarly for the others.
For time-dependent matrices, these unitaries additionally contain a finite qubit time register $\ket{t}$ such that $\lb \sbra{0,t} \otimes \one \rb U_\bA \lb \sket{0,t} \otimes \one \rb =\bA(t)/a$.

\begin{problem}
    Let $U_\bA$, $U_\bB$, $U_\bC$, $U_\bD$
    be block-encodings of $\bA/a$, $\bB/b$, $\bC/c$,
    $\bD/d$, respectively,
    for given $a,b,c,d>0$. Let $U_{\phi}$ and $U_{\psi}$ be state preparation unitaries such that $U_{\phi}\ket{0} = \ket{\phi} \in \mathbb C^N$, $U_{\psi}\ket{0} = \ket{\psi} \in \mathbb C^N$.
    Let $t \ge 0$
    be the evolution time and $\epsilon > 0$ the error tolerance. The goal is to output an entry $\bra \phi \bX(t) \ket \psi$
    with additive error at most~$\epsilon$.
\label{prob:matrix_entry}
\end{problem}

Our main result is an efficient quantum algorithm 
for this problem. For clarity, we first present a simplified statement valid only for the time-independent case, which is useful for instances
describing unitary or dissipative dynamics, and then comment on a more general result.  
Recall for a matrix $\bY$, the log-norm is the largest eigenvalue of $(\bY+\bY^\dagger)/2$.

\begin{theorem}
\label{thm:main}
    Let $\mu:=\max\{a,b\}$ and define 
     \begin{align}
         {\cL} := \max_{\bY \in \{\bA, \bB\}} \int_0^t \rd s\, e^{2s \xi_\bY} + \frac{d}{c} e^{2t \xi_\bY},
    \end{align}
    where $\xi_{\bY}$ is a known upper bound on the log-norm of $\bY$.
   Then, Problem~\ref{prob:matrix_entry} can be solved  w.h.p. using
    \begin{align}
    \label{eq:LCHScomplexity}
        \widetilde{\cO} \lb 
         \frac{c \cL}{\epsilon} \times t \mu \rb
    \end{align}
    queries to the available block-encodings and inverses, where the $\widetilde \cO$ notation hides subdominant logarithmic factors.
    The number of arbitrary gates is larger than the query complexity by polylogarithmic factors.
\end{theorem}

For matrices with negative log-norm, 
  we have  $\cL \leq 1/(2|\xi_\bY|) + d/c$.
These instances occur when the dynamics is  dissipative. Also, for these cases we can always cutoff $t$ at some sufficiently large value that depends on $1/|\xi_\bY|$ introducing negligible errors.
For unitary dynamics where $\xi_\bY=0$, we have $\cL=t+d/c$, and the query complexity is quadratic in $t$. 
Our quantum algorithm is then efficient for these cases. 

In Appendix~\ref{app:QAlgorithm} we actually prove a more general result that shows our algorithm is efficient even for matrices with positive log-norm.
Explicitly, if $\tilde \cL$
is a known upper bound on
    \begin{align}
        \max_{\bY \in \{ \bA, \bB \}} \int_{0}^t \rd s \norm{e^{s \bY}} + \frac{d}{c} \max_{s \in [0,t]} \norm{e^{s \bY}}  \;,
    \end{align}
the query complexity is
\begin{align}
\label{eq:QLSAcomplexity}
        \widetilde{\cO} \lb 
         \frac{c}{\epsilon} \times  \mu \tilde \cL
          \times \tilde \cL\max_{\substack{\bY \in \{ \bA, \bB \} \\
        s \in [0,t]}} \norm{e^{s \bY}} \times \log \lb t \mu \rb   \rb \;.
    \end{align}
This result might improve upon Thm.~\ref{thm:main}, provided that $\|e^{s\bY}\|$ is bounded or grows at most polynomially with time (e.g., transient dynamics).
Our algorithm for this case is based on a different approach.
 
We refer to Thms.~\ref{thm:mainformal} and \ref{thm:be-approach_time-indep},
and Table~\ref{tab:compare} in Appendix~\ref{app:QAlgorithm} for precise statements, where we 
distinguish the number of queries to $U_{\phi}$ and $U_{\psi}$ from those to the available block-encodings. 
To achieve arbitrary confidence $p \le 1$, the overhead includes a factor $\log(1/(1-p))$~\cite{knill2007optimal}.
The result extends to time-dependent matrices, see Appendix~\ref{app:QAlgorithm_time-dep}. 
The query complexity is essentially the same, and the gate complexity has an additional logarithmic dependence on the norm of the first derivatives of the matrices.

\section{Quantum algorithm}

The key idea is to reduce $\bra \phi \bX(t) \ket \psi$ to an overlap estimation problem with history states that encode the solution to a differential equation at different time steps in superposition~\cite{berry2014high}.
In particular, for $\bA$ and $\bB$ time independent, we prepare (sub-) normalized versions of the following history states:

\begin{align}
\nonumber
    \ket{\phi_{\hist}} &:= \sum_{m=0}^{M-1} \ket{m}e^{t\bA \frac{m}{M}}\ket{\phi}  + \sum_{m=M}^{M + R-1} \ket{m} e^{t \bA} \ket{\phi}, \\
    \nonumber
    \ket{\psi_{\hist}} &:= \sum_{m=0}^{M-1} \ket{m}e^{t \bB \frac{m}{M}}\ket{\psi}  + \sum_{m=M}^{M + R-1} \ket{m} e^{t \bB} \ket{\psi}.
\end{align}
Here, $\ket{m}$ is a ``clock'' register keeping track of the time step, $M \sim t \mu$ is the number of time steps and $R \sim \mu d/c$, is the number of steps for which the solution at the final time $t$ is held constant. In principle, other choices for $R$ might be more optimal depending on the problem, but this choice suffices for us.

Furthermore, letting $ I_\bC:=\int_{0}^{\frac{t}{M}} \rd \tau \, e^{\tau \bA^\dagger} \bC e^{\tau \bB}$, we can
decompose the integral term of $\bX(t)$ as
\begin{align}
    \int_0^t \rd s  \; e^{(t-s) \bA^\dagger} \bC e^{(t-s) \bB} = \sum_{m=0}^{M-1} e^{t\bA^\dagger \frac{m}{M}} I_\bC e^{t\bB \frac{m}{M}}.
\end{align}
Defining now
\begin{align}
\label{eq:cI}
    \mathcal{I} := \sum_{m=0}^{M-1} \ketbra{m}{m} \otimes I_{\bC} + \sum_{m=M}^{M+R-1} \ketbra{m}{m} \otimes \frac{\bD}{R},
\end{align}
we then obtain the desired reduction:
\begin{align}
\label{eq:overlapreductionmain}
    \bra \phi \bX(t) \ket \psi = \bra{\phi_{\hist}} \mathcal{I} \ket{\psi_{\hist}}\;.
\end{align}
This expression is exact and
Problem~\ref{prob:matrix_entry} is equivalent to an overlap estimation problem with history states.

The history states   can be prepared through different methods and the best choice depends on the behavior of $\norm{e^{s \bA}}$ and $\norm{e^{s \bB}}$ as functions of $s \in [0,t]$.
We use an approach based on
applying controlled time evolution operators via the Linear Combination of Hamiltonian Simulation (LCHS) framework~\cite{an2023lchs, an2023betterlchs, low2025optimal} and another approach based on
quantum linear system solvers~\cite{berry2014high, Berry2017lineardifferential, krovi2023improved, Berry2024time-dep_history}.

In the LCHS approach, we prepare subnormalized versions of $\ket{\phi_{\hist}}$ and $\ket{\psi_{\hist}}$ as follows. We use LCHS to construct the controlled time evolution 
\begin{align}
   \sum_{m=0}^{M-1} \ketbra{m}{m} \otimes e^{\frac{tm}{M} \lb \bY - \xi_\bY \one \rb} + \! \! \! \sum_{m=M}^{M+R-1} \ketbra{m}{m} \otimes e^{t \lb \bY -\xi_\bY \one \rb}
\end{align}
with $\bY \in \{ \bA, \bB \}$. Then, we apply it to a system-clock state prepared in
\begin{align}
    \frac{1}{\sqrt{\cN_{\xi_{\bY}}}} \lb \sum_{m=0}^{M-1} e^{t \frac{m}{M} \xi_{\bY}} \ket{m}  + \sum_{m=M}^{M+R-1} e^{t \xi_{\bY}} \ket{m} \rb \otimes \ket {\cal Y},
\end{align}
with $\cY=\phi$ or $\cY=\psi$.  

In Appendix~\ref{app:QAlgorithm}, we show that $\cN_{\xi_{\bY}} \in  \cO \lb \mu \cL \rb$.
The query complexity for preparing subnormalized history states in this way is $\cO \lb t \mu \log (1/\epsilon) \rb$~\cite{low2025optimal}.  Additionally, since $\|I_\bC\| \le c/\mu$ and $\|\bD/R\|\le c/\mu$,
we can block-encode $\cI$ in Eq.~\eqref{eq:cI} with block-encoding constant  $\lambda_\cI \propto c/\mu$.
It follows that Eq.~\eqref{eq:overlapreductionmain}
is proportional to an overlap with subnormalized matrices and states, where the normalization constant is 
proportional to $(c/\mu) \times (\mu \cL)=c\cL$. To estimate $\bra \phi \bX(t) \ket \psi$ within additive error $\epsilon$, we then need
to set the precision $\epsilon'=\epsilon/(c\cL)$ in the overlap estimation, giving a factor $1/\epsilon'=c\cL/\epsilon$ to the complexity using Heisenberg-limited approaches~\cite{knill2007optimal}.
Adding the complexity of initial state preparation,
we obtain Eq.~\eqref{eq:LCHScomplexity} for the overall query complexity.

In the linear systems approach, the states $\ket{\phi_{\hist}}$ and $\ket{\psi_{\hist}}$ are  solutions to systems of linear equations
resulting from a discretization of the differential equations. Preparing normalized versions of these history states has query cost that is linear in the condition number of the linear systems~\cite{childs2017quantum,costa2024linearsystem}, which we prove to be $\cO(\mu \tilde \cL)$.
(This generalizes prior results to
the case of positive log-norms, as long as  $\|e^{s\bA}\|$ and $\|e^{s\bB}\|$ 
are bounded.)
We also show $\|\ket{\phi_{\hist}}\| \times\|\ket{\psi_{\hist}}\| \in \cO ( \mu \tilde \cL \max_{s,\bY } \norm{e^{s \bY}} ) $.
Like in LCHS, for error tolerance $\epsilon$ in the entry of $\bX(t)$, we need to estimate the overlap with the subnormalized history states and the block-encoding of $\cI$ within error $\epsilon' \in \cO \lb \epsilon/(\|\ket{\phi_{\hist}}\| \times\|\ket{\psi_{\hist}}\| \times \lambda_{\cI} )\rb$. 
The number of queries to the history states and the block-encoding of $\cI$ is then   $\cO \lb 1/\epsilon' \rb$~\cite{knill2007optimal}. 
The overall query complexity to the available block-encodings and state preparation unitaries is therefore   $\widetilde{\cO} \lb  \frac{c \mu}{\epsilon} \tilde \cL^2 \times \max_{s,\bY} \norm{e^{s \bY}} \times \log \lb t \mu \rb \rb$, where the additional factor $\log (t \mu)$ stems from the truncation order of a Taylor series approximation to the matrix exponentials in the linear systems. 
This is Eq.~\eqref{eq:QLSAcomplexity}.
Note that we also require accurate estimates of the norms of the history states. These can be obtained using a method in Ref.~\cite{dalzell2024shortcut} with similar  complexity.

In Appendix~\ref{app:preconditioning} we also show how to use preconditioning to improve the complexity of the linear systems approach in certain cases, which might be of independent interest.

\section{Dissipative dynamics}

We illustrate an application to dissipative dynamics
of large free-fermion systems~\cite{purkayastha2022lyapunov}. 
The results can be extended to other physical systems, such as classical systems of coupled harmonic oscillators or bosonic quadratic systems, with similar conclusions. 

Let the system-bath Hamiltonian be $H=H_S + H_B + H_{\rm int}$, where $H_S = \sum_{j,k=1}^{N_s} \alpha_{jk} c^\dagger_j c^{}_k$, $H_B  = \sum_{l=1}^{N_b} \epsilon_l b^\dagger_l b^{}_l$, and $H_{\rm int}=\sum_{j=1}^{N_s} \sum_{l=1}^{N_b} \nu_{jl} c^\dagger_j b^{}_l +H.c.$
are the Hamiltonians for the system, the bath, and the system-bath interactions, respectively.
Here, $c^\dagger_j$ and $b^\dagger_l$
are the fermionic creation operators for the system and bath,
which satisfy the anticommutation relations
$\{c^\dagger_j, c^{}_k\}=\delta_{j,k}$
and $\{b^\dagger_l, b^{}_m\}=\delta_{j,k}$. The coefficients $\alpha_{jk} \in \mathbb C$ define the Hermitian $N_s \times N_s$ matrix $\bA$, and we are interested in the case where $N_s$ is very large and the interactions are small.
In the Heisenberg picture, a time evolved operator $O(t)$
satisfies $\frac{d}{dt} O(t) = \ri [H,O(t)]$. If the initial state of the full system-bath is $\rho$, then $\frac{\rd}{\rd t}\langle O(t)\rangle=\ri \langle[H,O(t)]\rangle$, where $\langle .\rangle=\tr(\rho.)$ is the expectation.

Consider the covariance matrix $\bX(t)$ whose entries are  the two-point correlators $\langle c^\dagger_j c^{}_k(t)\rangle$ at time $t$. Under a standard Markovian approximation, the evolution of this matrix can be shown to satisfy
 \begin{align}
 \label{eq:QLEmain}
      \frac{\rd}{\rd t} \bX(t) = \bB^\dagger \bX(t)
      + \bX(t) \bB + \bC \;,
  \end{align}
  where $\bB:=-\ri \bA  - \bGamma$, and ${\bGamma}$ is a positive semidefinite matrix that results from $H_{\rm int}$; that is, ${\bGamma}$ models the dissipation of the system. The matrix $\bC$ models the noise and quantum fluctuations due to the bath. 
  We give the  derivation of Eq.~\eqref{eq:QLEmain} in Appendix~\ref{app:QLE}. 
If the system evolves towards its thermal equilibrium at inverse temperature $\beta \ge 0$,
the fixed point is $\bX_\beta=(\one + e^{\beta \bA})^{-1}$,
which is the Fermi-Dirac distribution  for the Gibbs state. To reach this fixed point via Eq.~\eqref{eq:QLEmain}, the matrices must satisfy $\bC = \bX_\beta \bGamma + \bGamma \bX_\beta$.

We can readily apply our algorithm to this example.
We assume $H_S$ can be accessed efficiently via a block-encoding $U_\bA$ of $\bA/a$   and also assume a block-encoding $U_\bGamma$ of $\bGamma/\gamma$. 
Note that $\|\bX_\beta\|\le 1$ and 
a block-encoding $U_{\bX_\beta}$ for $\bX_\beta$ within error $\epsilon'$ can be constructed
using $\cO(\beta a \log(\beta a/\epsilon'))$
queries to $U_\bA$ and its inverse~\cite{gilyen2019qsvt,chen2024quantum}. 
Together, they provide a block-encoding of $\bC/(2\gamma)$ within error $\varepsilon$ using standard techniques,
also with cost $\cO(\beta a \log(\beta a/(\varepsilon \gamma)))$.
For the initial condition
we consider an infinite-temperature initial state, so that $\bD=\frac 1 2 \one$.

The characteristic time to observe dissipation in this problem is $t \sim 1/\gamma$, and in the weak interaction limit, we assume $\gamma \ll a$.
Then, the parameters in Thm.~\ref{thm:main} are as follows: $\mu = a$, $c = 2 \gamma$, $d =1/2$. Also, $\cL=\cO(1/\gamma)$, and the  query complexity from Eq.~\eqref{eq:LCHScomplexity} is $\widetilde \cO(a/(\gamma \epsilon)\times \log(a/\gamma))$ calls to $U_\bA$ and inverses,
and the additional query cost
$\widetilde \cO((\beta a  /\epsilon) \times \log(\beta a))$ for the calls to $U_\bC$ and inverse, see Thm.~\ref{thm:mainformal}.

\section{Quantum speedups}

The query complexity of our algorithm in the previous example is only logarithmic in $N_s$ under reasonable assumptions, but it depends almost
linearly on some effective time $t_e \sim 1/\gamma + \beta$ if $\epsilon$ is constant. Hence, for quantum systems with local interactions in a lattice  in $D$ spatial dimensions, a lightcone argument gives a relevant 
system size $N_e \sim (a t_e)^{D}$ for local properties. Since now $t_e$ and $N_e$ are polynomially related, the quantum speedup for these instances can only be polynomial. This contrasts the general case of  Hamiltonians with non geometrically-local interactions, where the quantum speedup is exponential.

Consider now the problem of computing an entry of $\bX(t)$ in the standard basis and within constant, small error. 
A main classical algorithm for this problem is based on Krylov subspaces~\cite{behr2019solution} and the related Kernel Polynomial method~\cite{weisse2006kernel}.
These involve matrix-vector multiplications,
but since the matrices are sparse for these instances, the cost of each such multiplication is $\cO(N_e)$. In Appendix~\ref{app:classicalalgs} we argue that
the classical complexity   is at least linear in space-time or $\Omega(a t_e N_e)=\Omega((a t_e)^{1+D}))$. The argument is based on a  selection of Krylov subspaces generated by repeated action of $\bA$ and $\bB$ on some initial vectors. 
In the best case, the classical complexity contains
an additional factor $\sim(a t_e)^{D}$ over the quantum query complexity. This suggests a polynomial quantum speedup that is more pronounced as $D$ increases. For example, the speedup is quartic when $D=3$, giving evidence of the practicality of our algorithm~\cite{babbush2021focus}. See Ref.~\cite{simoncini2016computational}
for other classical algorithms,
often with a superlinear scaling.

These polynomial quantum speedups are based on a lightcone argument. The speedups could be more significant for systems with longer-range interactions (e.g., $\sim 1/r^3$) or when the matrices are time-dependent, since our algorithm has a mild overhead for these cases in contrast to classical algorithms. Moreover, the quantum memory is $\cO(\log N)$,
while these classical algorithms require storing $\Omega(N)$ bits.

We also note that, in some cases,
faster quantum algorithms for simulating dissipative dynamics might be possible using the idea of spectral  amplification~\cite{king2026quantum}. For these cases, 
the exponential operators admit a more efficient LCHS decomposition using the Hubbard-Stratonovich identity~\cite{chowdhury2016quantum}.

\vspace{-.1cm}

\section{Optimality}

The prior discussion raises the question of whether improved quantum algorithms are possible. The answer is negative: 
 we  prove a lower bound $\Omega(\cL t/\epsilon)$ on the query complexity applicable to certain instances, showing our  algorithm is almost optimal. 
This lower bound is mainly due to the integral term in Eq.~\eqref{eq:mainsolution}.
\begin{theorem}[Query lower bound]
\label{thm:lowerbound}
   Let $\bA \preceq 0$ be diagonal, $\bC=\one$, $\bD=\bB=\bf 0$, and assume block-encoding access to $\bA$ (with constant $a=1$). Let $t \ge 6$ and $\epsilon>0$ satisfy $\epsilon \le t/100$.
    Then, computing $\bra{0}\bX(t)\ket{0}$ within error $\epsilon$
    and success probability $p \in (3/4,1)$ requires $\Omega (\frac{\cL t}{\epsilon}|\log(1-p)|)$ uses of the block-encoding and its inverse, where $\cL:=\int_0^t \rd s \; e^{2 s\xi_\bA}$ and $\xi_\bA$ is the log-norm of $\bA$. 
\end{theorem}

For these instances, $\|\bD\|=0$, $\|\bA\|\le 1$, $\|e^{t\bA}\| \le 1$, $\|e^{t\bB}\|=1$, and $\|\bC\|=c=1$. Then $\mu=1$, and $\cL \le t$, and the 
 bound from Thm.~\ref{thm:main} is $\widetilde{\cO}(\cL t/\epsilon)$. Since the
lower bound is tight up to logarithmic factors,
 our algorithm is optimal and no further general asymptotic improvements in query complexity are possible within the block-encoding access model.

The proof of Thm.~\ref{thm:lowerbound}
is in Appendix~\ref{app:lowerbound}, where we give a slightly stronger claim. It is based on a reduction from quantum phase estimation (QPE) and a decision version of that problem. We adapt the lower bound $\Omega(\frac 1 \delta |\log (1-p)|)$ for QPE within error $\delta$ and success probability $p$ in  Ref.~\cite{mande2023tight}. One can show
that computing $\bra{0}\bX(t)\ket{0}$ within error $\epsilon$ and success probability $p$ allows one to solve the decision problem. 
For these instances, we consider $\bA=-\sin\theta \ketbra 0-\sum_{n >0}\ketbra n$, where $ 0< \theta \le \pi/16$ and  $\bA \preceq 0$. The solution to the matrix differential equation satisfies $\bra{0}\bX(t)\ket{0}= \cL=(1-e^{-2t \sin \theta})/(2\sin \theta)$.
In particular, if the goal is to decide if $\theta=\delta>0$
or $\theta \ge 2\delta$, the cases can be distinguished if we set the precision to $\epsilon =\cO(\cL t \sin \delta)$. Then, $\cL t/\epsilon =\cO(1/\delta)$ and the QPE lower bound implies the lower bound $\Omega(\cL t |\log(1-p)|/\epsilon)$
for entry estimation.

\section{Conclusions and outlook}

Motivated by the limitations of existing quantum differential equation solvers,
we considered the problem of matrix differential equations and provided an optimal quantum algorithm. Our results extend the range of applications for quantum computers in quantum simulation and beyond; particularly, we avoid the readout issues of prior approaches (e.g., HHL) due to the normalization of quantum states. 
Moreover, for some relevant instances, we can 
access the initial conditions in a matrix $\bD$ efficiently,
while this could also be challenging for prior methods
that encode them in a quantum state (known as the data loading problem).

To highlight the advantages of our approach, we discussed an end-to-end application to dissipative quantum dynamics. 
This problem could also be tackled
via a reduction to differential equations in vectorized form,
but that approach would have been inefficient.
When analyzing this application, we provided evidence for polynomial quantum speedups over the best known classical algorithms.

\section{Acknowledgements}

We thank Ryan Babbush, Guang Hao Low, Nicholas Rubin, Malte Schade, and Nathan Wiebe for discussions.
DWB worked on this project under a sponsored research agreement with Google Quantum AI.
This project is supported by Australian Research Council Discovery Projects DP210101367, DP220101602, and DP260102543.

\bibliography{refs.bib}

\clearpage

 \onecolumngrid
 \appendix
 
{\center{\Large Supplementary Material}}

\section{Quantum algorithm for the time-independent case: Proof of Thm.~\ref{thm:main}}
\label{app:QAlgorithm}

Following Eq.~\eqref{eq:mainsolution}, our goal is to compute the entry $\bra \psi \bX(t) \ket \phi$ within error $\epsilon \ge 0$. To this end, we will prepare history states that, in particular, will allow us to compute both the integral part of the solution and the part that depends on the final time $t$. In this appendix, we consider time-independent matrices. The time-dependent setting is discussed in Appendix~\ref{app:QAlgorithm_time-dep}.
We present two methods for preparing history states, neither of which is strictly better than the other, and the best approach will depend on the instances.

First, we discuss preparing history states via an approach based on quantum linear system solvers~\cite{berry2014high, Berry2017lineardifferential, krovi2023improved,Berry2024time-dep_history}. 
Then, we present an method for preparing history states based on block-encoding controlled time-evolution operators via Linear Combination of Hamiltonian Simulation (LCHS)~\cite{an2023lchs, an2023betterlchs, low2025optimal}.
The LCHS approach has generally better complexity in terms of queries to the initial state preparation unitaries $U_{\phi}$ and $U_{\psi}$, and the overall query complexity can have slightly better scaling than the linear systems approach in the case of (nearly) unitary dynamics or exponential growth/decay as in dissipation. On the other hand, the linear systems approach can have exponentially better scaling than the LCHS approach in the case of transient dynamics where $\bA$ and $\bB$ have positive log-norm but the solution remains bounded or grows only polynomially as a function of time; see, e.g., Table~\ref{tab:compare} for a more detailed comparison.

The complexity of linear system based quantum differential equation solvers~\cite{berry2014high, Berry2017lineardifferential, krovi2023improved,Berry2024time-dep_history} depends on the condition number of the linear system encoding the discretized time dynamics. In comparison to prior work, we prove a tighter bound on the condition number in terms of a time integral of the spectral norm of the evolution operator, which might be of independent interest. Our improved bound on the condition number is applicable to arbitrary matrices, including time-dependent matrices. Recent work~\cite{An2026fastforwarding} also proved tighter bounds on the condition number of linear system based quantum differential equation solvers but their analysis is restricted to dissipative dynamics (non-positive log-norm) and not general matrices.

\subsection{History state preparation via quantum linear system solvers}

The main idea behind linear systems based quantum differential equation solvers~\cite{berry2014high, Berry2017lineardifferential, krovi2023improved,Berry2024time-dep_history} is to reduce a linear differential equation, such as $\frac{\rd}{\rd t} \vec x(t) = A  \vec x(t)$, to a large system of linear equations of the form $\cA \vec \cX = \vec \cB$, where $\cB$
contains the information of the initial condition $\vec x(0)$. The solution to this system is such that $\vec \cX$ is formed by many blocks, each corresponding to $\vec x(s)$ for many discrete values of $s$, and $0 \le s \le t$.
That is, $\vec \cX$ contains the information about the history of the solution, and corresponds to what we call the ``history state''.

To be more specific, let us consider the problem of producing the history state for the following linear vector differential equation:
\begin{align}
    \frac{\rd \vec x }{\rd t} = A \vec x(t), \quad \vec x(0) = \vec x_{\init},
\end{align}
with $A \in \mathbb{C}^{N \times N}$  assumed to be time-independent and $\vec x_{\init} \in \mathbb{C}^N$. The solution is
\begin{align}
    \vec x(t) = e^{tA} \vec x_{\init}.
\end{align}

Let $\ket{x_{\init}}:=\vec x_{\init} / \|\vec x_{\init}\|$ be a normalized quantum state that encodes the initial vector $\vec x_{\init}$ and let $U_x$ be a unitary that prepares it as $U_x \ket{0} =\ket{x_{\init}}$.
The linear systems approach allows us to implement a unitary quantum circuit $U_{\mathrm{hist}}^{(\QLSA)}$ such that
\begin{align}
\label{eq:historystateunitary}
  U_{\mathrm{hist}}^{(\QLSA)}\ket{0}_{\rm a}\ket{0}_{\rm clock}\ket{0}_{\hist}  \mapsto \sqrt{p_0} \ket 0_{\rm a}\frac{1}{\sqrt{\cN_A}}\left( \sum_{m=0}^{M-1} \ket{m}_{\rm clock}e^{\frac{t}M m A}\ket{x_{\init}}_{\hist} + \sum_{m=M}^{M+R-1} \ket{m}_{\rm clock}e^{t A}\ket{x_{\init}}_{\hist}\right) + \ket{0^\perp} \; ,
\end{align}
within error $\epshist \ge 0$ in Euclidean norm, where $p_0 >0$
is a constant, `a' is an ancilla qubit that flags success, `clock' is a register that encodes the time step using $\log_2 (M+R)$ qubits,  $\cN_A:=\sum_{m=0}^{M-1}\|e^{\frac{t}M m A}\ket{x_{\init}}\|^2+ R \|e^{tA}\ket{x_{\init}}\|^2$ is a normalization constant, and $\ket{0^\perp}$ is a (subnormalized) state orthogonal to $\ket 0_{\rm a}$. The dimensions $M\ge 1$ and $R\ge 1$ will be chosen according to the problem parameters including precision requirements. The last $R-1$ states are often referred to as the `runway' 
where the solution at the final time $t$ is held constant. This is often done
to boost the success probability of preparing a quantum state proportional to $\vec x(t)$, since for larger $R$, a measurement of the clock register has higher probability of returning $m \ge M$.

We would like first to find a simpler expression
for $\cN_A$ to state our complexities.

\begin{lemma}[Norm of the history state]
    Let $t \ge 0$, $a>0$ be such that $\|A\|\le a$,
    and set $M = \lceil ta\rceil$.
    Then, the normalization constant satisfies
    \begin{align}
    \label{eq:normalization_factor}
        \cN_A \le \frac{e^2 M} t \int_0^t \rd s \; \norm{e^{sA}\ket{x_{\init}}}^2 + R \|e^{tA}\ket{x_{\init}}\|^2 \;.
    \end{align}
\label{lem:norm_history_state}
\end{lemma}

\begin{proof}
    Recall that $\cN_A=\sum_{m=0}^{M-1}\|e^{\frac{t}M m A}\ket{x_{\init}}\|^2+ R \|e^{tA}\ket{x_{\init}}\|^2$.
    Then, the second term in $\cN_A$ already appears in Eq.~\eqref{eq:normalization_factor}, and the goal is to bound the first term.
    To simplify the argument, let us define a rescaled matrix $\wA := At/M$ with $\norm{\wA} \le 1$. Then the first term is $\sum_{m=0}^{M-1} \norm{e^{m \wA }\ket{x_{\init}}}^2$ and we obtain
    \begin{align}
      \sum_{m=0}^{M-1}  \norm{e^{m \wA }\ket{x_{\init}}}^2 & =  \sum_{m=0}^{M-1} \int_0^1 \rd s'
         \norm{e^{m \wA }\ket{x_{\init}}}^2 \nn
         & =  \sum_{m=0}^{M-1} \int_0^1 \rd s'
         \norm{e^{-s'\wA} e^{(m+s') \wA }\ket{x_{\init}}}^2 \nn
         & \le   \sum_{m=0}^{M-1} \int_0^1 \rd s'
        \|e^{-s' \wA} \|^2 \|e^{(m+s')\wA}\ket{x_{\init}}\|^2  \nn
        & \le   \sum_{m=0}^{M-1} \int_0^1 \rd s'
        e^{2s' \|\wA\|} \|e^{(m+s')\wA}\ket{x_{\init}}\|^2 \nn
        & \le e^2   \sum_{m=0}^{M-1} \int_m^{m+1} \rd s'\|e^{s'\wA}\ket{x_{\init}}\|^2 \nn
        & = e^2 \int_0^{M} \rd s'\|e^{s'\wA}\ket{x_{\init}}\|^2 \nn
        &\le e^2 \frac{M}{t} \int_0^{t} \rd s\|e^{sA}\ket{x_{\init}}\|^2 .
    \end{align}
\end{proof}

Before providing upper bounds on the query complexity of implementing $U_{\mathrm{hist}}^{(\QLSA)}$, let us first present a useful lemma that shows that the spectral norm of a block matrix is upper bounded by the spectral norm of a compressed matrix whose entries are given by the spectral norms of the blocks of the original matrix.

\begin{lemma}[Spectral norm bound for block matrices]
    Let $\cB \in \mathbb{C}^{LN \times LN}$ be a block matrix with blocks $\cB_{mn} \in \mathbb{C}^{N \times N}$. Let $B \in \mathbb{C}^{L \times L}$ be a matrix with entries $B_{mn} = \norm{\cB_{mn}}$. Then we have the following spectral norm inequality:
    \begin{align}
        \norm{\cB} \leq \norm{B}.
    \end{align}
\label{lem:block_matrix_norm}
\end{lemma}

\begin{proof}
    Recall the following definition of the spectral norm in terms of unit vectors $\vec u,\vec v \in \mathbb{C}^{LN}$:
    \begin{align}
        \norm{\cB} = \max_{\norm{\vec u} = \norm{\vec v} = 1} | \vec u^\dagger \cB \vec v |.
    \end{align}
    Using the triangle inequality and Cauchy-Schwarz, we have that
    \begin{align}
        | \vec u^\dagger \cB \vec v | &= \left|\sum_{m,n=1}^L \vec u_m^\dagger \cB_{mn} \vec v_n \right| \nn
        &\leq \sum_{m,n=1}^L \left| \vec u_m^\dagger \cB_{mn}  \vec v_n \right| \nn
        &\leq \sum_{m,n=1}^L \norm{\vec u_m} \norm{\cB_{mn}} \norm{\vec v_n} \nn
        &\leq \underline{u}^{\top} B \underline{v},
    \end{align}
    where $\vec u_m , \ \vec v_n \in \mathbb C^N$, $\underline{u} := \lb \norm{\vec u_1}, \norm{\vec u_2}, \dots, \norm{\vec u_L}  \rb^\top \in \mathbb{C}^L$ and similarly, $\underline{v} := \lb \norm{\vec v_1}, \norm{\vec v_2}, \dots, \norm{\vec v_L}  \rb^\top \in \mathbb{C}^L$. Note that $\norm{\underline{u}} = \norm{\underline{v}} = 1$ since $\norm{u} = \norm{v} = 1$. This means that 
    \begin{align}
        \underline{u}^{\top} B \underline{v} \leq \max_{\norm{x} = \norm{y} = 1} |x^\dagger B y| = \norm{B}
    \end{align}
    and therefore, $\norm{\cB} \leq \norm{B}$.
\end{proof}

The theorem below is adapted from the main theorem in Ref.~\cite{Berry2024time-dep_history} and provides an upper bound on the query complexity of implementing a unitary $U_{\mathrm{hist}}^{(\QLSA)}$ that outputs the history state with high probability and within error $\epshist$.

\begin{theorem}[History state preparation via the linear systems approach (time-independent)]
    Let $U_A$ be a block-encoding of $A/a$, for $a \ge \|A\|>0$ and set $M=\lceil ta\rceil > 1$, $R \ge 1$.
    Then the state preparation procedure in Eq.~\eqref{eq:historystateunitary} can be implemented within error $\epshist$ in Euclidean distance using
    \begin{align}
       \cO \lb \lb a  \int_0^t \rd s \norm{e^{sA}} + R \max_{s \in [0,t]} \norm{e^{sA}} \rb \times\log (1/\epshist) \rb
    \end{align}
    queries to the initial state preparation unitary $U_{x}$ and its inverse,
    \begin{align}
       \cO \lb \lb a  \int_0^t \rd s \norm{e^{sA}} + R \max_{s \in [0,t]} \norm{e^{sA}} \rb \times \log (1/\epshist) \times \log \lb \lb a t + R \rb \max_{s \in [0,t]} \norm{e^{As}}/\epshist \rb \rb
    \end{align}
    queries to $U_A$ and its inverse, and
    \begin{align}
       \cO \lb \lb a  \int_0^t \rd s \norm{e^{sA}} + R \max_{s \in [0,t]} \norm{e^{sA}} \rb \times \log (1/\epshist) \times \log^2 \lb \lb a t + R \rb \max_{s \in [0,t]} \norm{e^{sA}}/\epshist \rb \rb
    \end{align}
    additional arbitrary two-qubit gates.
\label{thm:history_state}
\end{theorem}

\begin{proof}
    The proof follows from the proofs of Thms.\ 4.1 and 4.2 in Ref.~\cite{Berry2024time-dep_history} after some modifications. In particular, 
    starting from $\frac{\rd \vec x }{\rd t}=A\vec x(t)$, we consider a larger linear system of the form
    \begin{align}
        \mA \vec \mX =\vec \mB,
    \label{linear_system}
    \end{align}
    where $\mA \in \mathbb{C}^{LN \times L N}$ is a block matrix and $\vec \mX, \vec \mB \in \mathbb{C}^{LN}$ are block vectors. Here $L = M +R $ is the total number of time steps. The $L^2$ blocks of $\mA$ are 
    of dimension $N \times N$ and given by ($0 \le m,n \le L-1$)
    \begin{align}
    \label{eq:bigmatrixdef}
        \mA_{mn} = 
        \begin{cases}
            \one_{} & {\rm if} \ m = n \;, \\
            - V & {\rm if} \ m = n + 1 \ {\rm and} \ n \le M \;,\\
            - \one_{} & {\rm if} \ m = n + 1 \ {\rm and} \   n > M \;,\\
            {\bf 0} & \text{otherwise}\;,
        \end{cases}
    \end{align}
    where $\one_{}$ is the $N$-dimensional identity, ${\bf 0}$ is the $N$-dimensional all-zero matrix, 
    \begin{align}
        V := \sum_{k=0}^K \frac{\lb Ah \rb^k}{k!}
    \end{align}
    is a truncated Taylor series approximation to $e^{hA}$, and we set $h=t/M$. (In essence, $V$ implements a constant time evolution under $A$, allowing us to make the number of blocks $L$ as small as possible.)
    Hence, $\cA$
    is a matrix whose diagonal blocks are the identities, there are nonzero blocks below the main diagonal, and all other blocks are zero.
    Since we only consider homogeneous differential equations, the only nonzero block of $\vec \mB \in \mathbb C^{LN}$ is the first block: $\vec \mB_0 = \vec x_{\init}\in \mathbb C^{N}$. This means that we can ignore all terms in the complexity statement of Thm.\ 4.2 in Ref.~\cite{Berry2024time-dep_history} that involve the inhomogeneity (denoted $\mathbf{b}$ in that work).
    
    Then, the $L$ blocks of the solution $\vec \mX=\frac 1 {\cA}\vec \cB$ are given by
    ($0 \le m \le L-1$)
    \begin{align}
        \mX_{m} = 
        \begin{cases}
            \vec x_{\init} &{\rm if}\ m = 0 ,\\
            \tilde{\vec{x}}(mh) & {\rm if}\ 1 \le m \le M ,\\
            \tilde{\vec{x}}(t) & {\rm if}\ m > M ,\\
        \end{cases}
    \end{align}
    where $\tilde{\vec{x}}(mh) = V \tilde{\vec{x}}((m-1)h)$ is an approximation to $\vec x(mh) = e^{mhA} \vec x_{\init}$. Note that $t=Mh$.

    As an illustrative example, consider the case of three time steps and three padding steps, i.e.\ $M = R-1 = 3$. Then we have the following linear system:
    \begin{align}
    \begin{bmatrix}
         \one_{} & {\bf 0} & {\bf 0} & {\bf 0} & {\bf 0} & {\bf 0} & {\bf 0}  \\
     -V & \one_{}  & {\bf 0} & {\bf 0} & {\bf 0} & {\bf 0} & {\bf 0}  \\
    {\bf 0} & -V & \one_{}  & {\bf 0} & {\bf 0}& {\bf 0} & {\bf 0} \\
     {\bf 0} &{\bf 0} & -V & \one_{}  & {\bf 0} & {\bf 0} & {\bf 0}  \\
     {\bf 0} & {\bf 0} & {\bf 0} & -\one_{} & \one_{}  & {\bf 0} & {\bf 0}  \\
     {\bf 0} & {\bf 0} & {\bf 0} & {\bf 0} & -\one_{} & \one_{}  & {\bf 0} \\
     {\bf 0} & {\bf 0} & {\bf 0} & {\bf 0} & {\bf 0} & -\one_{} & \one_{}  \\
    \end{bmatrix}
    \begin{bmatrix}
        \tilde{\vec{x}}(0) \\ \tilde{\vec{x}}(h) \\ \tilde{\vec{x}}(2h) \\ \tilde{\vec{x}}(3h) \\ \tilde{\vec{x}}(3h) \\ \tilde{\vec{x}}(3h) \\ \tilde{\vec{x}}(3h) \\
    \end{bmatrix} =
    \begin{bmatrix}
        \vec x_{\init} \\ \vec{0} \\ \vec{0} \\ \vec{0} \\ \vec{0} \\ \vec{0} \\ \vec{0}
    \end{bmatrix}.
    \label{block_example}
    \end{align}

    Under the above assumptions, the quantum state proportional to $\frac 1 {\cA} \vec \cB$
    is then approximately the history state appearing 
    in Eq.~\eqref{eq:historystateunitary}, since
    \begin{align}
    \label{eq:approxhistorystate}
       \vec \cX \equiv  \sum_{m=0}^{M-1} \ket m_{\rm clock} \ket{ \tilde{\vec x}(mh)} + \sum_{m=M}^{M+R-1} \ket m_{\rm clock} \ket{ \tilde{\vec x}(t)} =
      \sum_{m=0}^{M-1} \ket m_{\rm clock} V^m \ket{\vec x_{\init}} + \sum_{m=M+1}^{M+R-1} \ket m_{\rm clock} V^M \ket{ \vec x_{\init}}\;,
    \end{align}
where $\ket{ \vec x_{\init}} \equiv  \vec x_{\init}$ is the initial vector
and $\ket{ \tilde{\vec x}(mh)} \equiv  \tilde{\vec x}(mh)$ represent the vectors 
that are not necessarily normalized.

In contrast with Ref.~\cite{Berry2024time-dep_history},
we note that we do not need to perform amplitude amplification to obtain the solution $\vec x(t)$ at the final time (i.e., to boost the probability that $m \ge M$) since we consider the problem of outputting the {\em entire} history state, as we need all values of time to compute the integrals appearing in the solution to the matrix equation. This means that the factor $\mathcal{R}$ in Thm.\ 4.2 of Ref.~\cite{Berry2024time-dep_history} is absent in our complexity statements.
    Also, we would like to relax the stability assumptions on the differential equation made in Ref.~\cite{Berry2024time-dep_history}. In that work, it is assumed that $A$ has non-positive log-norm, meaning that the eigenvalues of $A + A^\dagger$ are non-positive. However, it is already hinted there that this condition is not necessary for the algorithm to work. Here we fill in the details for the case where $A$ may not have non-positive log-norm, a result that is also of independent interest.

    There are basically two aspects in the proof of Thm.\ 4.2 in Ref.~\cite{Berry2024time-dep_history} that we need to modify. First, we provide a tighter bound on the condition number of $\mA$ that does not require $A$ to have  non-negative log-norm. We bound the condition number of $\mA$ by bounding $\norm{\mA}$ and $\norm{\mA^{-1}}$ separately.
    As long as the step size $h$ is chosen such that $a h \leq 1$, we have that $\norm{\mathcal{A}} \in \cO(1)$. This is true regardless of the log-norm of $A$. One way of obtaining this upper bound is via Lemma~\ref{lem:block_matrix_norm}, since all the blocks in Eq.~\eqref{eq:bigmatrixdef} have zero or bounded spectral norm.

    Next, we bound the norm of the inverse. 
    As shown in Ref.~\cite{Berry2024time-dep_history}, $\mathcal{A}^{-1}$ can be simply written as a block lower triangular matrix,
    with $N$-dimensional blocks given by ($0 \le m,n \le L-1$)
    \begin{align}
        \lb \mathcal{A}^{-1} \rb_{mn} = 
        \begin{cases}
            \one_{} & {\rm if}\ m=n ,\\
            \prod_{l=n}^{m-1} V_l & (m > n) \wedge (m \leq M + 1) \wedge (n \leq M), \\
            \prod_{l=n}^{M} V_l & (m > n) \wedge (m > M + 1) \wedge (n \leq M), \\
            \one_{} & {\rm if} \ m > n \ {\rm and} \ n > M , \\
            {\bf 0} & {\rm if} \ m < n,
        \end{cases}
    \label{inverse_matrix}
    \end{align}
    where $V_0 := \one_{}$ and $V_l = V$ for any $l \in \{1, 2, \dots, M\}$ in our case.

    Going back to our earlier example of three time steps plus three padding steps in Eq.~\eqref{block_example}, the inverse of $\cA$ is given by
    \begin{align}
        \mathcal{A}^{-1} = 
        \begin{bmatrix}
             \one & {\bf 0} & {\bf 0} & {\bf 0} & {\bf 0} & {\bf 0} & {\bf 0}  \\
         V & \one  & {\bf 0} & {\bf 0} & {\bf 0} & {\bf 0} & {\bf 0}  \\
         V^2 & V & \one  & {\bf 0} & {\bf 0} &{\bf 0} & {\bf 0}  \\
         V^3 & V^2 & V & \one & {\bf 0} & {\bf 0} & {\bf 0}  \\
         V^3 & V^2 & V & \one & \one & {\bf 0} & {\bf 0}  \\
         V^3 & V^2 & V & \one & \one & \one  & {\bf 0}  \\
         V^3 & V^2 & V & \one & \one & \one & \one 
        \end{bmatrix}.
    \end{align}

    To bound $\norm{\mathcal{A}^{-1}}$, we use Lemma~\ref{lem:block_matrix_norm} that shows that $\norm{\mathcal{A}^{-1}}$ is upper bounded by the spectral norm of the matrix $B$ with entries $B_{mn} = \norm{\lb \mA^{-1} \rb_{mn}}$. 
    To bound the spectral norm of $B$, we then use the fact that the spectral norm of a matrix is upper bounded by the square root of the product of the maximum row and the maximum column sum of that matrix. Thus,
    \begin{align}
        \norm{\mA^{-1}} \leq \|B\|\le \sqrt{\lb \max_m \sum_{n} \norm{(\mA^{-1})_{mn}} \rb \lb \max_n \sum_{m} \norm{(\mA^{-1})_{mn}} \rb}.
    \end{align}
    Per the expression of $\mA^{-1}$ given in Eq.~\eqref{inverse_matrix}, we have the following bound on the maximum row sum of $\mA^{-1}$:
    \begin{align}
        \max_m \sum_{n} \norm{\mA^{-1}_{mn}} \leq  \sum_{m=1}^M \norm{V^m} + R,
    \end{align}
    where $R\geq 1$ again denotes the number of time steps where the solution at time $t$ is held constant.
    We choose the truncation order $K$ for the Taylor series such that for any $m \in [M]$,
    \begin{align}
    \label{eq:Vexperror}
        \norm{V^m - e^{hm A}} \leq \epsilon ' \le \frac{\epshist}{M + R}.
    \end{align}
    Then, since $\epshist \le 1$,
    \begin{align}
      \max_m  \sum_{n} \norm{\mA^{-1}_{mn}} \leq 1 + \sum_{m=1}^M \norm{e^{hm A}} + R  .
    \end{align}
    By a similar argument, we obtain the following bound on the maximum column sum:
    \begin{align}
        \max_n \sum_{m} \norm{\mA^{-1}_{mn}} &\leq  \sum_{m=0}^M \norm{V^m} + (R-1) \max_{m \in \{0,1,\dots, M\}} \norm{V^m} \nn
        &\leq 1 + \sum_{m=0}^M \norm{e^{hm A}} + (R-1) \max_{s \in [0,t]} \norm{e^{sA}} \nn
        &\leq 1 + \sum_{m=1}^M \norm{e^{hm A}} + R \max_{s \in [0,t]} \norm{e^{sA}}\;.
    \end{align}
    This implies that
    \begin{align}
         \norm{\mA^{-1}} &\leq \sqrt{\lb 1 + \sum_{m=1}^M \norm{e^{hm A}} + R  \rb \lb 1 + \sum_{m=1}^M \norm{e^{hm A}} + R \max_{s \in [0,t]} \norm{e^{s A}}   \rb} \nn
         &\leq 1 + \sum_{m=1}^M \norm{e^{hm A}} + R \max_{s \in [0,t]} \norm{e^{s A}} ,
    \label{eq:inversemAbound}
    \end{align}
    where we used the fact that $\max_{s \in [0,t]} \norm{e^{s A}} \geq 1$.
    
    Let us now show how to bound the sum of matrix exponentials. The idea is exactly the same as in the proof of Lemma~\ref{lem:norm_history_state} but we still provide the full argument for completeness.
    To simplify notation, let us define a rescaled matrix $\wA := A h$ with $\norm{\wA} \le 1$. We are interested in bounding $\sum_{m=1}^M \norm{e^{m \wA }}$.
  We can essentially follow the same proof in 
  Lemma~\ref{lem:norm_history_state} to obtain
  \begin{align}
      \sum_{m=1}^M \norm{e^{ m \wA}} \le \frac eh \int_0^t \rd s \|e^{sA}\|  \;.
  \end{align}

    Putting it all together, we then obtain the following asymptotic bound on the condition number of $\mA$:
    \begin{align}
        \kappa_{\mA} \in \cO \lb a \int_0^t \rd s \|e^{sA}\| + R \max_{s \in [0,t]} \norm{e^{sA}} \rb.
    \end{align}
    In this expression we have used $M>1$ to replace $1/h$ with $\cO(a)$.
    
    The other aspect of the proof of Thm.\ 4.2 in Ref.~\cite{Berry2024time-dep_history} that needs to be modified is the bound on the truncation order of the Taylor series approximation.
    Let $\delta$ denote the error matrix consisting of the remainder of the truncated Taylor series such that 
    \begin{align}
        V = e^{hA} + \delta,
    \end{align}
    meaning that
    \begin{align}
        \delta = - \sum_{k=K+1}^{\infty} \frac{(hA)^k}{k!}.
    \end{align}
    Furthermore, assume 
     \begin{align}
     \label{eq:deltabound}
        \norm{\delta} \leq \frac{\epsilon'}{2 \max_{s \in [0,t]} \norm{e^{sA}} M},
    \end{align}
    for some $\epsilon ' \le 1$.
    Note that $\delta$ and $e^{hA}$ commute. Then the error after $m \le M$ steps is bounded by
    \begin{align}
        \norm{V^m - e^{hmA}} &= \norm{\lb e^{hA} + \delta \rb^m - e^{hmA}} \nn
        &= \norm{\sum_{k=0}^m \binom{m}{k} \delta^{m-k} e^{hkA} - e^{hmA}} \nn
        &= \norm{\sum_{k=0}^{m-1} \binom{m}{k} \delta^{m-k} e^{hkA}} \nn
        &\leq \max_{s \in [0,t]} \norm{e^{sA}} \sum_{k=0}^{m-1} \frac{m!}{k!(m-k)!} \norm{\delta}^{m-k} \nn
        &\leq \max_{s \in [0,t]} \norm{e^{sA}} \norm{\delta} m \sum_{k=0}^{m-1} \frac{(m-1)!}{k!(m - 1 - k)!} \norm{\delta}^{m - 1 - k}  \nn
        &\leq \max_{s \in [0,t]} \norm{e^{sA}} \norm{\delta} m \lb 1 + \norm{\delta} \rb^{m-1} \nn
        &\leq \max_{s \in [0,t]} \norm{e^{sA}} \norm{\delta} m e^{\norm{\delta} m} \nn
        &\le \epsilon'.
    \end{align}
   This follows from a Taylor series expansion of the Lambert-W function near $0$ under the assumption that
    \begin{align}
        \frac{\epsilon'}{\max_{s \in [0,t]} \norm{e^{sA}}} < \frac{1}{e}.
    \end{align}
    To provide a bound on $\epsilon'$ in terms of $\epshist$, we need to determine how the errors propagate to the overall history state. Note that the $\ell_2$-norm error in the unnormalized history state when using $V$ instead of $e^{Ah}$, i.e., the vector in Eq.~\eqref{eq:approxhistorystate}, is bounded by
    \begin{align}
         \sqrt{\sum_{m=0}^{M+R-1} \norm{\tilde{\vec x}(mh) - \vec x(mh)}^2} &= \sqrt{\sum_{m=0}^{M-1} \norm{V^m \vec x_{\init} - e^{hmA} \vec x_{\init}}^2 + \sum_{m=M}^{M+R-1} \norm{V^M \vec x_{\init} - e^{hMA} \vec x_{\init}}^2} \nn
         &\leq \sqrt{M+R} \norm{\vec x_{\init}} \epsilon'.
    \end{align}
    To ensure that the error between the {\em normalized} history states is at most $\epshist$, it then suffices to pick 
    \begin{align}
        \epsilon' \in \cO \lb \frac{\epshist}{\sqrt{M+R}} \rb.
    \end{align} 
    This choice will set an upper bound on $\|\delta\|$ in Eq.~\eqref{eq:deltabound} as $\|\delta\| \in \cO(\epshist/((M+R)^{3/2} \max_{s\in[0,t]}\|e^{sA}\|))$. Since $M \sim ta$,
it suffices to truncate the Taylor series for the matrix exponential at order
    \begin{align}
        K \in \cO \lb \frac{\log \lb \lb  ta + R \rb \max_{s \in [0,t]} \norm{e^{sA}}/\epshist \rb}{\log \log \lb \lb ta + R \rb \max_{s \in [0,t]} \norm{e^{sA}}/\epshist \rb} \rb,
    \end{align}
    to satisfy the error bound.

    Using the optimal linear systems solver~\cite{Costa2022optimal_linear}, we can produce a normalized quantum state that is $\cO(\epshist)$-close to the true solution  of the linear system in Eq.~\eqref{linear_system} in Euclidean distance using
    \begin{align}
        \cO \lb \kappa_{\mA} \log (1/\epshist) \rb \subseteq \cO \lb \lb a  \int_0^t \rd s \; \|e^{sA}\| + R \max_{s \in [0,t]} \norm{e^{sA}} \rb \times \log (1/\epshist) \rb
    \end{align}
    queries to the initial state preparation unitary $U_x$ and its inverse, and the block-encoding of $\mA$ and its inverse. 
    Note that Eq.~\eqref{eq:historystateunitary} is a slightly different form of the solution than given in Ref.~\cite{Costa2022optimal_linear}, which uses filtering and a nondeterministic repeat-until-success process to provide a high-accuracy form of the solution.
    In contrast, Eq.~\eqref{eq:historystateunitary} requires the history-state preparation to be fully unitary.

    The algorithm given in Ref.~\cite{Costa2022optimal_linear} may be slightly modified to give the solution as in Eq.~\eqref{eq:historystateunitary}, by performing the quantum walk and filtering just once, with the ${\rm a}$ qubit flagging success of the filtering.
    Reference \cite{Costa2022optimal_linear} describes the filtering in terms of a linear combination of unitaries (LCU) with intermediate measurements to flag failure early.
    However, it may trivially be modified to be fully coherent and give a single qubit flagging success, as required for Eq.~\eqref{eq:historystateunitary}.
    
    As discussed in Ref.~\cite{Berry2024time-dep_history}, the block-encoding cost of $\mA$ is dominated by the cost of implementing the truncated Taylor series using the block-encoding of $A$. In particular, $\mA$ can be block-encoded using $K$ queries to $U_A$, the block-encoding of $A/a$. This means the overall number of queries to $U_A$ scales like
    \begin{align}
        \cO \lb \kappa_A \log (1/\epshist) K \rb \subseteq \cO \lb \lb a \int_0^t \rd s \; \|e^{sA}\| + R \max_{s \in [0,t]} \norm{e^{sA}} \rb  \times \log (1/\epshist) \times \log \lb \lb ta + R \rb \max_{s \in [0,t]} \norm{e^{sA}}/\epshist \rb \rb.
    \end{align}

    Last, we discuss the number of additional gates required. Since we are fundamentally still using the same algorithm as in Ref.~\cite{Berry2024time-dep_history}, we can use the same argument as in that work to conclude that the overall number of additional two-qubit gates is bounded as
    \begin{align}
        \cO \lb \lb a \int_0^t \rd s \; \|e^{sA}\| + R \max_{s \in [0,t]} \norm{e^{sA}} \rb \times \log (1/\epshist) \times \log^2 \lb \lb ta + R \rb \max_{s \in [0,t]} \norm{e^{sA}}/\epshist \rb \rb.
    \end{align}
\end{proof}

\subsection{History state preparation via LCHS}

Here we analyze an alternative method for preparing history states which is not based on quantum linear system solvers but instead constructs block-encodings of controlled time evolution operators via LCHS~\cite{an2023lchs, an2023betterlchs, low2025optimal}, which are then applied to appropriately prepared initial history states.
This will give the simplified version of our result in Thm.~\ref{thm:main}.
Recent work~\cite{yang2025fastforwardingLCHS} discusses a similar method for history state preparation via LCHS, but their analysis is restricted to dissipative dynamics and their results are based on a suboptimal version of LCHS.
For simplicity, we again discuss the time-independent case first; the time-dependent case is analyzed in Appendix~\ref{app:QAlgorithm_time-dep}.

Now consider the following controlled time evolution operator:
\begin{align}
    V_{A} &:= \sum_{m=0}^{M-1} \ketbra{m}{m} \otimes e^{t \frac{m}{M} \lb A - \xi_A \one \rb} + \sum_{m=M}^{M+R-1} \ketbra{m}{m} \otimes e^{t \lb A - \xi_A \one \rb} ,
    \label{controlled_time-ev-op}
\end{align}
with $M = \lceil t \mu \rceil$ and $R \in \cO \lb \mu d/ c \rb$ as before. Additionally, $\xi_{A}$ is a known upper bound on the log-norm of $A \in \mathbb{C}^{N \times N}$, which is the largest eigenvalue of $\lb A + A^\dagger \rb/2$.

For given $m$, we can block-encode $e^{t \frac{m}{M} \lb A -\xi_{A} \one \rb}$ via LCHS~\cite{an2023lchs, an2023betterlchs, low2025optimal} since with the shift, $A -\xi_{A}$ has non-positive log norm.
For convenience, we restate the main theorem of the optimal LCHS approach below~\cite{low2025optimal}.
We will use the standard notation for block-encodings
where we say that a $(\lambda,{\rm a},\epsilon)$
block-encoding of a matrix $A$ of dimension $N=2^n$ (i.e., $n$ qubits) is a unitary $U_A$ of dimension $2^{n+{\rm a}}$ (i.e., using ${\rm a}\ge 0$ ancilla qubits) such that its first $N \times N$ block is the matrix $A/\lambda$
within additive error at most $\epsilon/\lambda$.

\begin{theorem}[Block-encoding of optimal LCHS (time-independent); reformulation of Theorem 4 in \cite{low2025optimal}]
    Let $A \in \mathbb{C}^{N \times N}$ have non-positive log-norm and assume we have access to an $(a, \rm a, 0)$-block-encoding $U_{A}$ of $A$.
    For any $t \geq 0$, $\epsilon \in (0, 4/5]$, we can construct a block-encoding $U_{\exp(tA)}$ such that
    \begin{align}
        \norm{\lb \bra{0} \otimes \one \rb  U_{\exp(tA)}  \lb \bra{0} \otimes \one \rb - e^{t A}} \leq \epsilon,
    \end{align}
    using
    \begin{align}
        \cO \lb ta \log (1/\epsilon) \rb
    \end{align}
    queries to controlled-$U_{A}$ and its inverse, plus an additional
    \begin{align}
        \cO \lb ta \log (1/\epsilon) \times \lb \mathrm{a} + \log \lb ta  + \log(1/\epsilon) \rb \rb + \log^{5/2} (1/\epsilon) \times \log \lb ta + \log (1/\epsilon) \rb \rb
    \end{align}
    primitive quantum gates.
\label{thm:LCHS}
\end{theorem}

The factor of $\log^{5/2} (1/\epsilon)$ in the gate complexity of Theorem~\ref{thm:LCHS} stems from the use of QSVT~\cite{gilyen2019qsvt} for preparing a quantum state that encodes the values of the kernel function in its amplitudes. In principle, this factor could be improved to $\log (1/\epsilon) \times \mathrm{polylog}(\log(1/\epsilon))$ by using the nested-boxes approach~\cite{Berry2024pseudopotentials}. However, since we are mostly focused on query complexities in this work, we leave a detailed analysis of such improvements for future work.

The lemma below provides an upper bound on the complexity of block-encoding $V_{A}$.

\begin{lemma}[Block-encoding of controlled time evolution operator (time-independent)]
    Let $U_{A}$ be an $(a, \rm a, 0)$-block-encoding of $A$ and let $\epsilon > 0$ be an error tolerance.
    We can construct a block-encoding $U_{V_{A}}$ of $V_A$ as given in Eq.~\eqref{controlled_time-ev-op} such that
    \begin{align}
        \norm{\lb \bra{0} \otimes \one \rb U_{V_{A}}  \lb \bra{0} \otimes \one \rb - V_{A}} \leq \epsilon,
    \end{align}
    using
    \begin{align}
        \cO \lb ta \log (1/\epsilon) \rb
    \end{align}
    queries to controlled-$U_{A}$ and its inverse, plus an additional
    \begin{align}
        \cO \lb ta \log (1/\epsilon) \times \lb \mathrm{a} + \log \lb ta  + \log(1/\epsilon) \rb + \log \lb (ta + R)/\epsilon \rb \rb + \log^{5/2} (1/\epsilon) \times \log \lb ta + \log (1/\epsilon) \rb \rb
    \end{align}
    primitive quantum gates.
\label{lem:be_ctrl_ev}
\end{lemma}

\begin{proof}
    As shown in Ref.~\cite{low2025optimal}, we can block-encode $e^{t \lb A - \xi_{A} \one \rb}$ via the following multiplexed block-encoding:
    \begin{align}
       \textsc{Mul} := \sum_{j=-R'/h'}^{R'/h'} \ketbra{j}{j} \otimes \Be \left[ \frac{h'j\bL + \bH}{R' \lambda_{\bL} + \lambda_\bH} \right],
    \end{align}
    where $\bL := \lb A + A^\dagger \rb/2 - \xi_{A} \one$, $\bH := i \lb A^\dagger - A \rb/2$, $R' \in \cO \lb \log(1/\epsilon) \rb$ and $1/h' \in \cO \lb \norm{\bL}t + \log ( 1/\epsilon) \rb$ with $R'/h'$ being an integer. 
    Here, $\Be[\cdot]$ refers to a block-encoding of the expression inside the brackets with the block-encoding constant being in the denominator. In particular, $\lambda_{\bL}, \lambda_{\bH} \in \cO \lb a \rb$ are the block-encoding constants of $\bL$ and $\bH$, respectively.
    $\textsc{Mul}$ can be implemented using $\cO(1)$ queries to controlled-$U_{A}$ and its inverse and an additional $\cO \lb \log (R'/h')\rb$ primitive gates.
    
    We can block-encode $V_{A}$ by replacing \textsc{Mul} with the following multiplexed block-encoding,
    \begin{align}
        \textsc{MulT} := U_{\mathrm{frac}} \otimes \textsc{Mul},
    \end{align}
    where $U_{\mathrm{frac}}$ is a block-encoding of
    \begin{align}
        \sum_{m=0}^{M-1} \frac{m}{M}\ketbra{m}{m} + \sum_{m=M}^{M+R-1} \ketbra{m}{m},
    \end{align}
    and otherwise following the same optimal LCHS procedure as in Ref.~\cite{low2025optimal}.
    More specifically, by querying $\textsc{MulT}$ a number of times scaling like $\cO \lb \lb R'\lambda_\bL + \lambda_\bH \rb t + \log(1/\epsilon) \rb \subseteq \cO \lb ta \log(1/\epsilon)\rb$, we can use qubitization to implement
    \begin{align}
        \textsc{SelT}' := \sum_{m=0}^{M-1} \ketbra{m}{m} \otimes \sum_{j=-R'/h'}^{R'/h'} \ketbra{j}{j} \otimes \Be \left[ U_{j,m} \right] + \sum_{m=M}^{M+R-1} \ketbra{m}{m} \otimes \sum_{j=-R'/h'}^{R'/h'} \ketbra{j}{j} \otimes \Be \left[ U_{j,M} \right] ,
    \end{align}
    where $\Be \left[ U_{j,m} \right]$ is a block-encoding of an operator $U_{j,m}$ such that $\norm{U_{j,m} - e^{-it \frac{m}{M} (h'j \bL + \bH)}} \le \epsilon$.
    $\textsc{MulT}$ can be implemented with the same cost as $\textsc{Mul}$ plus an additional $\cO \lb \log((M+R)/\epsilon) \rb \subseteq \cO \lb \log ((t a + R) /\epsilon) \rb$ primitive gates for implementing $U_{\mathrm{frac}}$ via the alternating sign trick~\cite{Berry2014sparse, Simon2024improved}. 
    Putting everything together, we therefore find that the number of queries to controlled-$U_{A}$ and its inverse scales like
    \begin{align}
        \cO \lb ta \log (1/\epsilon) \rb,
    \end{align}
    and the overall number of additional primitive gates scales like
   \begin{align}
        \cO \lb ta \log (1/\epsilon) \times \lb \mathrm{a} + \log \lb ta  + \log(1/\epsilon) \rb + \log \lb (ta + R)/\epsilon \rb \rb + \log^{5/2} (1/\epsilon) \times \log \lb ta + \log (1/\epsilon) \rb \rb.
    \end{align}
\end{proof}

Now we are ready to present the main result of this subsection which provides a bound on the complexity of preparing (sub-)normalized history states via LCHS. For simplicity, going forward, we will ignore the dependence on `$\rm a$' (the number of block-encoding ancilla qubits for $A$) in the bound on the number of additional primitive gates since that contribution is expected to be subdominant.

\begin{theorem}[History state preparation via LCHS (time-independent)]
    Let $U_A$ be a block-encoding of $A/a$, for $a \ge \|A\|>0$ and set $M=\lceil ta\rceil > 1$, $R \ge 1$. Further, let $\xi_A$ be a known upper bound on the log-norm of $A$ and let $U_{\hist}^{(\LCHS)}$ be a unitary such that
    \begin{align}
       U_{\hist}^{(\LCHS)} \ket{0}_{\mathrm{f}}\ket{0}_{\mathrm{clock}}\ket{0}_{\hist} = \ket 0_{\rm f}\frac{1}{\sqrt{\cN_{\xi_A}}}\left( \sum_{m=0}^{M-1} \ket{m}_{\rm clock}e^{\frac{t}M m A}\ket{x_\init}_{\hist} + \sum_{m=M}^{M+R-1} \ket{m}_{\rm clock}e^{t A}\ket{x_\init}_{\hist}\right) + \ket{0^\perp} \; ,
    \end{align}
    where the $\rm f$ register flags success,
    \begin{align}
        \cN_{\xi_A} = \lb \sum_{m=0}^{M-1} e^{2t \frac{m}{M} \xi_{A}} + \sum_{m=M}^{M+R-1} e^{2t \xi_{A}} \rb \in \cO \lb a \int_0^t \rd s \, e^{2 s\xi_A} + R e^{2t \xi_A} \rb.
    \end{align}
    Let $U_x$ be a state preparation unitary such that $U_x \ket{0} =\ket{x_{\init}}$.
    Then we can implement $U_{\hist}^{(\LCHS)} \ket{0}_{\mathrm{f}}\ket{0}_{\mathrm{clock}}\ket{0}_{\hist}$ within error $\epshist$ in Euclidean distance using 1 query to $U_{x}$,
    \begin{align}
        \cO \lb ta \log (1/\epshist) \rb
    \end{align}
    queries to controlled-$U_{A}$ and its inverse, plus an additional
    \begin{align}
        \widetilde{\cO} \lb (ta + R) \log^{5/2}(1/\epshist) \rb
    \end{align}
    primitive quantum gates.
\label{thm:history_state_LCHS}
\end{theorem}

\begin{proof}
    We use Lemma~\ref{lem:be_ctrl_ev} to construct a block-encoding $U_{V_{A}}$ of 
    \begin{align}
        V_{A} = \sum_{m=0}^{M-1} \ketbra{m}{m} \otimes e^{t \frac{m}{M} \lb A - \xi_A \one \rb} + \sum_{m=M}^{M+R-1} \ketbra{m}{m} \otimes e^{t \lb A -\xi_A \one \rb}.
    \end{align}
    Then we apply $U_{V_{A}}$ to the following normalized quantum state:
    \begin{align}
        \ket{\mathcal{P}_{A, x}} := \ket{0}_{V} \otimes \frac{1}{\sqrt{\cN_{\xi_A}}} \lb \sum_{m=0}^{M-1} e^{t \frac{m}{M} \xi_{A}} \ket{m} + \sum_{m=M}^{M+R-1} e^{t \xi_{A}} \ket{m} \rb \otimes \ket{x_{\init}},
    \end{align}
    where the $V$ register is the block-encoding register of $U_{V_A}$ (which can also be interpreted as the register flagging success). Using the same integral upper bound as in Lemma~\ref{lem:norm_history_state}, we have that
    \begin{align}
        \cN_{\xi_A} \leq \lb \frac{e^2 M}{t} \int_0^t \rd s \, e^{2s \xi_{A}} + R e^{t \xi_{A}} \rb \in \cO \lb a \int_0^t \rd s \, e^{2 s\xi_A} + R e^{2t \xi_A} \rb.
    \end{align}
    The state $\ket{\mathcal{P}_{A, x}}$ can be implemented within error $\epsilon_{\cP_A}$ in Euclidean distance using $1$ query to $U_{x}$ and $\cO \lb (M+R) \log \lb 1/\epsilon_{\cP_A}\rb \rb$ primitive gates~\cite{NC2010, gosset2025stateprep}.
    To ensure that the overall error is at most $\epshist$, it suffices to block-encode $U_{V_{A}}$ within error $\epshist/2$ and also prepare $\ket{\mathcal{P}_{A, x}}$ within error $\epshist/2$. By Lemma~\ref{lem:be_ctrl_ev}, we therefore require
    \begin{align}
        \cO \lb ta \log (1/\epshist) \rb
    \end{align}
    queries to controlled-$U_{A}$ and its inverse. The overall number of additional primitive gates scales like 
    \begin{align}
        &\cO \Big( ta \log (1/\epshist) \times \lb \mathrm{a} + \log \lb ta  + \log(1/\epshist) \rb + \log \lb (ta + R)/\epshist \rb \rb + \log^{5/2} (1/\epshist) \times \log \lb ta + \log (1/\epshist) \rb \nn
        &\qquad + (ta + R) \log (1/\epshist) \Big) \nn
        &\subseteq \widetilde{\cO} \lb (ta + R) \log^{5/2}(1/\epshist) \rb.
    \end{align}
\end{proof}

\subsection{Quantum complexity of estimating a matrix element in the time-independent case}

In order to estimate $\bra{\phi} \bX(t) \ket{\psi}$ in Eq.~\eqref{eq:mainsolution} within error $\epsilon$, we need to be able to estimate
\begin{align}
    \int_0^t \rd s \bra{\phi} e^{(t-s) \bA^\dagger } \bC e^{(t-s) \bB  } \ket{\psi} 
\end{align}
within error, say, $\epsilon/2$. Simply approximating the integral by a Riemann sum with the grid points determined by the time steps of the history state would not be efficient. For example, when $\bA$ and $\bB$ are antiHermitian, we would need a  number of grid points in the Riemann sum that scales with $\sim t^2/\epsilon$ to ensure that the integral discretization error is also $\cO \lb \epsilon \rb$. The complexity of implementing $U_{\mathrm{hist}}^{(\QLSA)}$ scales at least linearly with the number of time steps $M$, due to
the condition number of $\mA$ and Eq.~\eqref{eq:inversemAbound}. 
(In this case, $M \gg ta$.)
The query complexity would then have a factor $\Omega \lb t^2/\epsilon \rb$ for this Riemann sum discretization,
which is too large. Even for the LCHS approach, a simple Riemann sum discretization would lead to an undesirably
large complexity due to the normalization factors of the history states, $\cN_{\xi_\bA}$ and $\cN_{\xi_\bB}$, scaling linearly with the number of discretization points.

It is possible to provide better performance with higher-order integral approximations.
For example, Clenshaw-Curtis quadrature would yield a $\log(t/\epsilon)$ factor in the complexity rather than $t/\epsilon$.
However, we can entirely eliminate that factor in the complexity in the following way.
Consider the following decomposition of the integral:
\begin{align}
\begin{split}
    \int_0^t \rd s \;e^{(t-s)\bA^\dagger } \bC e^{(t-s)\bB  }  &= \sum_{m=0}^{M-1} \int_{t_m}^{t_{m+1}} \rd s \; e^{ (t-s)\bA^\dagger} \bC e^{(t-s)\bB }  \\
    &= \sum_{m=0}^{M-1} \int_{0}^{h} \rd\tau \; e^{ (t - t_m - \tau)\bA^\dagger} \bC e^{ (t - t_m - \tau)\bB }  \\
    &= \sum_{m=0}^{M-1}  e^{(t-t_{m+1}) \bA^\dagger } \underbrace{\lb \int_{0}^{h} \rd\tau \; e^{(h - \tau)\bA^\dagger} \bC e^{(h - \tau)\bB }  \rb}_{=: I_{\bC}} e^{(t-t_{m+1})\bB },
\end{split}
\end{align}
where $h = t/M$ is again the size of the time step, $t_m = t \frac{m}{M} = mh$, and we performed a change of variables, $s = t_m + \tau$. 
If we were able to implement $I_{\bC}$ exactly, then there would be no integral discretization error. In practice, we can approximate $I_{\bC}$ via a truncated Taylor series at low cost.
The following lemma provides an error bound on a truncated Taylor series approximation to $I_\bC$.

\begin{lemma}[Truncated Taylor series approximation to $I_{\bC}$]
    There exists   $K \in  \cO \lb \frac{\log \lb \norm{\bC}h/\epstrunc \rb}{\log \log \lb \norm{\bC}h/\epstrunc \rb} \rb$ such that
    \begin{align}
    \label{eq:widetildeIC}
        \tilde{I}_{\bC} := \sum_{p,q = 0}^K \frac{(\bA^\dagger)^p \bC \bB^q}{p! q!} \frac{h^{p+q+1}}{p+q+1},
    \end{align}
    where $h \leq \min \{\frac{1}{\norm{\bA}}, \frac{1}{\norm{\bB}} \}$, satisfies
    \begin{align}
        \norm{I_{\bC} - \tilde{I}_{\bC}} \leq \epstrunc.
    \end{align}
    \label{lem:Iccutoff}
\end{lemma}

\begin{proof}
    First, note that
    \begin{align}
         \tilde{I}_{\bC} = \int_{0}^{h} \rd\tau \; \sum_{p=0}^K \frac{\lb (h-\tau) \bA^\dagger  \rb^p}{p!} \bC \sum_{q=0}^K \frac{\lb (h-\tau)\bB  \rb^q}{q!} .
    \end{align}
    Then
    \begin{align}
    \begin{split}
        \norm{I_{\bC} - \tilde{I}_{\bC}} &= \norm{\int_{0}^{h}  \rd\tau\;e^{(h - \tau) \bA^\dagger} \bC e^{(h - \tau)\bB }    - \int_{0}^{h}  \rd\tau\;\sum_{p=0}^K \frac{\lb (h - \tau)\bA^\dagger   \rb^p}{p!} \bC \sum_{q=0}^K \frac{\lb (h - \tau)\bB   \rb^q}{q!}  } \\
        &\leq  \norm{\int_{0}^{h}  \rd\tau\; e^{(h - \tau) \bA^\dagger} \bC e^{(h - \tau)\bB }    - \int_0^h  \rd\tau\; e^{(h - \tau)\bA^\dagger } \bC \sum_{q=0}^K \frac{\lb (h - \tau)\bB  \rb^q}{q!}} \\
        &\qquad + \norm{\int_0^h  \rd\tau\; e^{(h-\tau)\bA^\dagger } \bC \sum_{q=0}^K \frac{\lb (h-\tau) \bB   \rb^q}{q!}  - \int_{0}^{h}  \rd\tau\;\sum_{p=0}^K \frac{\lb(h-\tau) \bA^\dagger  \rb^p}{p!} \bC \sum_{q=0}^K \frac{\lb (h-\tau)\bB   \rb^q}{q!}  } \\
        &\leq \int_0^h  \rd\tau\;\norm{e^{(h-\tau)\bA^\dagger }} \norm{\bC} \norm{e^{(h-\tau)\bB} - \sum_{q=0}^K \frac{\lb (h-\tau) \bB   \rb^q}{q!}}  \\
        &\qquad + \int_0^h  \rd\tau\; \norm{e^{(h-\tau)\bA^\dagger} - \sum_{p=0}^K \frac{\lb (h-\tau)\bA^\dagger   \rb^p}{p!}} \norm{\bC} \norm{\sum_{q=0}^K \frac{\lb (h-\tau)\bB   \rb^q}{q!}}.
    \end{split}
    \end{align}
    By Taylor's remainder theorem, we have that
    \begin{align}
        \norm{e^{(h-\tau)\bA^\dagger} - \sum_{p=0}^K \frac{\lb (h-\tau)\bA^\dagger   \rb^p}{p!}} \leq e^{\norm{\bA^\dagger}h }\frac{\norm{\bA^\dagger h}^{K+1}}{(K+1)!}  \leq \frac{e}{(K+1)!}.
    \end{align}
    Thus,
    \begin{align}
         \norm{I_{\bC} - \tilde{I}_{\bC}} \leq \frac{2e^2 \norm{\bC} h}{(K+1)!}.
    \end{align}
    This implies that it suffices to pick $K \in  \cO \lb \frac{\log \lb \norm{\bC}h/\epstrunc \rb}{\log \log \lb \norm{\bC}h/\epstrunc \rb} \rb$ to ensure that $\norm{I_{\bC} - \tilde{I}_{\bC}} \leq \epstrunc$.
\end{proof}

Next, we show how to block-encode $\tilde{I}_{\bC}$.
We will use the standard notation for block-encodings
where we say that an $(\lambda,{\rm a},\epsbe)$
block-encoding of a matrix $A$ of dimension $N=2^n$ (i.e., $n$ qubits) is a unitary $U_A$ of dimension $2^{n+{\rm a}}$ (i.e., using ${\rm a}\ge 0$ ancilla qubits) such that its first $N \times N$ block is the matrix $A/\lambda$
within additive error at most $\epsbe/\lambda$.

\begin{lemma}[Block-encoding of $\tilde{I}_{\bC}$]
Let $a \ge \|\bA\|$, $b \ge \|\bB\|$, and $c \ge \|\bC\|$, and set $h \le \min \{1/a,1/b\}$.
Let $U_\bA$, $U_\bB$, and $U_\bC$ be block-encodings of $\bA/a$, $\bB/b$, and $\bC/c$, respectively, which use at most ${\rm a}\ge 0$ ancilla qubits. Consider Eq.~\eqref{eq:widetildeIC}.
    We can construct a $(\lambda_{\tilde{I}_{\bC}}, {\rm a}_{\tilde{I}_{\bC}}, \epsbe)$-block-encoding of $\tilde{I}_{\bC}$, where
    \begin{align}
        \lambda_{\tilde{I}_{\bC}} &\leq c h e^2, \\
        {\rm a}_{\tilde{I}_{\bC}} &= (2K+1){\rm a} + 2 \lceil \log_2 K \rceil,
    \end{align}
  using
    \begin{align}
        2 K   \subseteq \cO \lb  \frac{\log \lb \norm{\bC}h/\epsbe \rb}{\log \log \lb \norm{\bC}h/\epsbe \rb} \rb
    \end{align}
    queries to $U_{\bA}^\dagger$ and $U_{\bB}$, one query to $U_{\bC}$, and an additional number of arbitrary primitive gates that scales as
    \begin{align}
        \cO \lb K^2 \log (\lambda_{\tilde{I}_{\bC}} /\epsbe) \rb \subseteq \cO \lb \log^3 (\lambda_{\tilde{I}_{\bC}} /\epsbe) \rb.
    \end{align}
    Here, $K$ is the cutoff in the Taylor series (Lemma~\ref{lem:Iccutoff}).
\label{lem:block-encode_IC}
\end{lemma}

\begin{proof}
    We utilize Lemma 53 in Ref.~\cite{gilyen2019qsvt} which provides complexity bounds for products of $p \ge 0$ block-encodings. In particular, we can implement a $(a^p, p{\rm a}, 0)$-block-encoding $U_{\bA^p}$ of $(\bA)^p$ simply by applying the block-encoding $U_{\bA}$ $p$ times while keeping a separate ancilla register for each query and similarly for powers of $\bB$.
    Now let us define the following unitary operation:
    \begin{align}
        \mathtt{PREP}_{\tilde{I}_{\bC}}\ket{0}\ket{0} = \frac{1}{\sqrt{\lambda_{\tilde{I}_{\bC}}}}\sum_{p,q = 0}^K \sqrt{\frac{a^p c b^q}{p! q!} \frac{h^{p+q+1}}{p+q+1}} \ket{p}\ket{q},
    \end{align}
    with
    \begin{align}
        \lambda_{\tilde{I}_{\bC}} = \sum_{p,q = 0}^K \frac{a^p c b^q}{p! q!} \frac{h^{p+q+1}}{p+q+1} \leq c h \lb \sum_{p=0}^K \frac{(a h)^p}{p!} \rb \lb \sum_{q=0}^K \frac{(b h)^q}{q!} \rb \leq c h e^2.
    \end{align}
    The state $\mathtt{PREP}_{\tilde{I}_{\bC}}\ket{0}\ket{0}$ can be prepared within error $\epsilon'$ using $\cO ( K^2 \log(1/\epsilon'))$ primitive gates~\cite{NC2010, gosset2025stateprep}. To ensure that the overall block-encoding error is at most $\epsbe$ it suffices to have $\epsilon' \leq \epsbe/(2  \lambda_{\tilde{I}_{\bC}})$.
    Additionally, define
    \begin{align}
        \mathtt{SEL}_\bA &:= \sum_{p'=0}^{K} \ketbra{p'}{p'} \otimes \one \otimes U_{\bA^{p'}} \otimes \one_{\bB} \otimes \one_{\bC} \\
        \mathtt{SEL}_\bB &:= \one \otimes \sum_{q'=0}^{K} \ketbra{q'}{q'} \otimes \one_{\bA} \otimes U_{\bB^{q'}} \otimes \one_{\bC},
    \end{align}
    which are defined on ${\rm a}_{\tilde{I}_{\bC}} + \rm{a} = (2K+2) \rm{a} + 2 \lceil \log_2 K \rceil$ qubits. In the above  expressions, $U_{\bA^{p'}}$ and $U_{\bB^{q'}}$ act on different ancilla registers, indicated by subscripts $\bA$, $\bB$ and $\bC$ on the identity matrices, but they both act on the same system register (where the powers of $\bA$ and $\bB$ are applied to) which is difficult to write down properly in tensor product notation.
    Both $\mathtt{SEL}_\bA$ and $\mathtt{SEL}_\bB$ can be implemented with $K$ queries to controlled-$U_{\bA}$ and controlled-$U_{\bB}$, respectively.
    Then
    \begin{align}
        \lb \mathtt{PREP}_{\tilde{I}_{\bC}}^\dagger \otimes \one_{\bA, \bB, \bC} \rb \cdot \mathtt{SEL}^\dagger_\bA \cdot \lb \one \otimes \one \otimes \one_{\bA} \otimes \one_{\bB} \otimes U_{\bC} \rb \cdot \mathtt{SEL}_\bB \cdot \lb \mathtt{PREP}_{\tilde{I}_{\bC}} \otimes  \one_{\bA, \bB, \bC} \rb
    \end{align}
    is a $(\lambda_{\tilde{I}_{\bC}}, {\rm a}_{\tilde{I}_{\bC}}, \epsbe)$-block-encoding of $\tilde{I}_{\bC}$. 
    Here we are again using somewhat unconventional tensor product notation to indicate that $U_\bC$ is acting on an ancilla register that is separate from the ancilla registers of $\mathtt{SEL}_\bA$ and $\mathtt{SEL}_\bB$ but it is acting on the same system register as $U_{\bA^{p'}}$ and $U_{\bB^{q'}}$.
\end{proof}

Recall that $\bra{\phi} \bX(t) \ket{\psi}$ is a sum of two terms,
\begin{align}
    \bra{\phi} \bX(t) \ket{\psi} = \bra{\phi}e^{t \bA^\dagger} \bD e^{t \bB}\ket{\psi} + \int_0^t \rd s \; \bra{\phi} e^{(t-s)\bA^\dagger} \bC e^{(t-s)\bB} \ket{\psi}.
\end{align}
We can obtain this entry by invoking the history states 
resulting from the matrices $\bA$ and $\bB$, respectively,
and computing the corresponding entry of the following matrix:
\begin{align}
    \mathcal{I}_{\bC, \bD} := \sum_{m=0}^{M} \ketbra{m}{m}_{\rm clock} \otimes I_{\bC} + \sum_{m=M+1}^{M+R} \ketbra{m}{m}_{\rm clock} \otimes \frac{\bD}{R}.
\end{align}
Such history states are proportional to 
$\sum_{m=0}^M \ket{m}_{\rm clock}e^{\frac t M m \bA}\ket{\phi}+ \sum_{m=M+1}^{M+R}\ket{m}_{\rm clock} e^{t\bA}\ket {\psi}$ and 
$\sum_{m=0}^M \ket{m}_{\rm clock}e^{\frac t M m \bB}\ket{\phi}+ \sum_{m=M+1}^{M+R}\ket{m}_{\rm clock} e^{t\bB}\ket {\psi}$, respectively.

The lemma below provides a bound on the complexity of block-encoding $\mathcal{I}_{\bC, \bD}$.

\begin{lemma}[Block-encoding $\mathcal{I}_{\bC, \bD}$]
Let $a \ge \|\bA\|$, $b \ge \|\bB\|$, $c \ge \|\bC\|$, and $d \ge \|\bD\|$, and set $h \le \min \{1/a,1/b\}$.
Let $U_\bA$, $U_\bB$, $U_\bC$, and $U_\bD$ be block-encodings of $\bA/a$, $\bB/b$, $\bC/c$, and $\bD/d$, respectively, which use at most ${\rm a} \ge 0$ ancilla qubits.
    We can implement a $(\lambda_{\mathcal{I}_{\bC, \bD}}, {\rm a}_{\mathcal{I}_{\bC, \bD}}, \epsbe')$-block-encoding of $\mathcal{I}_{\bC, \bD}$, where
     \begin{align}
        \lambda_{\mathcal{I}_{\bC, \bD}} &\le \max \left\{ c h e^2, \frac{d}{R} \right\} ,\\
        {\rm a}_{\mathcal{I}_{\bC, \bD}} & = (2K+1) {\rm a} + 2 \lceil \log_2 K \rceil +1 \;,
    \end{align}
     using 
    \begin{align}
       2K \in \cO \lb \frac{\log \lb \lambda_{\mathcal{I}_{\bC, \bD}}/\epsbe' \rb}{\log \log \lb \lambda_{\mathcal{I}_{\bC, \bD}}/\epsbe' \rb} \rb
    \end{align}
    queries to $U_{\bA}^\dagger$ and $U_{\bB}$, one query to $U_{\bC}$ and one query to $U_{\bD}$, and a number of additional primitive gates that scales as
    \begin{align}
     \cO \lb K^2  \log \lb \lambda_{\mathcal{I}_{\bC, \bD}}/\epsbe'\rb\rb   \in {\cO} \lb \log^3 \lb \lambda_{\mathcal{I}_{\bC, \bD}}/\epsbe' \rb\rb.
    \end{align}
\label{lem:block-encode_ICD}
\end{lemma}

\begin{proof}
The diagonal blocks of the matrix $\cI_{\bC,\bD}$ are $M$ repetitions of $I_\bC$ and $R$ repetitions of $\bD/R$. In our approximation, we are going to replace
$I_\bC \mapsto \tilde I_\bC$ via  Lemma~\ref{lem:block-encode_IC}, thereby defining the approximation $\tilde \cI_{\bC,\bD}$.
For a block-encoding of $\tilde \cI_{\bC,\bD}$, we are going to rescale this matrix by a multiplicative factor that is the inverse of $\lambda_{\cI_{\bC,\bD}}=\max \{\tilde \lambda_{I_\bC}, d/R\}$. Hence, $\lambda_{\cI_{\bC,\bD}} \le \max \{che^2, d/R\}$ from Lemma~\ref{lem:block-encode_IC} and we can write
\begin{align}
\tilde \cI_{\bC,\bD} =  \lambda_{\cI_{\bC,\bD}} \lb  \sum_{m=0}^{M-1} \ketbra{m}{m}_{\rm clock} \otimes \frac{\tilde I_{\bC}}{\lambda_{\tilde I_\bC}} \times \frac{\lambda_{\tilde I_\bC}}{ \lambda_{\cI_{\bC,\bD}}}+ \sum_{m=M}^{M+R-1} \ketbra{m}{m}_{\rm clock} \otimes \frac{\bD}{d} \times \frac{d}{R \lambda_{\cI_{\bC,\bD}}}  \rb\;.
\end{align}
Note that $\lambda_{\tilde I_\bC}/ \lambda_{\cI_{\bC,\bD}} \le 1$ and also $d/(R \lambda_{\cI_{\bC,\bD}})\le 1$, so the block-encoding results from a block-encoding of the term inside the brackets. That is, we replace $\tilde I_{\bC}/\lambda_{\tilde I_\bC}$ by the unitary block-encoding in Lemma~\ref{lem:block-encode_IC} and
$\bD/d$ by the unitary $U_\bD$. To obtain the correct weights in this block-encoding, we can use an additional ancilla. Without loss of generality, we can assume $d/R \ge che^2$, since we are free to increase the value of $d$
without incurring in any additional cost; that is, if $U_\bD$ encodes $\bD/d$, it also encodes $\bD/d'$ for any $d' \ge d$. This will imply $\lambda_{\cI_{\bC,\bD}}=d/R$.

Hence,
we have readily obtained a
$(\lambda_{\mathcal{I}_{\bC, \bD}}, {\rm a}_{\mathcal{I}_{\bC, \bD}}, \epsbe')$-block-encoding of $\mathcal{I}_{\bC, \bD}$, where 
${\rm a}_{\mathcal{I}_{\bC, \bD}}={\rm a}_{\tilde I_\bC}+1$ due to  the additional ancilla for the weights, and the error is
$\epsbe' \le \epsbe \lambda_{\mathcal{I}_{\bC, \bD}}/
\lambda_{\tilde I_\bC} \le \epsbe d/(R c h e^2)$, where $\epsbe$ is the error in the block-encoding of $\tilde I_\bC$ also in Lemma~\ref{lem:block-encode_IC}.

To obtain the final result, we then need to set a correct value for $K$ in the Taylor series cutoff
from Lemma~\ref{lem:Iccutoff}. For the desired error, it suffices to set, say,
\begin{align}
    \epsbe \le \frac 1 2 \epsbe' Rche^2/d\;.
\end{align}
 Then, it follows that $\epsbe \le \epsbe' Rch/d$ suffices. We also need to set $\epstrunc$
in Lemma~\ref{lem:Iccutoff} so that the additional error in the block-encoding is also bounded by, say, $\epsbe'/2$. This can be accomplished if we set
\begin{align}
    \epstrunc \le \frac 1 2 \epsbe' \frac{\lambda_{\tilde I_\bC}}{\lambda_{\cI_{\bC,\bD}}} \le \frac 1 2 \epsbe' \frac{c h e^2}{(d/R)}
\end{align}
and it suffices if $\epstrunc \le \epsbe' R c h /d$ as well. 
Since $ c \ge \|\bC\|$, the resulting $K$ from Lemma~\ref{lem:Iccutoff} is then
\begin{align}
  K \in  \cO \lb \frac{\log \lb c h/\epstrunc \rb}{\log \log \lb c h/\epstrunc \rb} \rb \in 
   \cO \lb \frac{\log \lb d/(R\epsbe') \rb}{\log \log \lb d/(R\epsbe') \rb} \rb \in 
    \cO \lb \frac{\log \lb \lambda_{\cI_{\bC,\bD}}/\epsbe' \rb}{\log \log \lb \lambda_{\cI_{\bC,\bD}}/\epsbe'\rb} \rb\;,
\end{align}
from which all results follow.
In particular, the number of primitive gates follows from Lemma~\ref{lem:block-encode_IC},
where we have $\log(\lambda_{\tilde I_\bC}/\epsbe) \in \cO(\log(\lambda_{\cI_{\bC,\bD}}/\epsbe'))$
under the assumptions.
 
\end{proof}

Note that we have assumed $d/R \ge che^2$ but since we presented the results using asymptotic scalings,
they remain valid as long as $d/R$ is lower bounded by some positive constant times $ch$.

Now we are ready to present the main theorems of this section which provides a bound on the complexity of estimating $\bra{\phi} \bX(t) \ket{\psi}$ within additive error $\epsilon$.
Our main Thm.~\ref{thm:main} will be a corollary of the following in a more succinct form 
by using simple bounds.

\begin{theorem}[Cost of estimating a matrix entry in the time-independent setting (linear systems approach)]
\label{thm:mainformal}
Let $a \ge \|\bA\|$, $b \ge \|\bB\|$, $c \ge \|\bC\|$, and $d \ge \|\bD\|$.
Let $U_\bA$, $U_\bB$, $U_\bC$, and $U_\bD$ be block-encodings of $\bA/a$, $\bB/b$, $\bC/c$, and $\bD/d$, respectively, which use at most ${\rm a} \ge 0$ ancilla qubits. Let $U_\phi$
and $U_\psi$ be the state preparation unitaries that map $\ket 0 \mapsto \ket \phi$ and $\ket 0 \mapsto \ket \psi$, respectively.
    We can estimate
    \begin{align}
        \bra{\phi} \bX(t) \ket{\psi} = \bra{\phi}e^{t \bA^\dagger} \bD e^{t \bB }\ket{\psi} + \int_0^t \rd s \;\bra{\phi} e^{(t-s)\bA^\dagger} \bC e^{(t-s)\bB } \ket{\psi} 
    \end{align}
    within additive error $\epsilon$ with probability of success at least $2/3$ using 
    \begin{align}
    \cO \lb \frac{ c \tcL_2}{\epsilon} \times \mu \tcL_1 \times \log (c \tcL_2/\epsilon) \rb
    \end{align}
    queries to $U_{\phi}$, $U_{\psi}$, and their inverses,
    \begin{align}
      \cO \lb \frac{ c \tcL_2}{\epsilon} \times \mu \tcL_1 \times \log (c \tcL_2/\epsilon) \times \log(\mu t  c^2 \tcL_2 \tcL_1/(d\epsilon))\rb
    \end{align}
    queries to $U_{\bA}$ and $U_{\bB}$, and their inverses,
    \begin{align}
        \cO \lb \frac{c \tcL_2}{\epsilon}\rb
    \end{align}
    queries to $U_{\bC}$ and $U_{\bD}$, and their inverses, and
    \begin{align}
      \widetilde{\cO} \lb \frac{c \tcL_2}{\epsilon} \times \mu \tcL_1 \times \log^2 \lb t \mu \rb \rb
    \end{align}
    additional gates, where
    \begin{align}
    \mu&:=\max\{a,b\},\\
        \tcL_1 &\geq \max \left\{ \int_0^t \rd s \; \norm{e^{s \bA }}   + \frac{d}{c} \max_{s \in [0,t]} \norm{ e^{s\bA }}, \int_0^t \rd s \; \norm{e^{\bB s}}   + \frac{d}{c} \max_{s \in [0,t]} \norm{e^{\bB s}} \right\} \; , \\
        \tcL_2 &\geq \sqrt{\lb \int_0^t \rd s \; \norm{e^{s\bA }\ket{\phi}}^2   + \frac{d}{c} \norm{e^{t \bA  }\ket{\phi}}^2 \rb \lb \int_0^t \rd s \; \norm{e^{s \bB }\ket{\psi}}^2   + \frac{d}{c} \norm{e^{t\bB }\ket{\psi}}^2 \rb}.
    \end{align}
    The $\widetilde \cO$ notation drops subdominant logarithmic factors such as $\log(c\tcL_2/\epsilon)$. Then, the overall number of queries to $U_{\phi}$, $U_{\psi}$, $U_{\bA}$, $U_{\bB}$, $U_{\bC}$, $U_{\bD}$ 
    and their inverses can be written as
    \begin{align}
        \widetilde{\cO} \lb \frac{c \tcL_2}{\epsilon} \times \mu \tcL_1 \times \log \lb t \mu \rb \rb \;.
    \end{align}
\label{thm:matrix_element}
\end{theorem}

\begin{proof}
    To estimate $\bra{\phi} \bX(t) \ket{\psi}$, we perform overlap estimation with history states.
    Let $U_{\hist, \bA}$ be such that
    \begin{align}
   U_{\hist, \bA} \ket 0 \ket 0_{\rm clock}     \ket{0}_{\hist} = \sqrt{p_\bA} \ket 0 \frac{1}{\sqrt{\mathcal{N}_\bA}} \lb    \sum_{m=0}^{M-1} \ket{m}_{\rm clock} e^{t \bA  \frac{m}{M}}\ket{\phi}_{\hist}  + \sum_{m=M}^{M + R - 1} \ket{m}_{\rm clock} e^{t \bA} \ket{\phi}_{\hist} \rb + \ket {0_\bA^\perp}
    \end{align}
    and define $U_{\hist, \bB}$ similarly, using $\bB$
    and $\ket \psi$ instead.
    Accordingly, we define
    the following two normalized states:
    \begin{align}
        \ket{\phi'_{\hist}} &:= \ket{0}_{\mathcal{I}}   \otimes \lb U_{\hist, \bA}  \ket 0 \ket{0}_{\mathrm{clock}} \ket{0}_{\hist}\rb \\
        \ket{\mathcal{I} \psi'_{\hist}} &:=  {U}_{\mathcal{I}_{\bC,\bD}}
        \ket{0}_{\mathcal{I}}   \otimes \lb U_{\hist, \bB}  \ket 0 \ket{0}_{\mathrm{clock}} \ket{0}_{\hist}\rb \;,
    \end{align}
    where ${U}_{\mathcal{I}_{\bC,\bD}}$ is a unitary  block-encoding of
    \begin{align}
        \mathcal{I}_{\bC,\bD} \otimes \ketbra{0}{0} ,
    \end{align}
    that uses the $\cI$ register of ancillas and normalization constant $\lambda_{\cI_{\bC,\bD}}$.
    This block-encoding can be constructed from the block-encoding of $\mathcal{I}_{\bC,\bD}$ in Lemma~\ref{lem:block-encode_ICD}. Note that
    \begin{align}
    \bra{\phi'_{\hist}}\ket{\cI \psi'_\hist}& =\frac 1 {\lambda_{\cI_{\bC,\bD}}} \lb \sqrt{\frac {p_\bA p_\bB}{\cN_\bA \cN_\bB}}  \sum_{m=0}^{M-1} \bra \phi e^{t \frac m M \bA^\dagger} I_\bC e^{t \frac m M \bB}\ket \psi + 
    \sum_{m=M}^{M+R-1} 
    \bra \phi e^{t \bA^\dagger} \frac {\bD}R e^{t \bB}\ket \psi
    \rb \\
    \label{eq:overlapreduction}
    & =\frac 1 {\lambda_{\cI_{\bC,\bD}}}\sqrt{\frac {p_\bA p_\bB}{\cN_\bA \cN_\bB}}
    \bra{\phi} \bX(t) \ket{\psi} \;,
    \end{align}
    and recall that $p_\bA>0$ and $p_\bB>0$ are constants.

    Unless $\bra{\phi} \bX(t) \ket{\psi}$ is guaranteed to be real, we need to run two slightly different quantum circuits to extract the real and imaginary parts of $\bra{\phi} \bX(t) \ket{\psi}$. The circuits are the same apart from an additional $S$ gate on an ancilla qubit for the estimation of the imaginary part~\cite{knill2007optimal}.
    We can use Thm.~\ref{thm:history_state} to implement the state preparation unitaries $U_{\mathrm{hist}, \bA}$ and $U_{\mathrm{hist}, \bB}$ within respective errors bounded by some $\epshist$ to be determined shortly.
     The bounded probabilities $p_\bA$ and $p_\bB$
     are unknown a priori, but these can be estimated
     within some intermediate error tolerance $\epsilon_p$ via amplitude estimation using $\cO \lb 1/\epsilon_p \rb$ queries to $U_{\hist, \bA}$ and $U_{\hist, \bB}$, and inverses, respectively.
     We will need this information to obtain
     $ \bra{\phi} \bX(t) \ket{\psi}$ from $\bra{\phi'_{\hist}}\ket{\cI \psi'_\hist}$.
    Also, we can use Lemma~\ref{lem:block-encode_ICD} to construct the block-encoding $U_{\mathcal{I}_{\bC,\bD}}$ within some error $\epsbe'$ and with block-encoding constant $\lambda_{\mathcal{I}_{\bC,\bD}}$.

    It is also possible to perform the estimation of $\bra{\phi} \bX(t) \ket{\psi}$ without estimating $p_\bA$ and $p_\bB$.
    This is because the amplitude estimation primarily requires reflections about the history state, which can be achieved more efficiently than the state preparation, and without the need to consider success probabilities.
    In particular, the LCU filter in Ref.~\cite{Costa2022optimal_linear} can be used to apply a reflection about the solution without using the quantum walk.
    Alternatively, Ref.~\cite{dalzell2024shortcut} gives a method to reflect about the solution using a quantum singular value transformation.
    The phase estimation used to perform the amplitude estimation can be performed on the history state, but that history state need only be prepared once.
    Therefore it may be performed with the nondeterministic repeat-until-success approach of Ref.~\cite{Costa2022optimal_linear}, and the probability of success does not need to be estimated.

    In order to determine the real and imaginary components of $\bra{\phi} \bX(t) \ket{\psi}$, we may apply a Hadamard-test approach and perform a reflection about
    \begin{align}\label{eq:superhist}
        \frac 1{\sqrt 2} \left( \ket{0}\ket{\phi'_{\hist}} + \ket{1}\ket{\psi'_{\hist}} \right).
    \end{align}
    This state can be given as the solution of a linear equation with blocks corresponding to the linear equations that $\ket{\phi'_{\hist}}$ and $\ket{\psi'_{\hist}}$ need to satisfy individually.
    The weighting of the blocks needs to be adjusted according to the relative values of $\mathcal{N}_\bA$ and $\mathcal{N}_\bB$.
    Those correspond to the norm of the solution of linear equations, which may be estimated with complexity $\cO(\kappa\log(1/\epsilon)/\epsilon)$ \cite{dalzell2024shortcut}.
    Given that Eq.~\eqref{eq:superhist} is the solution of a system of linear equations, the reflection may be performed via an LCU filter \cite{Costa2022optimal_linear} or QSVT \cite{dalzell2024shortcut}.

    For our problem, we will 
    pick $R \sim   \frac{\mu d}c=
    \frac{ d}{hc}$. This would imply in Lemma~\ref{lem:block-encode_ICD} that $\lambda_{\cI_{\bC,\bD}} \sim hc \sim d/R \sim c/\mu$ \footnote{This is an arbitrary choice that suffices for our results, and we note that depending on the available block-encodings, other choices might be more optimal}.
    Then our estimate of $\bra{\phi} \bX(t) \ket{\psi}$ from the overlap estimation between $\ket{\phi'_{\hist}}$ and $\ket{\mathcal{I} \psi'_{\hist}}$ will be normalized by a factor inverse to $  \sqrt{\mathcal{N}_\bA \mathcal{N}_\bB} \lambda_{\mathcal{I}_{\bC,\bD}}$ as in Eq.~\eqref{eq:overlapreduction}, which is
    \begin{align}
    \begin{split}
         &  \cO \lb \sqrt{\lb \mu \int_0^t \rd s \norm{e^{s \bA}\ket{\phi}}^2 + R\norm{e^{t \bA}\ket{\phi}}^2 \rb \lb \mu \int_0^t \rd s \norm{e^{s \bB}\ket{\psi}}^2 + R\norm{e^{t \bB}\ket{\psi}}^2 \rb}  c h \} \rb \\
         &\subseteq \cO \lb  c \sqrt{\lb \int_0^t \rd s \norm{e^{s \bA }\ket{\phi}}^2 + \frac{d}{c} \norm{e^{t \bA }\ket{\phi}}^2 \rb \lb \int_0^t \rd s \norm{e^{s \bB}\ket{\psi}}^2 + \frac{d}{c} \norm{e^{t \bB}\ket{\psi}}^2 \rb} \rb \\
         &\subseteq \cO \lb c \tcL_2 \rb.
    \end{split}
    \end{align}
    This means that we need to perform the overlap estimation between $\ket{\phi'_{\hist}}$ and $\ket{\mathcal{I} \psi'_{\hist}}$ within error  that is $\cO \lb \epsilon/\lb c \tcL_2 \rb \rb$. Additionally, to ensure that the additive error for our estimate of $\bra{\phi} \bX(t) \ket{\psi}$ is at most $\epsilon$, it will suffice if we demand that all other errors are as follows:
    \begin{align}
        \epsilon_p \in \cO \lb \frac{\epsilon}{ c \tcL_2}\rb, \; \epsilon_{\hist} \in \cO \lb \frac{\epsilon}{ c \tcL_2} \rb ,
        \end{align}
        and
        \begin{align}
       \epsbe' \in \cO \lb \frac{\epsilon }{\sqrt{\cN_\bA \cN_\bB}}\rb \in \cO \lb \frac{\epsilon }{\mu{\tcL_2}}\rb \;. 
    \end{align}
    Here, $\epsbe'$ is the error in the block-encoding of ${\cI}_{\bC,\bD}$ obtained from Lemma~\ref{lem:block-encode_ICD}.

We are ready to bound the complexities.

\begin{itemize}
\item From Thm.~\ref{thm:history_state} we obtain the query and gate complexities of preparing $U_{\hist,\bA}$ and $U_{\hist,\bB}$ with precision
$\epshist$ as above. For amplitude estimation at the Heisenberg limit,
these unitaries have to be invoked $\cO(c\tcL_2/\epsilon)$ times. Since $R \sim \mu d/c$, this readily gives a number of uses of $U_\phi$ and inverse that is 
\begin{align}
    \cO \Bigg( \frac{c \tcL_2}{\epsilon} \times \lb \mu \int_0^t \rd s \norm{e^{s \bA}}  + \frac{\mu d}{c} \max_{s \in [0,t]} \norm{e^{s \bA}} \rb \times \log (c\tcL_2/\epsilon) \Bigg)
\end{align}
and, accordingly, a number of uses of $U_\psi$ and its inverse that is
\begin{align}
    \cO \Bigg( \frac{c \tcL_2}{\epsilon}\times \lb \mu \int_0^t \rd s \norm{e^{s \bB}}  + \frac{\mu d}{c} \max_{s \in [0,t]} \norm{e^{s \bB}} \rb \times \log (c\tcL_2/\epsilon) \Bigg)
\end{align}

\item Also from Thm.~\ref{thm:history_state},
the number of queries to $U_\bA$ and its inverse contains an extra logarithmic factor $\log((at+R)\max_s \|e^{s\bA}\|/\epshist)$, which with the current choices gives the query complexity
\begin{align}
    \cO \Bigg( \frac{c \tcL_2}{\epsilon} \times \lb \mu \int_0^t \rd s \norm{e^{s \bA}}  + \frac{\mu d}{c} \max_{s \in [0,t]} \norm{e^{s \bA}} \rb \times \log (c\tcL_2/\epsilon) \times \log(\mu t \max_{s\in[0,t]}\|e^{s\bA}\|c \tcL_2/\epsilon) \Bigg)\;.
\end{align}
The number of queries to $U_\bB$ and its inverse is then
\begin{align}
    \cO \Bigg( \frac{c \tcL_2}{\epsilon} \times \lb \mu \int_0^t \rd s \norm{e^{s \bB}}  + \frac{\mu d}{c} \max_{s \in [0,t]} \norm{e^{s \bB}} \rb \times \log (c\tcL_2/\epsilon) \times \log(\mu t \max_{s\in[0,t]}\|e^{s\bB}\|c \tcL_2/\epsilon) \Bigg)\;.
\end{align}

\item Also from Thm.~\ref{thm:history_state},
the number of arbitrary two qubit gates will contain another logarithmic factor on top of the query complexity of the form 
$\log(\mu t \max_{s\in[0,t]}\{\|e^{s\bA}\|,\|e^{s\bB}\|\}c \tcL_2/\epsilon)$.
   
  \item From overlap estimation at the Heisenberg limit, within error $\cO(\epsilon/(c\tcL_2))$,
   the  number of queries to $U_{\bC}$, $U_{\bD}$ and their inverses is
\begin{align}
    \cO \lb \frac{c \tcL_2}{\epsilon} \rb\;.
\end{align}

    \end{itemize}
   
    We now simplify the expressions
    by using the definition
    \begin{align}
       \tcL_1 &:= \max \left\{ \int_0^t \rd s \; \norm{e^{s \bA }}   + \frac{d}{c} \max_{s \in [0,t]} \norm{ e^{s\bA }}, \int_0^t \rd s \; \norm{e^{\bB s}}   + \frac{d}{c} \max_{s \in [0,t]} \norm{e^{\bB s}} \right\} .
    \end{align}
    Dropping certain subdominant logarithmic factors (e.g., $\log(c\tcL_2/\epsilon)$), the overall query complexity to the state preparation unitaries and all block-encodings can be simplified to the following query upper bound:
    \begin{align}
    \label{eq:simplequerycomplexityTI}
        \widetilde{\cO} \lb \frac{c \tcL_2}{\epsilon} \times \mu \tcL_1 \times \log \lb t \mu \rb \rb.
    \end{align}
    The overall gate complexity is
    \begin{align}
        \widetilde{\cO} \lb \frac{c \tcL_2}{\epsilon} \times \mu \tcL_1 \times \log^2 \lb t \mu \rb \rb.
    \end{align}

\end{proof}

The theorem below discusses the complexity of estimating an entry of the solution matrix when the history states are prepared via LCHS.

\begin{theorem}[Cost of estimating a matrix entry in the time-independent setting (LCHS approach)]
    Let $a \ge \|\bA\|$, $b \ge \|\bB\|$, $c \ge \|\bC\|$, and $d \ge \|\bD\|$.
    Let $U_\bA$, $U_\bB$, $U_\bC$, and $U_\bD$ be block-encodings of $\bA/a$, $\bB/b$, $\bC/c$, and $\bD/d$, respectively, which use at most ${\rm a} \ge 0$ ancilla qubits. Let $U_\phi$
    and $U_\psi$ be the state preparation unitaries that map $\ket 0 \mapsto \ket \phi$ and $\ket 0 \mapsto \ket \psi$, respectively.
    Define $\mu :=\max\{a,b\}$ and
    \begin{align}
        \cL_2 := \sqrt{\lb \int_0^t \rd s\, e^{2s \xi_\bA} + \frac{d}{c} e^{2t \xi_\bA} \rb \lb \int_0^t \rd s\, e^{2s \xi_\bB} + \frac{d}{c} e^{2t \xi_\bB} \rb},
    \end{align}
    with $\xi_{\bA}$ being a known upper bound on the log-norm of $\bA$ and $\xi_{\bB}$ being a known upper bound on the log-norm of $\bB$.
    Then we can estimate
    \begin{align}
        \bra{\phi} \bX(t) \ket{\psi} = \bra{\phi}e^{t \bA^\dagger} \bD e^{t \bB }\ket{\psi} + \int_0^t \rd s \;\bra{\phi} e^{(t-s)\bA^\dagger} \bC e^{(t-s)\bB } \ket{\psi} 
    \end{align}
    within additive error $\epsilon$ with probability of success at least $2/3$ using
    \begin{align}
         \cO \lb \frac{c \cL_2}{\epsilon} \rb
    \end{align}
    queries to controlled-$U_{\phi}$, controlled-$U_{\psi}$, controlled-$U_{\bC}$, controlled-$U_{\bD}$ and their inverses, 
    \begin{align}
         \cO \lb \frac{c \cL_2}{\epsilon} \times t \mu \times \log \lb \frac{c \cL_2}{\epsilon} \rb \rb
    \label{LCHS_hist_queries}
    \end{align}
    queries to controlled-$U_{\bA}$, controlled-$U_{\bB}$ and their inverses, and an additional
    \begin{align}
        \widetilde{\cO} \lb \frac{c \cL_2}{\epsilon} \times \lb t \mu + \frac{\mu d}{c} \rb \rb
    \end{align}
    primitive gates.
\label{thm:be-approach_time-indep}
\end{theorem}

\begin{proof}
    The proof is similar to the proof of Theorem~\ref{thm:mainformal}.
    In particular, to estimate $\bra{\phi} \bX(t) \ket{\psi}$, we also perform overlap estimation with history states. However, the history states have a different normalization constant compared to the history states in Theorem~\ref{thm:mainformal}. More specifically, let $U_{\hist, \bA}^{(\LCHS)}$ be a unitary such that
    \begin{align}
        U_{\hist, \bA}^{(\LCHS)} \ket{0}_{\mathrm{f}}\ket{0}_{\mathrm{clock}}\ket{0}_{\hist} = \ket 0_{\rm f}\frac{1}{\sqrt{\cN_{\xi_\bA}}}\left( \sum_{m=0}^{M-1} \ket{m}_{\rm clock}e^{\frac{t}M m \bA}\ket{\phi}_{\hist} + \sum_{m=M}^{M+R-1} \ket{m}_{\rm clock}e^{t \bA}\ket{\phi}_{\hist}\right) + \ket{0^\perp}
    \end{align}
    with $\cN_{\xi_\bA} = \lb \sum_{m=0}^{M-1} e^{2t \frac{m}{M} \xi_{\bA}} + \sum_{m=M}^{M+R-1} e^{2t \xi_{\bA}} \rb$, and define $U_{\hist, \bB}^{(\LCHS)}$ similarly, using $\bB$
    and $\ket \psi$ instead. These history state unitaries can be constructed using Theorem~\ref{thm:history_state_LCHS}.
    Accordingly, we define
    the following two normalized states:
    \begin{align}
        \ket{\phi''_{\hist}} &:= \ket{0}_{\mathcal{I}}  \otimes \lb U_{\hist, \bA}^{(\LCHS)}  \ket{0}_\mathrm{f} \ket{0}_{\mathrm{clock}} \ket{0}_{\hist}\rb \\
        \ket{\mathcal{I} \psi''_{\hist}} &:=  {U}_{\mathcal{I}_{\bC,\bD}}
        \ket{0}_{\mathcal{I}}   \otimes \lb U_{\hist, \bB}^{(\LCHS)}  \ket 0_\mathrm{f} \ket{0}_{\mathrm{clock}} \ket{0}_{\hist}\rb \;,
    \end{align}
    where ${U}_{\mathcal{I}_{\bC,\bD}}$ is a unitary  block-encoding of
    \begin{align}
        \mathcal{I}_{\bC,\bD} \otimes \ketbra{0}{0}_f ,
    \end{align}
    that uses the $\cI$ register of ancillas and normalization constant $\lambda_{\cI_{\bC,\bD}}$.
    This block-encoding can be constructed from the block-encoding of $\mathcal{I}_{\bC,\bD}$ in Lemma~\ref{lem:block-encode_ICD}. Note that
    \begin{align}
    \bra{\phi''_{\hist}}\ket{\cI \psi''_\hist}& =\frac 1 {\lambda_{\cI_{\bC,\bD}}} \lb \frac{1}{\sqrt{\cN_{\xi_{\bA}} \cN_{\xi_{\bB}}}}  \sum_{m=0}^{M-1} \bra \phi e^{t \frac m M \bA^\dagger} I_\bC e^{t \frac m M \bB}\ket \psi + 
    \sum_{m=M}^{M+R-1} 
    \bra \phi e^{t \bA^\dagger} \frac {\bD}R e^{t \bB}\ket \psi
    \rb \\
    \label{eq:overlapreduction_LCHS}
    & = \frac{1}{\lambda_{\cI_{\bC,\bD}} \sqrt{\cN_{\xi_{\bA}} \cN_{\xi_{\bB}}}} 
    \bra{\phi} \bX(t) \ket{\psi} \;.
    \end{align}
    Furthermore,
    \begin{align}
        \lambda_{\cI_{\bC,\bD}} \sqrt{\cN_{\xi_{\bA}} \cN_{\xi_{\bB}}} &\in \cO \lb \frac{c}{\mu} \times \sqrt{\lb \mu \int_0^t \rd s \, e^{2s \xi_\bA} + \frac{\mu d}{c} e^{2t \xi_\bA} \rb \lb \mu \int_0^t \rd s \, e^{2s \xi_\bB} + \frac{\mu d}{c} e^{2t \xi_\bB} \rb} \rb \nn
         &\subseteq \cO \lb c \, \sqrt{\lb \int_0^t \rd s \, e^{2s \xi_\bA} + \frac{d}{c} e^{2t \xi_\bA} \rb \lb \int_0^t \rd s \, e^{2s \xi_\bB} + \frac{d}{c} e^{2t \xi_\bB} \rb}  \rb \nn
         &\subseteq \cO \lb c \cL_2 \rb.
    \end{align}
    
    Thus, in order to estimate $\bra{\phi} \bX(t) \ket{\psi}$ within additive error $\epsilon$, we need to perform amplitude estimation within error $\cO \lb \epsilon/\lb  c \cL_2 \rb \rb$. Additionally, we require the errors of the history states and the block-encoding of $\mathcal{I}_{\bC,\bD}$ to satisfy
    \begin{align}
        \epsilon_{\hist} \in \cO \lb \frac{\epsilon}{c \cL_2} \rb, \quad
        \epsilon'_{\cI_{\bC, \bD}}  \in \cO \lb \frac{\epsilon}{\mu \cL_2} \rb.
    \end{align}
    Therefore, we require a total of
    \begin{align}
        \cO \lb \frac{c \cL_2}{\epsilon} \rb
    \end{align}
    queries to controlled versions of $U_{\phi}$, $U_{\psi}$, $U_{\bC}$, $U_{\bD}$ and their inverses,
    \begin{align}
        \cO \lb \frac{c \cL_2}{\epsilon} \times t \mu \times \log \lb \frac{c \cL_2}{\epsilon} \rb \rb
    \end{align}
    queries to controlled versions of $U_{\bA}$ and $U_{\bB}$ and their inverses, and an additional number of primitive gates scaling like
    \begin{align}
        \tilde{\cO} \lb \frac{c \cL_2}{\epsilon} \times \lb t \mu + \frac{\mu d}{c} \rb \rb,
    \end{align}
    where the $\widetilde{O}$ notation hides subdominant logarithmic factors.
\end{proof}

In Table~\ref{tab:compare} we summarize the cost of estimating $\bra{\phi} \bX(t) \ket{\psi}$ within additive error $\epsilon$ in the time-independent setting using either quantum linear system solvers or LCHS for the history state preparation. The number of queries in the table also include the inverses.
Last, we show how Thm.~\ref{thm:main} is a corollary of Thms.~\ref{thm:mainformal} and \ref{thm:be-approach_time-indep}.

\begin{table}[t]
    \renewcommand{\arraystretch}{1.7}
    \setlength{\tabcolsep}{8pt}
    \centering
    \begin{tabular}{|c|c|c|}
    \hline
         & Linear systems approach & LCHS approach \\
         \hline
         \# queries to $U_{\phi}$ and $U_{\psi}$ & $\cO \lb \frac{ c \tcL_2}{\epsilon} \times \mu \tcL_1 \times \log (c \tcL_2/\epsilon) \rb$ & $\cO \lb \frac{c \cL_2}{\epsilon} \rb$\\
         \# queries to $U_{\bA}$ and $U_{\bB}$ & $\cO \lb \frac{ c \tcL_2}{\epsilon} \times \mu \tcL_1 \times \log (c \tcL_2/\epsilon) \times \log(\mu t  c^2 \tcL_2 \tcL_1/(d\epsilon))\rb$ & $\cO \lb \frac{c \cL_2}{\epsilon} \times t \mu \times \log \lb \frac{c \cL_2}{\epsilon} \rb \rb$ \\
         \# queries to $U_{\bC}$ and $U_{\bD}$ & $\cO \lb \frac{c \tcL_2}{\epsilon}\rb$ & $\cO \lb \frac{c \cL_2}{\epsilon} \rb$ \\
         \# additional primitive gates & $\widetilde{\cO} \lb \frac{c \tcL_2}{\epsilon} \times \mu \tcL_1 \times \log^2 \lb t \mu \rb \rb$ & $\widetilde{\cO} \lb \frac{c \cL_2}{\epsilon} \times \lb t \mu + \frac{\mu d}{c} \rb \rb$ \\
         \hline
    \end{tabular}
    \caption{Comparison of the cost for estimating $\bra{\phi} \bX(t) \ket{\psi}$ within additive error $\epsilon$ in the time-independent setting using either quantum linear system solvers or LCHS for the history state preparation; see Theorems~\ref{thm:mainformal} and \ref{thm:be-approach_time-indep} for more details.}
\label{tab:compare}
\end{table}

\begin{corollary}[Simplified complexity bound in the time-independent setting; more precise form of Theorem~\ref{thm:main}]
    Let $a \ge \|\bA\|$, $b \ge \|\bB\|$, $c \ge \|\bC\|$, and $d \ge \|\bD\|$.
    Let $U_\bA$, $U_\bB$, $U_\bC$, and $U_\bD$ be block-encodings of $\bA/a$, $\bB/b$, $\bC/c$, and $\bD/d$, respectively, which use at most ${\rm a} \ge 0$ ancilla qubits. Let $U_\phi$
    and $U_\psi$ be the state preparation unitaries that map $\ket 0 \mapsto \ket \phi$ and $\ket 0 \mapsto \ket \psi$, respectively.
    Let $\mu =\max\{a,b\}$ and let $\tcL$ be a known upper bound on
    \begin{align}
        \max_{\bY \in \{\bA, \bB\}} \int_0^t \rd s \; \norm{e^{s \bY }} + \frac{d}{c} \max_{s \in [0,t]} \norm{ e^{s\bY }}.
    \end{align}
    Furthermore, define
    \begin{align}
        \cL := \max_{\bY \in \{\bA, \bB\}} \int_0^t \rd s\, e^{2s \xi_\bY} + \frac{d}{c} e^{2t \xi_\bY},
    \end{align}
    with $\xi_{\bA}$ being a known upper bound on the log-norm of $\bA$ and $\xi_{\bB}$ being a known upper bound on the log-norm of $\bB$.
    Then we can estimate
    \begin{align}
        \bra{\phi} \bX(t) \ket{\psi} = \bra{\phi}e^{t \bA^\dagger} \bD e^{t \bB }\ket{\psi} + \int_0^t \rd s \;\bra{\phi} e^{(t-s)\bA^\dagger} \bC e^{(t-s)\bB } \ket{\psi} 
    \end{align}
    within additive error $\epsilon$ using
    \begin{align}
        \widetilde{\cO} \lb \min \left\{ \frac{c \cL}{\epsilon} \times t\mu, \; \frac{c}{\epsilon} \times \mu \tcL \times \tcL \max_{\substack{\bY \in \{ \bA, \bB \} s \in [0,t]}} \norm{e^{s \bY}} \times \log(t \mu) \right\} \rb
    \label{simplified}
    \end{align}
    queries to $U_{\phi}$, $U_{\psi}$, $U_{\bA}$, $U_{\bB}$, $U_{\bC}$, $U_{\bD}$ 
    and their inverses.
    The number of additional primitive gates is larger than the above query complexity by at most polylogarithmic factors.
\end{corollary}

\begin{proof}
The proof of the above corollary follows fairly straightforwardly from Theorems~\ref{thm:mainformal} and \ref{thm:be-approach_time-indep}.
The first term inside the $\min$ in Eq.~\eqref{simplified} comes from Eq.~\eqref{LCHS_hist_queries} after upper bounding $\cL_2$ by $\cL$ and dropping subdominant logarithmic factors.

For the second term in Eq.~\eqref{simplified}, consider for example the term
\begin{align}
    \int_0^t \rd s \; \norm{e^{s\bA }\ket{\phi}}^2   + \frac{d}{c} \norm{e^{t \bA  }\ket{\phi}}^2
\end{align}
in the definition of $\tcL_2$. This can be upper bounded as follows:
\begin{align}
    \int_0^t \rd s \; \norm{e^{s\bA }\ket{\phi}}^2   + \frac{d}{c} \norm{e^{t \bA  }\ket{\phi}}^2 &\le \lb \int_0^t \rd s \; \norm{e^{s\bA }\ket{\phi}}   + \frac{d}{c} \norm{e^{t \bA  }\ket{\phi}}\rb \max_{s\in[0,t]}\|e^{s\bA}\|\\
    & \le \tcL_1 \max_{s\in[0,t]}\|e^{s\bA}\| ,
\end{align}
and hence
\begin{align}
    \tcL_2 \le \tcL_1 \max_{s, s' \in [0,t] } \sqrt{\norm{e^{s\bA}} \norm{e^{s' \bB}}} \leq \tcL_1 \max_{\substack{\bY \in \{ \bA, \bB \} \\ s \in [0,t]}} \norm{e^{s \bY}} \;.
\end{align}
Using this bound in Eq.~\eqref{eq:simplequerycomplexityTI} and replacing $\tcL_1 \rightarrow \tcL$ gives the second term inside the $\min$ in Eq.~\eqref{simplified}. 

\end{proof}

\section{Quantum algorithm for time-dependent case}
\label{app:QAlgorithm_time-dep}

Consider the general case where all $N$-dimensional matrices $\bA$, $\bB$ and $\bC$ are time dependent, i.e.
\begin{align}
    \frac{\rd}{\rd t} \bX(t) = \bA^\dagger(t) \bX(t) + \bX(t) \bB(t) + \bC(t), \quad \bX(0) = \bD.
\label{time-dep_diffeq}
\end{align}
Assume that we have access to block-encodings $U_{\bA}$, $U_{\bB}$, $U_{\bC}$, $U_{\bD}$, and their inverses, such that
\begin{align}
    \lb \bra{0}  \otimes \one \rb  U_{\bA} \lb \ket{0}  \otimes \one \rb &= \sum_{j=0}^{J-1} \ketbra{j}{j} \otimes \frac{\bA(\tau_j)}{a} , \\
    \lb \bra{0}  \otimes \one \rb  U_{\bB} \lb \ket{0}  \otimes \one \rb &= \sum_{j=0}^{J-1} \ketbra{j}{j} \otimes \frac{\bB(\tau_j)}{b} , \\
    \lb \bra{0}  \otimes \one \rb  U_{\bC} \lb \ket{0}  \otimes \one \rb &= \sum_{j=0}^{J-1} \ketbra{j}{j} \otimes \frac{\bC(\tau_j)}{c} , \\
    \lb \bra{0}  \otimes \one \rb  U_{\bD} \lb \ket{0}  \otimes \one \rb &= \one  \otimes \frac{\bD}{d},
\end{align}
where $\tau_j$ is a discretization of time, and $J$ scales polynomially with $t$ and $\epsilon$ and also depends on the maximum derivative of $\bA(t)$, $\bB(t)$ and $\bC(t)$. 
Additionally, $\tau_0 = 0$ and $\tau_{J-1} = \frac{J-1}{J} t$.
The state $\ket 0$ is some zero state of an ancillary system for the block-encodings, using   ${\rm a}\ge 0$ ancilla qubits, and $\one$
is the identity whose dimension is clear from context.

The goal is to estimate $\bra{\phi} \bX(t) \ket{\psi}$ within additive error $\epsilon$. 
Let $W_{\bA}(t,s)$ and $W_{\bB}(t,s)$ be the transition matrices under $\bA$ and $\bB$, respectively, from time $s$ to time $t$, that satisfy
\begin{align}
\label{eq:WAmotion}
    \frac{\rd W_{\bA}(t,s)}{\rd t} &= W_{\bA}(t,s) \bA(t) , \quad W_{\bA}(s,s) = \one \quad \forall \, s \in [0,t], \\
    \label{eq:WBmotion}
     \frac{\rd W_{\bB}(t,s)}{\rd t} &= W_{\bB}(t,s) \bB(t) , \quad W_{\bA}(s,s) = \one \quad \forall \, s \in [0,t],,
\end{align}
and can be written as
\begin{align}
\label{eq:DysonWA}
W_{\bA}(t,s) & = \tilde \cT e^{\int_s^t \rd \tau \bA (\tau)}     = 
\sum_{k=0}^\infty \int_{s}^t \rd t_1 \int_{s}^{t_1} \rd t_2 \cdots \int_{s}^{t_{k-1}} \rd t_k \bA (t_k) \cdots \bA (t_2) \bA (t_1)   \; , \\
\label{eq:DysonWB}
W_{\bB}(t,s) & =\tilde \cT e^{\int_s^t \rd \tau \bB (\tau)}     = 
\sum_{k=0}^\infty \int_{s}^t \rd t_1 \int_{s}^{t_1} \rd t_2 \cdots \int_{s}^{t_{k-1}} \rd t_k \bB (t_k) \cdots \bB (t_2) \bB (t_1) \;.
\end{align}
($\tilde \cT$ is anti-time ordering.)
They also obey the following composition rules:
\begin{align}
\label{eq:WAcomposition}
   W_{\bA}(\tau, s) W_{\bA}(t,\tau)  &= W_{\bA}(t,s), \\
   \label{eq:WBcomposition}
    W_{\bB}(\tau, s) W_{\bB}(t,\tau)  &= W_{\bB}(t,s).
\end{align}
Using these matrices, 
the solution to Eq.~\eqref{time-dep_diffeq} is
\begin{align}
 \bX(t) = (W_{\bA}(t,0))^\dagger \bD W_{\bB} (t,0)  + \int_{0}^t \rd s \, (W_{\bA}(t,s))^\dagger \bC(s) W_{\bB} (t,s)  \;.
\end{align}
This solution follows from defining $\bY(t)$ such that $\bX(t)=(W_\bA(t,0))^\dagger \bY(t) W_\bB(t,0)$. Taking the derivative in both sides implies $(W_\bA(t,0))^\dagger \dot \bY(t) W_\bB(t,0)  = \bC(t)$ and hence
$\bY(t)=\bD + \int_0^t \rd s (W_\bA(0,s))^\dagger \bC(s) W_\bB(0,s)$, or $\bX(t) = (W_{\bA}(t,0))^\dagger \bD W_{\bB} (t,0)  + \int_{0}^t \rd s \, (W_{\bA}(t,s))^\dagger \bC(s) W_{\bB} (t,s)$.
Then the matrix element is
\begin{align}
    \bra{\phi} \bX(t) \ket{\psi} = \bra{\phi}(W_{\bA}(t,0))^\dagger \bD W_{\bB} (t,0) \ket{\psi} + \int_{0}^t \rd s \, \bra{\phi} (W_{\bA}(t,s))^\dagger \bC(s) W_{\bB} (t,s) \ket{\psi}.
\end{align}

Similar to the time-independent case, we split the integral into smaller pieces to make it more suitable for a time-dependent version of the linear systems approach. In particular,
\begin{align}
\begin{split}
    &\int_0^t \rd s \,  (W_{\bA}(t,s))^\dagger \bC(s) W_{\bB} (t,s)    =\\
    &= \sum_{m=0}^{M-1}  (W_{\bA}(t,t_{m+1}) )^\dagger\underbrace{\lb \int_0^h \rd\tau \, (W_{\bA}(t_{m+1},t_m + \tau))^\dagger \bC(t_m + \tau)  W_{\bB} (t_{m+1},t_m + \tau)  \rb}_{=: I_{\bC}(t_m)}  W_{\bB} (t,t_{m+1})   ,
\end{split}
\end{align}
where $M \sim \max \{ a, b\} t$ is again the number of time steps, $h = t/M$ is the step size and $t_m = t \frac{m}{M} = mh$.

Let us define the following two normalized, reverse-order history states:
\begin{align}
    \ket{\wphi_\hist} &:= \frac{1}{\sqrt{\mathcal{N}_{\bA}}} \lb \sum_{m=0}^{R-1} \ket{m}_{\rm clock} W_{\bA} (t,0) \ket{\phi} + \sum_{m=R}^{R+M-1} \ket{m}_{\rm clock} W_{\bA} (t,t_{m+1-R}) \ket{\phi} \rb \label{phi_hist}
    \\
    \ket{\wpsi_\hist} &:= \frac{1}{\sqrt{\mathcal{N}_{\bB}}} \lb \sum_{m=0}^{R-1} \ket{m}_{\rm clock} W_{\bB} (t,0) \ket{\psi} + \sum_{m=R}^{R+M-1} \ket{m}_{\rm clock} W_{\bB} (t,t_{m+1-R}) \ket{\psi} \rb.
\end{align}
Note that in contrast to the standard history state setting, here the first $R$ time steps encode the solution at the final time $t$ and then the remaining steps include the overall amount of time evolution that is being applied to $\ket{\phi}$ or $\ket{\psi}$ with each increasing step. Reverse-order history states can be easily obtained from standard-order history states by first producing a history state in the standard ordering and then swapping basis states, e.g., applying a single qubit Pauli $X$ gate to all qubits in the clock register $\ket{m}_{\rm clock}$.
We will be then interested in the unitaries that prepare the history states as
\begin{align}
\label{eq:historystateunitaryTD}
    U_{\hist,\bA}\ket 0_{\rm a}\ket{0}_{\rm clock}\ket 0 &\mapsto \sqrt{p_\bA} \ket0_{\rm a} \ket{\wphi_\hist} + \ket{0_\bA^\perp}, \\
     U_{\hist,\bB}\ket 0_{\rm a}\ket{0}_{\rm clock}\ket 0 &\mapsto \sqrt{p_\bB} \ket0_{\rm a} \ket{\wpsi_\hist} + \ket{0_\bB^\perp}
\end{align}
within error $\epshist \ge 0$ in Euclidean norm, where $p_\bA>0$ and $p_\bB$ are constants, `a' is an ancilla qubit that flags success, `clock' is the clock register.
The states $\ket{0_\bA^\perp}$ and $\ket{0_\bB^\perp}$
are orthogonal to $\ket 0_{\rm a}$.

The lemma below provides a bound on the norm of $\ket{\wphi_{\hist}}$. It is analogous to Lemma~\ref{lem:norm_history_state}.

\begin{lemma}[Norm of history states in the time-dependent setting]
    Let
    \begin{align}
        \mathcal{N}_\bA =  R \norm{W_{\bA} (t,0) \ket{\phi}}^2 + \sum_{m=R}^{R+M-1} \norm{W_{\bA} (t,t_{m+1-R}) \ket{\phi}}^2
    \label{normalization_factor_time-dep}
    \end{align}
    denote the squared normalization factor of the history state $\ket{\wphi_{\hist}}$ defined in Eq.~\eqref{phi_hist}. Then 
    \begin{align}
        \mathcal{N}_\bA \le R  \norm{W_{\bA} (t,0)\ket{\phi}}^2 + a \int_0^t \rd s \,\norm{W_{\bA} (t,s)\ket{\phi}}^2 .
    \end{align}
\label{lem:norm_history_state_time-dep}
\end{lemma}

\begin{proof}
    The proof is similar to the proof of Lemma~\ref{lem:norm_history_state}.
    The first term in the bound on $\mathcal{N_{\bA}}$ can be read off directly from Eq.~\eqref{normalization_factor_time-dep}. Let us therefore focus on the second term.
    To simplify the argument, let us define a rescaled matrix $\wbA := \bA h = \bA/a$ with $\norm{\wbA} \leq 1$. Then we need to bound the following sum:
    \begin{align}
        \sum_{m=R}^{R+M-1} \norm{W_{\wbA}(M,m+1-R) \ket{\phi}}^2 = \sum_{m=1}^{M} \norm{W_{\wbA}(M,m) \ket{\phi}}^2.
    \end{align}
    Here, the transition matrix $W_{\wbA}$ is such that $W_{\wbA}(m',m)=W_\bA(t_{m'},t_m)$, and $t_m=m/a$.
    Let $\tau_m \in [m-1, m]$. From Eq.~\eqref{eq:WAcomposition} we have
    \begin{align}
    \begin{split}
       \norm{W_{\wbA} (M,m) \ket{\phi}} &= \norm{W_{\wbA} (\tau_m,m) W_{\wbA} (M,\tau_m) \ket{\phi}} \\
    \end{split}
    \end{align}
    and hence
    \begin{align}
         \norm{W_{\wbA} (M,m) \ket{\phi}}^2 \leq e^{2\norm{\wbA }} \norm{W_{\wbA} (M,\tau_m) \ket{\phi}}^2 \leq e^2 \norm{W_{\wbA} (M,\tau_m) \ket{\phi}}^2.
    \end{align}
    Integrating both sides from $m-1$ to $m$, we observe that
    \begin{align}
        \norm{W_{\wbA} (M,m) \ket{\phi}}^2 = \int_{m-1}^{m} \rd\tau_m \, \norm{W_{\wbA} (M,m) \ket{\phi}}^2 \leq e^2 \int_{m-1}^{m} \rd\tau_m \, \norm{W_{\wbA} (M,\tau_m) \ket{\phi}}^2.
    \end{align}
    Now we sum over $m$ on both sides:
    \begin{align}
    \begin{split}
         \sum_{m=1}^{M}\norm{W_{\wbA} (M,m) \ket{\phi}}^2 &\leq \sum_{m=1}^{M} e^2 \int_{m-1}^{m} \rd\tau_m \, \norm{W_{\wbA} (M,\tau_m) \ket{\phi}}^2 \\
         &\leq e^2 \int_{0}^{M} \rd\tau \, \norm{W_{\wbA} (M,\tau) \ket{\phi}}^2 \\
         &\leq  \frac{e^2 M}{t} \int_{0}^{t} \rd s \, \norm{W_{\bA}(t,s) \ket{\phi}}^2  .
    \end{split}
    \end{align}
\end{proof}
The same type of bound can be obtained for the norm of $\ket{\wpsi_{\hist}}$ by replacing $\bA$ with $\bB$ and $\ket{\phi}$ with $\ket{\psi}$ in Lemma~\ref{lem:norm_history_state_time-dep}.

Analogous to the time-independent case, let us also define
\begin{align}
    \mathcal{I}_{\bC, \bD} := \sum_{m=0}^{R-1} \ketbra{m}{m} \otimes \frac{\bD}{R} + \sum_{m=R}^{R+M-1} \ketbra{m}{m} \otimes I_{\bC}(t_{m-R}).
\end{align}
Then we can express the matrix element $\bra{\phi} \bX(t) \ket{\psi}$ as follows:
\begin{align}
    \bra{\phi} \bX(t) \ket{\psi} = \sqrt{\mathcal{N}_\bA \mathcal{N}_\bB} \; \bra{\wphi_{\hist}} \mathcal{I}_{\bC, \bD}\ket{\wpsi_\hist}.
\end{align}
We then reduced the problem to one of overlap estimation: we can obtain $\bra{\phi} \bX(t) \ket{\psi}$ by preparing the history states $\ket{\wphi_{\hist}}$ and $\ket{\wpsi_{\hist}}$, applying a block-encoding of $\mathcal{I}_{\bC, \bD}$ to, say, $\ket{\wpsi_{\hist}}$, and then estimating the inner product between the resulting two states via a Heisenberg limited approach~\cite{knill2007optimal}.

\subsection{History state preparation via quantum linear system solvers in the time-dependent case}

Let us now show how to produce the history state $\ket{\wphi_{\hist}}$ (the same argument applies also to $\ket{\wpsi_{\hist}}$). We first prepare a standard-order version of $\ket{\wphi_{\hist}}$ and then reverse the order of the clock states. The algorithm for producing the standard-order version of $\ket{\wphi_{\hist}}$ is basically the same as in Ref.~\cite{Berry2024time-dep_history} but with a slightly modified linear system. In our case, the final time $t$ of $W_{\bA}(t,s)$ (and $W_{\bB}(t,s)$) is fixed but the initial time $s$ decreases with increasing time steps, so the last time step still corresponds to the largest amount of time evolution ($s=0$). In the standard setting, however, the initial time of the evolution operator is fixed (usually set to $s=0$) and then the final time $t$ keeps getting extended throughout the time steps.

\begin{theorem}[Complexity of history state preparation via the linear systems approach (time-dependent)]
Let $U_\bA$ be a block-encoding of $\sum_{j=0}^{J-1}\ketbra j \otimes \frac{\bA(\tau_j)}a$ for $a \ge \|\bA(t)\|>0$ and set $M= \lceil ta \rceil$, $R\ge 1$.
    Then the state preparation procedure in Eq.~\eqref{eq:historystateunitaryTD}  can be implemented within error $\epshist$ in Euclidean distance using
     \begin{align}
       \cO \lb \lb a \norm{W_{\bA}}_{L^1} + R \max_{s \in [0,t]} \norm{W_{\bA}(s,0)} \rb \times \log (1/\epshist) \rb
    \end{align}
    queries to the initial state preparation unitary $U_{\phi}$ and its inverse,
    \begin{align}
       \cO \lb \lb a \norm{W_{\bA} }_{L^1} + R \max_{s \in [0,t]} \norm{W_{\bA} (s,0)} \rb \times \log (1/\epshist) \times \log \lb \lb a t + R \rb \max_{\substack{s' \geq s \\ s,s' \in [0,t]} } \norm{W_{\bA} (s',s)}/\epshist \rb \rb
    \end{align}
    queries to $U_{\bA}$ and its inverse, and
    \begin{align}
    \begin{split}
        &\cO \Bigg( \lb a \norm{W_{\bA} }_{L^1} + R \max_{s \in [0,t]} \norm{W_{\bA} (s,0)} \rb \times \log (1/\epshist) \times \log \lb \lb a t + R \rb \max_{\substack{s' \geq s \\ s,s' \in [0,t]} }\norm{W_{\bA} (s',s)}/\epshist \rb \\
        &\quad \times \lb \log \lb \frac{t \max_{s \in [0,t]} \norm{\bA'(s)}}{a \epshist} \rb + \log \lb \lb a t + R \rb \max_{\substack{s' \geq s \\ s,s' \in [0,t]} } \norm{W_{\bA}(s',s)}/\epshist \rb \rb \Bigg)
    \end{split}
    \end{align}
    additional gates, where $\bA'(s)=\frac{\rd}{\rd s}\bA(s)$ and
    \begin{align}
        \norm{W_{\bA}}_{L^1} := \max \left\{ \max_{t_1 \in [0,t]}\int_{0}^{t_1} \rd s \, \norm{W_{\bA} (t_1,s)}, \max_{t_0 \in [0,t]} \int_{t_0}^t \rd s \, \norm{W_{\bA} (s, t_0)} \right\}.
    \end{align}
\label{thm:history_state_time-dep}
\end{theorem}

\begin{proof}
    The proof is similar to the proof of Thm.~\ref{thm:history_state}. 
    First, note that the standard-ordering form of $\ket{\wphi_\hist}$ is given by
    \begin{align}
        \frac{1}{\sqrt{\mathcal{N}_{\bA}}} \lb \sum_{m=0}^{M-1} \ket{m} W_{\bA} (t,t_{M-m}) \ket{\phi} + \sum_{m=M}^{M+R-1} \ket{m} W_{\bA} (t,0) \ket{\phi} \rb =:\ket{\wphi_{\hist, \mathrm{st}}}.
    \end{align}
    Let
    \begin{align}
        V_m := \sum_{k=0}^{K} \int_{t_{M-m}}^{t_{M-m+1}} \rd\tau_1 \int_{t_{M-m}}^{\tau_1} \rd\tau_2 \cdots \int_{t_{M-m}}^{\tau_{k-1}} \rd\tau_k \bA (\tau_k) \cdots \bA (\tau_2) \bA (\tau_1),
    \end{align}
    for $m \in \{1,2,\dots, M\}$, be the truncated Dyson series approximation to $W_{\bA}(t_{M-m+1},t_{M-m})$ in Eq.~\eqref{eq:DysonWA} of order $K<\infty$. 
    With the above definition of $V_m$, we have the same type of linear system as in Ref.~\cite{Berry2024time-dep_history}. In particular, we consider a block matrix $\mA \in \mathbb C^{LN \times LN}$, $L=M+R$, with the following $L^2$ blocks ($0 \le m,n \le L-1$):
    \begin{align}
        \mA_{mn} = 
        \begin{cases}
            \one_{} & {\rm if} \ m = n \;, \\
            - V_n & {\rm if} \ m = n + 1 \ {\rm and} \ n \le M \;,\\
            - \one_{} & {\rm if} \ m = n + 1 \ {\rm and} \ n > M \;,\\
            {\bf 0} & \text{otherwise}\;,
        \end{cases}
    \end{align}
    which is similar to the block matrix as in the time-independent case, except that the $V_n$'s are not all the same anymore. Here $\one$ is the $N$-dimensional identity and $\bf 0$ the $N-$dimensional all zero matrix. We will then consider the system of linear equations $\cA \vec \cX = \vec \cB$, where $\cB \in \mathbb C^{LN}$ contains $L$ blocks $\cB_m \in \mathbb C^N$, and the only nonzero block is the first one: $\vec \cB_0=\vec x_{\rm in}\equiv \ket \phi$ in this case. The solution vector $\vec \cX$ is also made of $L$ blocks $\vec \cX_m$ of dimension $N$, and each can be shown to be $\vec \cX_0 = \vec x_{\rm in}$, $\vec \cX_m = \tilde {\vec x}(mh)$, for $1 \le m \le M$, and $\vec \cX_m =  \tilde {\vec x}(Mh)$ for $m>M$, where $\tilde{\vec x}(mh)=V_m \tilde{\vec x}((m-1)h)$.

    As an illustrative example, consider the case of three time steps and three padding steps, i.e. $M = R = 3$. Then we have the following linear system:
    \begin{align}
    \begin{bmatrix}
         \one_{} & {\bf 0} & {\bf 0} & {\bf 0} & {\bf 0} & {\bf 0} & {\bf 0}  \\
     -V_1 & \one_{}  & {\bf 0} & {\bf 0} & {\bf 0} & {\bf 0} & {\bf 0}  \\
    {\bf 0} & -V_2 & \one_{}  & {\bf 0} & {\bf 0}& {\bf 0} & {\bf 0} \\
     {\bf 0} &{\bf 0} & -V_3 & \one_{}  & {\bf 0} & {\bf 0} & {\bf 0}  \\
     {\bf 0} & {\bf 0} & {\bf 0} & -\one_{} & \one_{}  & {\bf 0} & {\bf 0}  \\
     {\bf 0} & {\bf 0} & {\bf 0} & {\bf 0} & -\one_{} & \one_{}  & {\bf 0} \\
     {\bf 0} & {\bf 0} & {\bf 0} & {\bf 0} & {\bf 0} & -\one_{} & \one_{}  \\
    \end{bmatrix}
    \begin{bmatrix}
        \tilde{\vec{x}}(0) \\ \tilde{\vec{x}}(h) \\ \tilde{\vec{x}}(2h) \\ \tilde{\vec{x}}(3h) \\ \tilde{\vec{x}}(3h) \\ \tilde{\vec{x}}(3h) \\ \tilde{\vec{x}}(3h) \\
    \end{bmatrix} =
    \begin{bmatrix}
        \vec x_{\init} \\ \vec{0} \\ \vec{0} \\ \vec{0} \\ \vec{0} \\ \vec{0} \\ \vec{0}
    \end{bmatrix}.
    \label{block_example_time-dep}
    \end{align}

    Let us now bound the condition number of $\mA$.
    As long as we choose the step size $h \sim 1/a$, we have that $\norm{\mA} \in \cO(1)$. 
    The $L^2$ blocks  of the inverse of $\mA$ can be shown to be ($0 \le n,m \le L-1)$
    \begin{align}
        \lb \mathcal{A}^{-1} \rb_{mn} = 
        \begin{cases}
            \one, & m=n \\
            \prod_{l=n}^{m-1} V_l, & (m > n) \wedge (m \leq M + 1) \wedge (n \leq M) \\
            \prod_{l=n}^{M} V_l, & (m > n) \wedge (m > M + 1) \wedge (n \leq M) \\
            \one, & (m > n) \wedge (n > M), \\
            \bf 0, & m < n.
        \end{cases}
    \label{inverse_matrix_time-dep}
    \end{align}

    Going back to our earlier example of three time steps plus three padding steps in Eq.~\eqref{block_example_time-dep}, the inverse of $\cA$ is given by
    \begin{align}
        \mathcal{A}^{-1} = 
        \begin{bmatrix}
             \one & {\bf 0} & {\bf 0} & {\bf 0} & {\bf 0} & {\bf 0} & {\bf 0}  \\
         V_1 & \one  & {\bf 0} & {\bf 0} & {\bf 0} & {\bf 0} & {\bf 0}  \\
         V_2 V_1 & V_2 & \one  & {\bf 0} & {\bf 0} & {\bf 0} & {\bf 0}  \\
         V_3 V_2 V_1 & V_3 V_2 & V_3 & \one & {\bf 0} & {\bf 0} & {\bf 0}  \\
         V_3 V_2 V_1 & V_3 V_2 & V_3 & \one & \one & {\bf 0} & {\bf 0} \\
         V_3 V_2 V_1 & V_3 V_2 & V_3 & \one & \one & \one  & {\bf 0}  \\
         V_3 V_2 V_1 & V_3 V_2 & V_3 & \one & \one & \one & \one 
        \end{bmatrix}.
    \end{align}

    To bound $\norm{\mA^{-1}}$, we use Lemma~\ref{lem:block_matrix_norm} and the fact that the spectral norm of a matrix is upper bounded by the square root of the product of the maximum row and the maximum column sum of that matrix. Let $B$ be the matrix of entries $B_{mn}=\|(\cA^{-1})_{mn}\|$
    Thus,
    \begin{align}
        \norm{\mA^{-1}} \le \|B\| \leq \sqrt{\lb \max_m \sum_{n} \norm{(\mA^{-1})_{mn}} \rb \lb \max_n \sum_{m} \norm{(\mA^{-1})_{mn}} \rb}.
    \end{align}
    Per the expression of $\mA^{-1}$ given in Eq.~\eqref{inverse_matrix_time-dep}, we have the following bound on the maximum row sum of $\mA^{-1}$:
    \begin{align}
        \max_m \sum_{n} \norm{(\mA^{-1})_{mn}} \leq \max_{m \in \{1,2,\dots, M\}}  \lb \|V_m \ldots V_1 \|+ \|V_m \ldots V_2\| + \ldots + \|V_m\|\rb  + R ,
    \end{align}
    where $R$  denotes the number of additional time steps where the solution is held constant.
    We choose the truncation order $K$ for the Dyson series such that for any $m_2 \geq m_1 \in \{1,2, \dots, M\}$,
    \begin{align}
        \|  V_{m_2}\ldots V_{m_1} - W_{\bA}((M-m_1+1)h, (M-m_2)h) \| \leq \frac{\epsilon}{M + R}.
    \end{align}
    Then the maximum  row sum of $\mA^{-1}$ is bounded by
    \begin{align}
        \max_m \sum_{n} \norm{(\mA^{-1})_{mn}} \leq \max_{m \in \{1,2,\dots, M\}}\sum_{n=0}^{m-1} \norm{W_{\bA}((M-m+n+1)h, (M-m)h)} + R + \epsilon.
    \end{align}
    
    Let us now show how to bound the sum in terms of an integral. To simplify notation, let us define a rescaled matrix $\wbA:= \bA h = \bA/a$. Then we are interested in bounding the following expression:
    \begin{align}
        \sum_{n=0}^{m-1} \norm{W_{\wbA}(M-m+n+1, M-m)}.
    \end{align}
    Here, $W_{\wbA}(m,m')$ is $W_{\bA}(mh,m'h)$.
    Let $\tau_m \in [m, m+1]$ and define $T_n := M-n$.
    Using the composition rule for the Hermitian conjugate of a transition matrix, we then have
    \begin{align}
    \begin{split}
        \norm{W_{\wbA} (T_n+m+1, T_n)} &= \norm{W_{\wbA} (T_n + \tau_m, T_n) W_{\wbA} (T_n + m+1, T_n + \tau_m)} \\
        &\leq \norm{W_{\wbA} (T_n + m+1, T_n + \tau_m)} \norm{W_{\wbA} (T_n + \tau_m, T_n)} \\
        &\leq e^{\norm{\wbA }} \norm{W_{\wbA} (T_n + \tau_m, T_n)} \\
        &\leq e \norm{W_{\wbA} (T_n + \tau_m, T_n)}.
    \end{split}
    \end{align}
    Integrating both sides from $m$ to $m+1$, we see that
    \begin{align}
        \norm{W_{\wbA} (T_n+m+1, T_n)} = \int_m^{m+1} \rd\tau_m \, \norm{W_{\wbA} (T_n+m+1, T_n)}
        \leq e \int_m^{m+1} \rd\tau_m \, \norm{W_{\wbA} (T_n + \tau_m, T_n)}.
    \end{align}
    Now we  sum over $m$ on both sides:
    \begin{align}
    \begin{split}
        \sum_{m=0}^{n-1} \norm{W_{\wbA} (T_n+m+1, T_n)} &\leq \sum_{m=0}^{n-1} e \int_m^{m+1} \rd\tau_m \, \norm{W_{\wbA} (T_n + \tau_m, T_n)} \\
        &\leq e \int_{0}^n \rd\tau \norm{W_{\wbA} (T_n + \tau, T_n)} \\
        &\in \cO \lb a \int_0^{nh} \rd s \, \norm{W_{\bA} ((M-n)h + s, (M-n)h)} \rb.
    \end{split}
    \end{align}
    Thus, the maximum row sum can be upper bounded as follows:
    \begin{align}
         \max_m \sum_{n} \norm{(\mA^{-1})_{mn}} \le e a \max_{t_0 \in [0,t]}\int_{t_0}^{t} \rd s \, \norm{W_{\bA} (s, t_0)}  +R+1.
    \end{align}

    By a similar argument, we obtain the following bound on the maximum column sum:
    \begin{align}
    \begin{split}
        \max_n \sum_{m} \norm{(\mA^{-1})_{mn}} &\leq \max_{n \in \{1,2,\dots, M\}}\lb \|V_n\|+ \|V_{n+1}V_n\|+\ldots + \|V_M \ldots V_n\|\rb + R \max \left\{1, \max_{n \in \{1,\ldots, M\}} \norm{V_M \ldots V_n} \right\} \\
        &\leq \max_{n \in \{1,2,\dots, M\}} \sum_{m=n}^M \norm{W_{\bA}((M-n + 1)h, (M-m)h)} + R \max_{s \in [0,t]} \norm{W_{\bA}(s,0)} + \epsilon \\
        &\in \cO \lb a  \max_{t_1 \in [0,t]}\int_{0}^{t_1} \rd s \, \norm{W_{\bA} (t_1, s)}  + R \max_{s \in [0,t]} \norm{W_{\bA} (s,0)} \rb .
    \end{split}
    \end{align}

    We therefore have that
    \begin{align}
        \kappa_{\mA} \in \cO \lb a \max \left\{ \max_{t_1 \in [0,t]}\int_{0}^{t_1} \rd s \, \norm{W_{\bA}(t_1,s)} , \max_{t_0 \in [0,t]} \int_{t_0}^t \rd s \, \norm{W_{\bA}(s, t_0)}  \right\} + R \max_{s \in [0,t]} \norm{W_{\bA}(s,0)} \rb.
    \end{align}
    As shown in Ref.~\cite{Berry2024time-dep_history}, the number of queries to the initial state preparation unitary $U_{\phi}$ scales like $\cO \lb \kappa_{\mA} \log (1/\epshist) \rb$. Using the above bound on $\kappa_{\mA}$ then yields the claimed query complexity for $U_{\phi}$. 

    To determine the required number of queries to $U_{\bA}$ as well as the number of additional gates, we need to bound the truncation order of the Dyson series for approximating short time evolutions.
    Let $\delta_m$ denote the error matrix associated with $V_m$ consisting of the remainder of the truncated Dyson series such that 
    \begin{align}
        V_m =  W_{\bA}(t_{M-m+1},t_{M-m}) + \delta_m,
    \end{align}
    meaning that
    \begin{align}
        \delta_m = - \sum_{k=K+1}^{\infty} \int_{t_{M-m}}^{t_{M-m+1}} \rd\tau_1 \int_{t_{M-m}}^{\tau_1} \rd\tau_2 \cdots \int_{t_{M-m}}^{\tau_{k-1}} \rd\tau_k \;\bA (\tau_k) \cdots \bA (\tau_2) \bA (\tau_1).
    \end{align}
    Let $\Delta$ be an upper bound on the norm of any $\delta_m$ and let $W_{\max}$ be an upper bound on the norm of $W_{\bA}(\tau_1, \tau_0)$ for any $\tau_1 \geq \tau_0 \in [0,t]$. Note that $W_{\max} \geq 1$.
    In contrast to the time-independent case, $W_{\bA}(t_{M-m+1},t_{M-m})$ and $\delta_m$ will generally not commute. However, we can still bound the error using a similar union-bound like argument. In particular, the error for a sequence of $V_m'$s is bounded as ($t_m=mh$)
    \begin{align}
    \begin{split}
        \norm{\prod_{j=m}^n V_j - W_{\bA}(t_{M-j+1},t_{M-j})} &= \norm{\prod_{j=m}^n\lb W_{\bA} (t_{M-j+1},t_{M-j}) + \delta_j \rb - \prod_{j=m}^n W_{\bA} (t_{M-j+1},t_{M-j})} \\
        &\leq \binom{n+1-m}{1} W_{\max}^2 \Delta + \binom{n+1-m}{2} W_{\max}^3 \Delta^2 
        + \dots + \binom{n+1-m}{n+1-m} W_{\max}^{n+2-m} \Delta^{n+1-m} \\
        &= W_{\max}^2 \Delta \sum_{j=1}^{n+1-m} \binom{n+1-m}{j} \lb W_{\max}\Delta \rb^{j-1} \\
        &= W_{\max}^2 \Delta \sum_{j=0}^{n-m} \binom{n+1-m}{j+1} \lb W_{\max}\Delta \rb^{j} \\
        &= W_{\max}^2 \Delta (n+1-m) \sum_{j=0}^{n-m} \frac{(n-m)!}{(j+1)!(n-m-j)!}  \lb W_{\max}\Delta \rb^{j} \\
        &\leq W_{\max}^2 \Delta (n+1-m) \sum_{j=0}^{n-m} \binom{n-m}{j}  \lb W_{\max}\Delta \rb^{j}  \\
        &\leq W_{\max}^2 \Delta (n+1-m) e^{W_{\max} \Delta}.
    \end{split}
    \end{align}
    To ensure that for any $n\geq m \in \{1, 2, \dots, M\}$ this error is at most $\epsilon'$ for some intermediate error tolerance $\epsilon' < 1$, it suffices to ensure that
    \begin{align}
        \Delta \leq \frac{\epsilon'}{2 W_{\max}^2 M},
    \end{align}
    which follows from a Taylor series expansion of the Lambert-W function near $0$ under the assumption that
    \begin{align}
        \frac{\epsilon'}{ W_{\max}^2 M} < \frac{1}{e}.
    \end{align}
    To provide a bound on $\epsilon'$ in terms of $\epshist$, we need to determine how the errors propagate to the overall history state. Note that the $\ell_2$-norm error in the unnormalized history state when using a sequence of $V_m$'s instead of $W_{\bA} (t_{M-m+1},t_{M-m})$ is bounded by
    \begin{align}
    \begin{split}
        \sqrt{\sum_{m=1}^{M-1} \norm{ V_m\ldots V_1  \ket{\phi} - W_{\bA} (t,t_{M-m}))   \ket{\phi}}^2 + \sum_{m=M}^{M+R-1} \norm{V_M\ldots V_1 \ket{\phi} -  W_{\bA} (t,0) \ket{\phi}}^2} \leq \sqrt{M+R} \epsilon'.
    \end{split}
    \end{align}
    To ensure that the error between the normalized history states is at most $\epsilon$, it then suffices to pick 
    \begin{align}
        \epsilon' \in \cO \lb \frac{\epshist}{\sqrt{M+R}} \rb.
    \end{align}
    Thus, it suffices to truncate the Dyson series for the time-ordered matrix exponentials at order
    \begin{align}
        K \in \cO \lb \frac{\log \lb \lb a t + R \rb W_{\max}/\epshist \rb}{\log \log \lb \lb a t + R \rb W_{\max}/\epshist \rb} \rb.
    \end{align}
    As shown in Ref.~\cite{Berry2024time-dep_history}, the number of queries to $U_{\bA}$ scales like $\cO \lb \kappa_{\mA} \log (1/\epsilon) \times K \rb$. Using the above bounds on $\kappa_\mA$ and $K$ yields the claimed query complexity for $U_\bA$.
    
    The bound on the number of additional gates follows directly from the proof of Thm.\ 4.1 in Ref.~\cite{Berry2024time-dep_history} by ignoring all terms related to the inhomogeneous term. Note that the number of queries to $U_{\bA}$ does not depend on the norm of the derivative of $\bA(s)$. Only the number of additional gates has a logarithmic dependence on $\max_{s \in [0,t]} \norm{\bA'(s)}$ which stems from the discretization of the integrals in the Dyson series. Here $\bA'(s)=\frac{\rd}{\rd s}\bA(s)$.
\end{proof}

The same bounds can be obtained for preparing an approximation to $\ket{\wpsi_\hist}$ by replacing $\bA$ with $\bB$ in the above theorem.

\subsection{History state preparation via LCHS in the time-dependent case}

Similar to the time-independent case in Appendix~\ref{app:QAlgorithm}, we can also use LCHS~\cite{an2023lchs, an2023betterlchs, low2025optimal} to prepare history states in the case of time-dependent matrices.
Below, we give a reformulated version of the main theorem of Ref.~\cite{low2025optimal} for convenience.

\begin{theorem}[Block-encoding of LCHS (time-dependent); reformulation of Theorem 4 in \cite{low2025optimal}]
    Let $A(s): [0,t] \rightarrow \mathbb{C}^{N \times N}$ have non-positive log-norm for all $s \in [0,t]$ and assume we have access to a time-dependent block-encoding $U_{A}$ of $A(s)$ with block-encoding constant $a \geq \max_{s \in [0,t]} \norm{A(s)}$.
    For any $t \geq 0$, $\epsilon \in (0, 4/5]$, we can construct a block-encoding $U_{\cT \exp(tA)}$ such that
    \begin{align}
        \norm{\lb \bra{0} \otimes \one \rb  U_{\cT \exp(tA)}  \lb \bra{0} \otimes \one \rb - \cT e^{\int_0^t \rd s A(s)}} \leq \epsilon,
    \end{align}
    using
    \begin{align}
        Q \in \cO \lb ta \log(1/\epsilon) \times \log(ta \log(1/\epsilon)/\epsilon) \rb
    \end{align}
    queries to controlled-$U_{A}$ and its inverse, plus an additional
    \begin{align}
        \cO \lb Q \times \log \lb \frac{t \log(1/\epsilon)}{a \epsilon} \max_{s \in [0,t]} \norm{\frac{\rd A(s)}{\rd s}} + \frac{ta \log(1/\epsilon)}{\epsilon}  \rb + \log \lb at + \log(1/\epsilon) \rb \times \log^{5/2} \log(1/\epsilon) \rb 
    \end{align}
    primitive quantum gates.
\label{thm:LCHS_time-dep}
\end{theorem}

Note that we assumed that $A$ has non-positive log-norm in the above theorem. In the following, we will no longer make this assumption explicitly, but instead shift $A$ by an appropriately scaled identity term such that the resulting matrix has non-positive log-norm.
The goal now is to implement a block-encoding of the following controlled time-evolution operator:
\begin{align}
    V_{A(s)} &:= \sum_{m=0}^{R-1} \ketbra{m}{m} \otimes W_{A_\xi}(t,0) + \sum_{m=R}^{M+R-1} \ketbra{m}{m} \otimes W_{A_\xi}(t,t_{m+1-R}) 
\label{controlled_time-ev-op_time-dep},
\end{align}
where $A_\xi(s) := A(s) - \xi_A(s)$ with $\xi_A(s)$ being a known function that upper bounds the log-norm of $A(s)$ for any given $s \in [0,t]$.
The lemma below provides an upper bound on the complexity of block-encoding $V_{A(s)}$.

\begin{lemma}[Block-encoding of controlled time evolution operator (time-dependent)]
    Let $U_{A}$ be an $(a, \rm a, 0)$-block-encoding of $A(s)$ such that
    \begin{align}
        \lb \bra{0}_{\rm a}  \otimes \one \rb  U_{A} \lb \ket{0}_{\rm a}  \otimes \one \rb &= \sum_{j=0}^{J-1} \ketbra{j}{j} \otimes \frac{A(\frac{j}{J}t)}{a},
    \end{align}
    with $J \in \cO \lb \frac{t^2}{\epsilon} \max_{s \in [0,t]} \norm{\frac{\rd A}{\rd s}} \rb$, and let $\epsilon > 0$ be an error tolerance.
    We can construct a block-encoding $U_{V_{A(s)}}$ of $V_{A(s)}$ as given in Eq.~\eqref{controlled_time-ev-op_time-dep} such that
    \begin{align}
        \norm{\lb \bra{0} \otimes \one \rb U_{V_{A(s)}}  \lb \ket{0} \otimes \one \rb - V_{A(s)}} \leq \epsilon,
    \end{align}
    using
    \begin{align}
        \cO \lb ta \log(1/\epsilon) \times \log(ta \log(1/\epsilon)/\epsilon) \rb
    \end{align}
    queries to controlled-$U_{A}$ and its inverse, plus an additional
    \begin{align}
        \cO \lb ta \times \mathrm{polylog} \lb ta, 1/\epshist, g_t \max_{s \in [0,t]}\norm{\frac{\rd A}{\rd s}} \rb \rb
    \end{align}
    primitive quantum gates, where $g_t := \max \{ t/a, t^2 \}$.
\label{lem:be_ctrl_ev_time-dep}
\end{lemma}

\begin{proof}
    As shown in Ref.~\cite{low2025optimal}, we can block-encode the time-ordered matrix exponential $W_{A_\xi}(t,0)$ via the following multiplexed block-encoding:
    \begin{align}
       \textsc{Mul} := \sum_{j'=-R'/h'}^{R'/h'} \ketbra{j'}{j'} \otimes \Be \left[ \frac{h'j'\bL(s) + \bH(s)}{R' \lambda_{\bL} + \lambda_\bH} \right],
    \end{align}
    where $\bL(s) = \lb A(s) + A^\dagger(s) \rb/2 - \xi_{A}(s) \one$, $\bH(s) = i \lb A^\dagger(s) - A(s) \rb/2$, $R' \in \cO \lb \log(1/\epsilon) \rb$ and $1/h' \in \cO \lb \norm{\bL}_{L^1} + \log ( 1/\epsilon) \rb$ with $R'/h'$ being an integer. 
    Here, $\Be[\cdot]$ refers to a block-encoding of the expression inside the brackets with the block-encoding constant being in the denominator. In particular, $\lambda_{\bL}, \lambda_{\bH} \in \cO \lb a \rb$ are the block-encoding constants of $\bL$ and $\bH$, respectively.
    $\textsc{Mul}$ can be implemented using $\cO(1)$ queries to controlled-$U_{A}$ and its inverse and an additional $\cO \lb \log (R'/h')\rb$ primitive gates.
    
    We can block-encode $V_{A(s)}$ by replacing \textsc{Mul} with the following multiplexed block-encoding:
    \begin{align}
        \textsc{MulT} &:= \sum_{m=0}^{M+R-1} \ketbra{m}{m} \otimes \sum_{j'=-R'/h'}^{R'/h'} \ketbra{j'}{j'} \otimes \Be \left[ \mathtt{IND}(m) \otimes \one \right] \Be \left[\frac{h'j'\bL(s) + \bH(s)}{R' \lambda_{\bL} + \lambda_\bH} \right],
    \end{align}
    where $\mathtt{IND}(m)$ is the matrix version of a time-step-dependent indicator function. Specifically, $\mathtt{IND}(m)$ is a diagonal matrix of dimension $J$ (the size of the time register in the block-encoding of $A(s)$) with entries
    \begin{align}
        \mathtt{IND}(m)_{j,j} = 
        \begin{cases}
            1 , & 0 \leq m \leq R-1 \\
            1 , & (R \leq m \leq M+R-1) \wedge \lb \frac{j}{J}t \in [t_{m+1-R}, t] \rb \\
            0 , & \text{otherwise}.
        \end{cases}
    \end{align}
    This matrix can be block-encoded via inequality testing using $\cO \lb \log (J) \rb$ primitive quantum gates.
    With \textsc{MulT} as defined above we can then follow the same optimal LCHS procedure as in Ref.~\cite{low2025optimal}.
    
    More specifically, by querying $\textsc{MulT}$ a number of times scaling like $\cO \lb ta \log(1/\epsilon) \times \log(ta \log(1/\epsilon)/\epsilon) \rb$, we can use the results from Ref.~\cite{low2019interactionpicture} for time-dependent Hamiltonian simulation to implement
    \begin{align}
        \textsc{SelT}' := \sum_{m=0}^{R-1} \ketbra{m}{m} \otimes \sum_{j'=-R'/h'}^{R'/h'} \ketbra{j'}{j'} \otimes \Be \left[ U_{j',0} \right] + \sum_{m=R}^{M+R-1} \ketbra{m}{m} \otimes \sum_{j'=-R'/h'}^{R'/h'} \ketbra{j'}{j'} \otimes \Be \left[ U_{j',m+1-R} \right] ,
    \end{align}
    where $\Be \left[ U_{j',m} \right]$ is a block-encoding of an operator $U_{j',m}$ such that $\norm{U_{j',m} - \cT e^{-it \int_{t_m}^t \rd s \,(h'j' \bL(s) + \bH(s))}} \le \epsilon$.
    Putting everything together, we therefore find from Theorem~\ref{thm:LCHS_time-dep} that the number of queries to controlled-$U_{A}$ and its inverse scales like
    \begin{align}
        \cO \lb ta \log(1/\epsilon) \times \log(ta \log(1/\epsilon)/\epsilon) \rb,
    \end{align}
    and the overall number of additional primitive gates scales like
   \begin{align}
        \cO \lb ta \times \mathrm{polylog} \lb ta, 1/\epshist, g_t \max_{s \in [0,t]}\norm{\frac{\rd A}{\rd s}} \rb \rb.
    \end{align}
\end{proof}

Now we are ready to present the main result of this subsection which provides a bound on the complexity of preparing (sub-)normalized history states via LCHS in the time-dependent setting.

\begin{theorem}[History state preparation via LCHS (time-dependent)]
    Let $U_A$ be a block-encoding of $\sum_{j=0}^{J-1}\ketbra j \otimes \frac{A(\tau_j)}a$ for $a \ge \|A(t)\|>0$ and set $M= \lceil ta \rceil$, $R\ge 1$. Further, let $\xi_A(s)$ be a known function upper bounding the log-norm of $A(s)$ for any $s \in [0,t]$ and let $U_{\hist}^{(\LCHS)}$ be a unitary such that
    \begin{align}
       U_{\hist}^{(\LCHS)} \ket{0}_{\mathrm{f}}\ket{0}_{\mathrm{clock}}\ket{0}_{\hist} = \ket 0_{\rm f} \frac{1}{\sqrt{\cN_{\xi_A}}}\lb \sum_{m=0}^{R-1} \ket{m}_{\rm clock} W_{A} (t,0) \ket{x_\init} + \sum_{m=R}^{R+M-1} \ket{m}_{\rm clock} W_{A} (t,t_{m+1-R}) \ket{x_\init} \rb + \ket{0^\perp} \; ,
    \end{align}
    where the $\rm f$ register flags success, and
    \begin{align}
        \cN_{\xi_A} = \lb \sum_{m=0}^{R-1} e^{2 \int_0^t \rd \tau \, \xi_A(\tau)} + \sum_{m=R}^{M+R-1} e^{2 \int_{t_{m+1-R}}^t \rd \tau \, \xi_A(\tau)} \rb \in \cO \lb R \, e^{2 \int_{0}^t \rd \tau \, \xi_A(\tau)} + a \int_0^t \rd s \, e^{2 \int_{s}^t \rd \tau \, \xi_A(\tau)} \rb.
    \end{align}
    Let $U_x$ be a state preparation unitary such that $U_x \ket{0} =\ket{x_{\init}}$.
    Then we can implement $U_{\hist}^{(\LCHS)} \ket{0}_{\mathrm{f}}\ket{0}_{\mathrm{clock}}\ket{0}_{\hist}$ within error $\epshist$ in Euclidean distance using 1 query to $U_{x}$,
    \begin{align}
        \cO \lb ta \log(1/\epshist) \times \log(ta \log(1/\epshist)/\epshist) \rb,
    \end{align}
    queries to controlled-$U_{A}$ and its inverse, plus an additional
    \begin{align}
        \cO \lb (ta + R) \times \mathrm{polylog} \lb ta, 1/\epshist, g_t \max_{s \in [0,t]}\norm{\frac{\rd A}{\rd s}} \rb \rb
    \end{align}
    primitive quantum gates, where $g_t := \max \{ t/a, t^2 \}$.
\label{thm:history_state_LCHS_time-dep}
\end{theorem}

\begin{proof}
    We use Lemma~\ref{lem:be_ctrl_ev_time-dep} to construct a block-encoding $U_{V_{A(s)}}$ of 
    \begin{align}
        V_{A(s)} = \sum_{m=0}^{R-1} \ketbra{m}{m} \otimes W_{A_\xi}(t,0) + \sum_{m=R}^{M+R-1} \ketbra{m}{m} \otimes W_{A_\xi}(t,t_{m+1-R}),
    \end{align}
    where $A_\xi = A(s) - \xi_A(s)$ as before.
    Then we apply $U_{V_{A(s)}}$ to the following normalized quantum state:
    \begin{align}
        \ket{\mathcal{P}_{A, x}} := \ket{0}_{V} \otimes \frac{1}{\sqrt{\cN_{\xi_A}}} \lb \sum_{m=0}^{R-1} e^{ \int_0^t \rd \tau \, \xi_A(\tau)} \ket{m} + \sum_{m=R}^{M+R-1} e^{\int_{t_{m+1-R}}^t \rd \tau \, \xi_A(\tau)} \ket{m} \rb \otimes \ket{x_{\init}},
    \end{align}
    where the $V$ register is the block-encoding register of $U_{V_{A(s)}}$ (which can also be interpreted as the register flagging success). Using the same integral upper bound as in Lemma~\ref{lem:norm_history_state_time-dep}, we have that
    \begin{align}
        \cN_{\xi_A} \leq \lb R \, e^{2 \int_0^t \rd \tau \, \xi_A(\tau)} + \frac{e^2 M}{t} \int_0^t \rd s \, e^{2 \int_{t_{m+1-R}}^t \rd \tau \, \xi_A(\tau)}  \rb \in \cO \lb R \, e^{2 \int_{0}^t \rd \tau \, \xi_A(\tau)} + a \int_0^t \rd s \, e^{2 \int_{s}^t \rd \tau \, \xi_A(\tau)} \rb.
    \end{align}
    The state $\ket{\mathcal{P}_{A, x}}$ can be implemented within error $\epsilon_{\cP_A}$ in Euclidean distance using $1$ query to $U_{x}$ and $\cO \lb (M+R) \log \lb 1/\epsilon_{\cP_A}\rb \rb$ primitive gates~\cite{NC2010, gosset2025stateprep}.
    To ensure that the overall error is at most $\epshist$, it suffices to block-encode $U_{V_{A}}$ within error $\epshist/2$ and also prepare $\ket{\mathcal{P}_{A, x}}$ within error $\epshist/2$. By Lemma~\ref{lem:be_ctrl_ev_time-dep}, we therefore require
    \begin{align}
        \cO \lb ta \log(1/\epshist) \times \log(ta \log(1/\epshist)/\epshist) \rb
    \end{align}
    queries to controlled-$U_{A}$ and its inverse. The overall number of additional primitive gates scales at most like 
    \begin{align}
        \cO \lb (ta + R) \times \mathrm{polylog} \lb ta, 1/\epshist, g_t \max_{s \in [0,t]}\norm{\frac{\rd A}{\rd s}} \rb \rb.
    \end{align}
\end{proof}

\subsection{Quantum complexity of estimating a matrix element in the time-dependent case}

Let us first discuss how to block-encode $\mathcal{I}_{\bC, \bD}$. Recall that
\begin{align}
    \mathcal{I}_{\bC, \bD} = \sum_{m=0}^{R-1} \ketbra{m}{m} \otimes \frac{\bD}{R} + \sum_{m=R}^{R+M-1} \ketbra{m}{m} \otimes I_{\bC}(t_{m-R}),
\end{align}
with
\begin{align}
    I_{\bC}(t_{m}) = \int_0^h d\tau \,  W_{\bA}^\dagger(t_{m+1},t_m + \tau) \bC(t_m + \tau)  W_{\bB} (t_{m+1},t_m + \tau) 
\end{align}
and $h \in \cO \lb 1/\mu \rb$ where $\mu := \max \{ a, b\}$.

\begin{lemma}[Block-encoding $\mathcal{I}_{\bC, \bD}$ in the time-dependent setting]
Let $a \ge \|\bA\|$, $b \ge \|\bB\|$, $c \ge \|\bC\|$, and $d \ge \|\bD\|$, and set $\mu = \max \{a,b\}$, $h=t/\lceil \mu t \rceil$. Let $U_\bA$, $U_\bB$, $U_\bC$, and $U_\bD$ be block-encodings of $\bA(\tau)/a$, 
$\bB(\tau)/b$, $\bC(\tau)/c$, $\bD(\tau)/d$, respectively, which use at most ${\rm a}\ge 0$ ancilla qubits. 
    We can implement a $(\lambda_{\mathcal{I}_{\bC, \bD}}, \rm a_{\mathcal{I}_{\bC, \bD}}, \epsbe)$-block-encoding of $\mathcal{I}_{\bC, \bD}$, where
    \begin{align}
        \lambda_{\mathcal{I}_{\bC, \bD}} &\le \max \left\{ {che^2} , \frac{d}{R} \right\} ,\\
        \rm a_{\mathcal{I}_{C, D}} &= 2K \rm a+ \cO \lb   \log \lb \frac{\max_{s \in [0,t]} \norm{\bA'(s)}}{\mu^2 \, \epsbe}  \rb + \log \lb \frac{\max_{s \in [0,t]} \norm{\bB'(s)} }{\mu^2 \, \epsbe} \rb + \log \lb \frac{c}{\mu \epsbe} + \frac{\max_{s \in [0,t]} \norm{\bC'(s)}}{\mu^2 \epsbe} \rb \rb,
    \end{align}
    using 
    \begin{align}
      2K\in  \cO \lb \frac{\log \lb \lambda_{\mathcal{I}_{\bC, \bD}}/\epsbe \rb}{\log \log \lb \lambda_{\mathcal{I}_{\bC, \bD}}/\epsbe \rb} \rb
    \end{align}
    queries to $U_{\bA}^\dagger$ and $U_{\bB}$, one query to $U_{\bC}$ and one query to $U_{\bD}$ and additional
    \begin{align}
    \begin{split}
         &\cO \Bigg( \lb K\rm a + \log \lb \frac{\max_{s \in [0,t]} \norm{\bA'(s)}}{\mu^2 \, \epsbe}  \rb + \log \lb \frac{\max_{s \in [0,t]} \norm{\bB'(s)} }{\mu^2 \, \epsbe} \rb \rb \times \frac{\log \lb \lambda_{\mathcal{I}_{\bC, \bD}}/\epsbe \rb}{\log \log \lb \lambda_{\mathcal{I}_{\bC, \bD}}/\epsbe \rb} \\
         &\qquad + \log \lb \frac{c}{\mu \epsbe} + \frac{\max_{s \in [0,t]} \norm{\bC'(s)}}{\mu^2 \epsbe} \rb \Bigg),
    \end{split}
    \end{align}
    primitive gates.
\label{lem:block-encode_ICD_time-dep}
\end{lemma}

\begin{proof}
    First, we discuss how to block-encode $I_{\bC}(t_{m-R})$, which is the short-time integral connecting time steps $m$ and $m+1$. For the remainder of this proof, we will refer to $I_{\bC}(t_{m-R})$ as $I_{\bC}(t_m)$ for simplicity since $R$ is just an offset arising from the reverse-order history states.
    We block-encode $I_{\bC}(t_m)$ by discretizing the integral via a Riemann sum. In principle, a better numerical integrator could be used but that would only affect a logarithmic factor in the number of additional gates so we choose the Riemann sum discretization for simplicity.
    
    We can ensure that the discretization error of $I_{\bC}(t_m)$ is at most $\epsbe/2$ for any $m \in \{0, 1, \dots, M-1\}$ by evaluating the integrand at
    \begin{align}
        G := \Bigg\lceil \frac{h^2}{\epsbe} \max_{m \in \{ 0, 1, \dots, M-1 \}}\max_{\tau \in [0,h]} \norm{\frac{\rd }{\rd \tau}{W_{\bA}^\dagger(t_{m+1},t_m + \tau) \bC(t_m + \tau)  W_{\bB}  (t_{m+1},t_m + \tau)} } \Bigg\rceil
    \end{align}
    equidistant grid points.
    Note that from Eq.~\eqref{eq:DysonWA}
    and~\eqref{eq:DysonWB},
    \begin{align}
    \nonumber
       & \frac{\rd}{\rd \tau} W_{\bA}(t_{m+1},t_m + \tau)  = - \bA(t_m + \tau) W_{\bA}(t_{m+1},t_m + \tau) \\
        \implies & \frac{\rd}{\rd \tau} W^\dagger_{\bA}(t_{m+1},t_m + \tau) = -  W^\dagger_{\bA}(t_{m+1},t_m + \tau)\bA^\dagger(t_m + \tau),\\
        & \frac{\rd}{\rd \tau} W_{\bB}(t_{m+1},t_m + \tau)  = -  \bB(t_m + \tau) W_{\bB}(t_{m+1},t_m + \tau).
    \end{align}
     Since $h = 1/\mu$ we have that $\norm{W^\dagger_{\bA}(t_{m+1},t_m + \tau_g)} \leq e$ and $\norm{W_{\bB} (t_{m+1},t_m + \tau_g)} \leq e$. Thus, for $\bC'(s)=\frac{\rd}{\rd s}\bC(s)$,
    \begin{align}
        G \in \cO \lb \frac{h^2}{\epsbe} \lb \mu c + \max_{s \in [0,t]} \norm{\bC'(s)} \rb \rb \subseteq \cO \lb \frac{c}{\mu \epsbe} + \frac{\max_{s \in [0,t]} \norm{\bC'(s)}}{\mu^2 \epsbe} \rb.
    \end{align}

    Let us denote the Riemann sum approximation of $I_{\bC}(t_m)$ by
    \begin{align}
       \tilde{I}_{\bC}(t_m) := \frac{h}{G}\sum_{g=0}^{G-1} W^\dagger_{\bA}(t_{m+1},t_m + \tau_g) \bC(t_m + \tau_g)  W_{\bB} (t_{m+1},t_m + \tau_g),
    \end{align}
    where $\tau_g := hg/G$. 
    
    We can block-encode $W^\dagger_{\bA}(t_{m+1},t_m + \tau_g)$ and $W_{\bB} (t_{m+1},t_m + \tau_g)$ approximately via truncated Dyson series. 
    Let $\tilde{W}_\bA^\dagger$ and $\tilde{W}_\bB$ denote the corresponding approximations, where we suppress the dependence on $t_m$ and $\tau_g$ for simplicity.
    It suffices to truncate the Dyson series at order
    \begin{align}
        K \in \cO \lb \frac{\log \lb 1/\epsilon_{\mathrm{Dyson}} \rb}{\log \log \lb 1/\epsilon_{\mathrm{Dyson}} \rb} \rb,
    \end{align}
    to ensure that the spectral norm error between the true transition matrices and the truncated Dyson series is at most $\epsilon_{\mathrm{Dyson}}$, i.e.
    \begin{align}
        \norm{\tilde{W}_\bA^\dagger - W_{\bA}^\dagger(t_{m+1},t_m + \tau_g)} &\leq \epsilon_{\mathrm{Dyson}}, \\
        \norm{\tilde{W}_\bB - W_\bB(t_{m+1},t_m + \tau_g)} &\leq \epsilon_{\mathrm{Dyson}}.
    \end{align}
  Since these correspond to short-time evolutions, both $\tilde{W}_\bA^\dagger$ and $\tilde{W}_\bB$ can be block-encoded with block-encoding constants scaling like $\cO(1)$ using $K$ queries to $U^\dagger_{\bA}$ and $U_{\bB}$, respectively~\cite{low2019interactionpicture, Berry2024time-dep_history}.
    
    By choosing $\epsilon_{\mathrm{Dyson}} \in \cO \lb \epsbe/(c h) \rb$, we can construct $\epsbe/2$-precise block-encodings $\{U_{\tilde{I}(m,g)}\}$ of the individual terms of the Riemann sum,
    \begin{align}
        W_{\bA}^\dagger(t_{m+1},t_m + \tau_g) \bC(t_m + \tau_g)  W_{\bB} (t_{m+1},t_m + \tau_g) \,h,
    \end{align}
    using $K \in \cO \lb \log \lb c h/\epsbe  \rb \rb$ queries to $U_{\bA}^\dagger$, $K$ queries to $U_{\bB}$, and one query to $U_{\bC}$. The block-encoding constant of any $U_{\tilde{I}(m,g)}$ is upper bounded by $ c h e^2=ce^2/\mu$.

    To block-encode $\mathcal{I}_{\bC, \bD}$, we need need to apply $U_{\tilde{I}(m,g)}$ and $U_{\bD}$ in a controlled manner. As in the time-independent case, we require the block-encoding constants of all $U_{\tilde{I}(m,g)}$ and $U_{\bD}$ to be the same to ensure the correct weighting of the terms in the approximation to the integral. This can be achieved by purposefully increasing the block-encoding constants of all terms to the maximum block-encoding constant of any of the original terms. We can increase a block-encoding constant with the help of a single ancilla qubit. In particular, we perform a controlled $R_Y$ rotation on that ancilla qubit. The rotation angles are determined by the ratios of the original block-encoding constants and the largest block-encoding constant. The block-encoding constant of $U_{\bD}$ is $d/R$. Since we assume that $U_{\bD}$ is an error-free block-encoding of $\bD/R$, we can change its block-encoding cost without affecting the query complexity in the case where $d/R \leq c h e^2$. However, if $d/R > c h e^2$, then we need to adjust the block-encoding constants of $\{U_{\tilde{I}(m,g)}\}$. 
    To cover both scenarios, it suffices to demand
    \begin{align}
        \epsilon_{\mathrm {Dyson}} \in \cO \lb \epsbe/\lambda_{\mathcal{I}_{\bC, \bD}} \rb,
    \end{align}
    where $\lambda_{\mathcal{I}_{\bC, \bD}}   \max \left\{ c he^2, d/R \right\}$.
    
    Letting $\{\tilde{U}_{\tilde{I}(m,g)}\}$ and $\tilde{U}_{\bD}$ denote the renormalized block-encodings, we can then block-encode $\mathcal{I}_{\bC, \bD}$ using the following unitaries:
    \begin{align}
        \mathtt{PREP}_G\ket{0}\ket{m} &:= 
        \begin{cases}
            \ket{0}\ket{m},  &0\leq m \leq R-1, \\
             \frac{1}{\sqrt{G}} \sum_{g=0}^{G-1} \ket{g} \ket{m}, & \text{otherwise}.
        \end{cases} \\
        \mathtt{SEL}_G &:= \one \otimes \sum_{m=0}^{R-1} \ketbra{m}{m} \otimes \tilde{U}_{\bD} + \sum_{g=0}^{G-1} \sum_{m=R}^{R+M-1} \ketbra{g}{g} \otimes \ketbra{m}{m} \otimes \tilde{U}_{\tilde{I}(m,g)}.
    \end{align}
    In particular,
    \begin{align}
        \lb \mathtt{PREP}_G^\dagger \otimes \one \rb \mathtt{SEL}_G \lb \mathtt{PREP}_G \otimes \one \rb
    \end{align}
    is a $(\lambda_{\mathcal{I}_{\bC, \bD}}, \rm a_{\mathcal{I}_{\bC, \bD}}, \epsbe)$-block-encoding of $\mathcal{I}_{\bC, \bD}$. The bound on the error of the block-encoding follows from the triangle inequality. Specifically, we chose the number of discretization points for the Riemann sum such that the discretization error is upper by $\epsbe/2$. Furthermore, we chose the truncation order of the Dyson series such that we can block-encode the Riemann sums within error $\epsbe/2$.
    
    The bound on the number of ancilla qubits,
    \begin{align}
       { \rm a}_{\mathcal{I}_{C, D}}= 2K \rm a+\cO \lb  \log \lb \frac{\max_{s \in [0,t]} \norm{\bA'(s)}}{\mu^2 \, \epsbe}  \rb + \log \lb \frac{\max_{s \in [0,t]} \norm{\bB'(s)} }{\mu^2 \, \epsbe} \rb + \log \lb \frac{c}{\mu \epsbe} + \frac{\max_{s \in [0,t]} \norm{\bC'(s)}}{\mu^2 \epsbe} \rb \rb,
    \end{align}
    follows from Thm.\ 3 in Ref.~\cite{low2019interactionpicture} and Thm.\ 4.1 in Ref.~\cite{Berry2024time-dep_history} plus an additional $\cO \lb \log (G) \rb$ ancilla qubits for implementing the Riemann sums.

    Overall, we therefore require
    \begin{align}
        K \in \cO \lb \frac{\log \lb \lambda_{\mathcal{I}_{\bC, \bD}}/\epsbe \rb}{\log \log \lb \lambda_{\mathcal{I}_{\bC, \bD}}/\epsbe \rb} \rb
    \end{align}
    queries to $U_{\bA}^\dagger$ and $U_{\bB}$, and one query to $U_{\bC}$ and one to $U_{\bD}$.

    The bound on the number of additional gates,
    \begin{align}
    \begin{split}
         &\cO \Bigg( \lb K \rm a + \log \lb \frac{\max_{s \in [0,t]} \norm{\bA'(s)}}{\mu^2 \, \epsbe}  \rb + \log \lb \frac{\max_{s \in [0,t]} \norm{\bB'(s)} }{\mu^2 \, \epsbe} \rb \rb \frac{\log \lb \lambda_{\mathcal{I}_{\bC, \bD}}/\epsbe \rb}{\log \log \lb \lambda_{\mathcal{I}_{\bC, \bD}}/\epsbe \rb} \\
         &\qquad + \log \lb \frac{c}{\mu \epsbe} + \frac{\max_{s \in [0,t]} \norm{\bC'(s)}}{\mu^2 \epsbe} \rb \Bigg),
    \end{split}
    \end{align}
    follows from Thm.\ 3 in Ref.~\cite{low2019interactionpicture} and Thm.\ 4.1 in Ref.~\cite{Berry2024time-dep_history}. There is an additive factor of $\log(G)$ in the gate complexity due to the overhead associated with implementing $\mathtt{PREP}_G$ and $\mathtt{SEL}_G$ for the Riemann sum.
\end{proof}

Now we are ready to state the main theorem of this section, which provides an upper bound on the query complexity of estimating a matrix element in the case of time-dependent matrices.

\begin{theorem}[Cost of estimating a matrix entry in the time-dependent setting (linear systems approach)]
Let $a \ge \|\bA\|$, $b \ge \|\bB\|$, $c \ge \|\bC\|$, and $d \ge \|\bD\|$. Let $U_\bA$, 
$U_\bB$, $U_\bC$, and $U_\bD$ be the block-encodings of $\bA(\tau)/a$, $\bB(\tau)/b$, $\bC(\tau)/c$, and $\bD(\tau)/d$, respectively, which use at most ${\rm a}\ge 0$ ancilla qubits. 
Let $U_\phi$ and $U_\psi$ be the state preparation unitaries that map $\ket 0 \mapsto \ket \phi$ and $\ket 0 \mapsto \psi$, respectively.
    We can estimate
    \begin{align}
        \bra{\phi} \bX(t) \ket{\psi} = \bra{\phi}W_{\bA}^\dagger(t,0) \bD W_{\bB }(t,0)\ket{\psi} + \int_0^t \rd s \, \bra{\phi} W^\dagger_{\bA}(t,s) \bC (s) W_{\bB}(t,s) \ket{\psi} 
    \end{align}
    within additive error $\epsilon$ with probability of success at least $2/3$ using
    \begin{align}
        \cO \lb \frac{ c \tcL_2}{\epsilon} \times \mu \tcL_1 \times \log (c \tcL_2/\epsilon) \rb
    \end{align}
    queries to $U_{\phi}$ and $U_{\psi}$, and their inverses, 
    \begin{align}
        \cO \lb \frac{ c \tcL_2}{\epsilon}\times \mu \tcL_1 \times \log (c \tcL_2/\epsilon) \times \log \lb t \mu c^2 \tcL_1 \tcL_2/(d\epsilon) \rb \rb
    \end{align}
    queries to $U_{\bA}$ and $U_{\bB}$,
    \begin{align}
        \cO \lb \frac{c \tcL_2}{\epsilon}\rb
    \end{align}
    queries to $U_{\bC}$ and $U_{\bD}$, and their inverses, and
    \begin{align}
        \widetilde  \cO \lb \frac{c \tcL_2}{\epsilon} \times \mu \tcL_1 \times \log(t \mu) \times \log \lb \frac 1 \mu \max_{s\in[0,t]} \{t \|\bA'(s)\|,t\|\bB'(s)\|,\|\bC'(s)\|/c\}\rb\rb 
    \end{align}
    additional gates, where
    \begin{align}
        \mu &:= \max \{ a, b \}\;, \\
        \tcL_1 &:= \max_{\bY \in \{ \bA, \bB \}} \lb \norm{W_{\bY}}_{L^1} + \frac{d}{c} \max_{\substack{s' \geq s \\ s,s' \in [0,t]} } \norm{W_{\bY}(s',s)} \rb \;, \\
        \tcL_2 &:= \sqrt{\lb \int_0^t \rd s \, \norm{W_{\bA}(t,s)\ket{\phi}}^2 + \frac{d}{c} \norm{W_{\bA}(t,0)\ket{\phi}}^2 \rb \lb \int_0^t \rd s \, \norm{W_{\bB}(t,s)\ket{\psi}}^2 + \frac{d}{c} \norm{W_{\bB}(t,0)\ket{\psi}}^2 \rb},
    \end{align}
    with
    \begin{align}
        \norm{W_{\bY}}_{L^1} := \max \left\{ \max_{t_1 \in [0,t]} \int_{0}^{t_1} ds \, \norm{W_{\bY}(t_1,s)},   \max_{t_0 \in [0,t]} \int_{t_0}^t ds \, \norm{W_{\bY} (s, t_0)} \right\}.
    \end{align}

    Dropping subdominant logarithmic factors, the overall query complexity to $U_{\phi}$, $U_{\psi}$, $U_{\bA}$, $U_{\bB}$, $U_{\bC}$ and $U_{\bD}$ and their inverses scales like
    \begin{align}
        \widetilde{\cO} \lb \frac{\mu  c}{\epsilon} \tcL_1 \tcL_2 \log (t \mu) \rb.
    \end{align}
\label{thm:matrix_element_time-dep}
\end{theorem}

\begin{proof}
    The proof is similar to the proof of Thm.~\ref{thm:matrix_element}.
    The key idea is to perform overlap estimation with the following two states:
    \begin{align}
         \ket{\phi'_{\hist}} &:= \ket{0}_{\mathcal{I}_{\bC, \bD}} \otimes U_{\hist, \bA}  \ket 0\ket{0}_{\mathrm{clock}}\ket{0}_{\hist} \;,\\
        \ket{\mathcal{I} \psi'_{\hist}} &:= \tilde{U}_{\mathcal{I}_{\bC,\bD}} \lb \ket{0}_{\mathcal{I}_{\bC, \bD}} \otimes U_{\hist, \bB}   \ket 0 \ket{0}_{\mathrm{clock}} \ket{0}_{\hist}\rb,
    \end{align}
    where $\tilde{U}_{\mathcal{I}_{\bC,\bD}}$ is a block-encoding of
    \begin{align}
        \mathcal{I}_{\bC,\bD} \otimes \ketbra{0}{0}_{\hist},
    \end{align}
    which can be easily constructed from a block-encoding of $\mathcal{I}_{\bC,\bD}$.
    Note that
    \begin{align}
        \bra{\phi'_{\rm hist}}  \ket{\mathcal{I} \psi'_{\hist}} = \frac 1 {\lambda_{\cI_{\bC,\bD}}} \sqrt{\frac{p_\bA p_\bB}{\cN_\bA \cN_\bB}}
        \bra \phi \bX(t) \ket \psi \;,
    \end{align}
    reducing the problem to overlap estimation.
    Here $p_\bA>0$ and $p_\bB>0$ are constants.

    We can use Thm.~\ref{thm:history_state_time-dep} to implement $U_{\hist, \bA}$ and $U_{\hist, \bB}$ within error $\epshist$
    that we bound below. 
    We can also estimate $p_\bA$ and $p_\bB$ within some intermediate error tolerance $\epsilon_p$ via amplitude estimation, which is needed to obtain $\bra \phi \bX(t) \ket \psi$ from the above expression, using $\cO \lb 1/\epsilon_p \rb$ queries to $U_{\hist, \bA}$, $U_{\hist, \bB}$, and inverses.
    Moreover, we can use Lemma~\ref{lem:block-encode_ICD_time-dep} to block-encode $\mathcal{I}_{\bC, \bD}$ within some error $\epsbe$ and with block-encoding constant $\lambda_{\mathcal{I}_{\bC,\bD}}$.

    We choose the number of additional time steps $R \in \Theta \lb \frac{\mu d}{c} \rb$, but note that other options might be more optimal depending on the problem. This choice suffices for us. Then our estimate of $\bra{\phi} \bX(t) \ket{\psi}$ from the overlap estimation of $\ket{\phi'_{\hist}}$ and $\ket{\mathcal{I} \psi'_{\hist}}$ will be normalized by a factor (inverse) of
    \begin{align}
    \begin{split}
         & \sqrt{\mathcal{N}_\bA \mathcal{N}_\bB} \lambda_{\mathcal{I}_{\bC, \bD}}   \\
         &\subseteq \cO \Bigg( \sqrt{\lb \mu \int_0^t \rd s \, \norm{W_{\bA} (t,s) \ket{\phi}}^2 + R  \norm{W_{\bA} (t,0)\ket{\phi}}^2 \rb \lb \mu \int_0^t \rd s \, \norm{W_{\bB } (t,s)\ket{\psi}}^2 + R  \norm{W_{\bB } (t,0)\ket{\psi}}^2 \rb} \\
         &\qquad \quad \times \max \left\{{c e^2h} , \frac{d}{R} \right\} \Bigg) \\
         &\subseteq \cO \lb c \tcL_2 \rb,
    \end{split}
    \end{align}
    where $h=1/\mu$.
    This means that we need to perform amplitude estimation within error $\cO \lb \epsilon/\lb c \tcL_2 \rb \rb$. 
    Additionally, to ensure that the additive error for our estimate of $\bra{\phi} \bX(t) \ket{\psi}$ is at most $\epsilon$, it suffices if
    \begin{align}
        \epsilon_p \in \cO \lb \frac{\epsilon}{ c \tcL_2}\rb, \ \epshist  \in \cO \lb \frac{\epsilon}{ c \tcL_2} \rb 
        \end{align}
        and
       \begin{align}
        \epsbe &\in \cO \lb \frac{\epsilon}{\mu \tcL_2} \rb.
    \end{align}
    This is similar to Thm.~\ref{thm:matrix_element}.

    We are now ready to bound all complexities.
    \begin{itemize}
        \item From Thm.~\ref{thm:history_state_time-dep} we have the complexities of $U_{\rm hist,\bA}$
        and $U_{\rm hist,\bB}$ with precision $\epshist$ as above. For amplitude estimation at the Heisenberg limit, these unitaries and their inverses are needed $\cO(c \tcL_2/\epsilon)$ times. Since $R \sim \mu d /c$, this readily gives the number of uses of $U_\phi$, $U_\psi$, and inverses that is
        \begin{align}
            \cO \lb \frac{c \tcL_2}{\epsilon} \times \lb \max_{\bY\in\{\bA,\bB\}}\mu \norm{W_{\bY}}_{L^1} + R \max_{s \in [0,t]} \norm{W_{\bY}(s,0)} \rb \times \log (c\tcL_2/\epsilon) \rb = \cO\lb \frac{c \tcL_2}{\epsilon} \times \mu \tcL_1 \times \log (c\tcL_2/\epsilon)\rb\;. 
        \end{align}

        \item Also from  Thm.~\ref{thm:history_state_time-dep}, the number of queries to $U_\bA$ and $U_\bB$
        contain an extra logarithmic factor like $\log((t \mu + R) \max \|W_\bY(s',s)\|/\epshist)$.
        Considering a safe bound for all the parameters in this factor, we obtain
        for the number of queries to $U_\bA$ and $U_\bB$, and inverses, 
          \begin{align}
         Q=  \cO\lb \frac{c \tcL_2}{\epsilon} \times \mu \tcL_1 \times \log (c\tcL_2/\epsilon) \times \log (t \mu c^2 \tcL_1 \tcL_2/(d \epsilon))\rb\;. 
        \end{align}
We used $\max\|W_\bY(s',s)\| \le c\tcL_1/d$.

\item From overlap estimation at the Heisenberg limit, within error $\cO(\epsilon/(c\tcL_2))$, the number of queries to $U_\bC$, $U_\bD$, and the inverses is
\begin{align}
 Q'  = \cO \lb \frac{c \tcL_2}{\epsilon}\rb.
\end{align}

\item The gate complexity also follows from 
Thm.~\ref{thm:history_state_time-dep} and the gate complexity of $U_{\cI_{\bC,\bD}}$ from Lemma~\ref{lem:block-encode_ICD_time-dep}. 
The first contribution brings another factor to the query complexity $Q$, which is $\log ((t/\mu \epshist)\max\{ \|\bA'(s)\|, \|\bB'(s)\| \})$.
The second contribution brings another factor to $Q'$ that, in particular, contains a factor of the form $\log(\max_s \|\bC'(s)\|/(\mu^2 \epsbe))$.
 While there are other logarithmic factors, we will drop them to simplify the notation and obtain
 \begin{align}
     Q \log ((tc \tcL_2/\mu \epsilon)\max\{ \|\bA'(s)\|, \|\bB'(s)\| \}) + Q' 
     \log(\tcL_2 \max_s \|\bC'(s)\|/(\mu \epsilon))\;,
 \end{align}
 or
 \begin{align}
     \widetilde \cO \lb \frac{c \tcL_2}{\epsilon} \times \mu \tcL_1 \times \log (c\tcL_2/\epsilon) \times \log (t \mu c^2 \tcL_1 \tcL_2/(d \epsilon)) \times \log \lb  (\tcL_2/(\mu \epsilon))\max\{tc \|\bA'(s)\|,tc\|\bB'(s)\|, \|\bC'(s)\|\}\rb\rb
 \end{align}
 for the gate complexity, after replacing for $\epshist$ and $\epsbe$.
    \end{itemize}

We can simplify the expressions by dropping subdominant logarithmic factors.
Using the definition of $\tcL_1$, the overall query complexity to the state preparation unitaries and all block-encodings is 
\begin{align}
  \widetilde  \cO \lb \frac{c \tcL_2}{\epsilon} \times \mu \tcL_1 \times \log(t \mu)\rb \;,
\end{align}
being dominated by $Q$. The overall gate complexity is simplified to
\begin{align}
  \widetilde  \cO \lb \frac{c \tcL_2}{\epsilon} \times \mu \tcL_1 \times \log(t \mu) \times \log \lb \frac 1 \mu \max_{s\in[0,t]} \{t \|\bA'(s)\|,t\|\bB'(s)\|,\|\bC'(s)\|/c\}\rb\rb \;.
\end{align}

\end{proof}

The theorem below discusses the complexity of estimating an entry of the solution matrix when the history states are prepared via LCHS.

\begin{theorem}[Cost of estimating a matrix entry in the time-dependent setting (LCHS approach)]
    Let $a \ge \|\bA\|$, $b \ge \|\bB\|$, $c \ge \|\bC\|$, and $d \ge \|\bD\|$. Let $U_\bA$, 
    $U_\bB$, $U_\bC$, and $U_\bD$ be the block-encodings of $\bA(\tau)/a$, $\bB(\tau)/b$, $\bC(\tau)/c$, and $\bD(\tau)/d$, respectively, which use at most ${\rm a}\ge 0$ ancilla qubits. 
    Let $U_\phi$ and $U_\psi$ be the state preparation unitaries that map $\ket 0 \mapsto \ket \phi$ and $\ket 0 \mapsto \psi$, respectively.
    Define $\mu :=\max\{a,b\}$ and
    \begin{align}
        \cL_2 := \sqrt{\lb \frac{d}{c} \, e^{2 \int_{0}^t \rd \tau \, \xi_\bA(\tau)} + \int_0^t \rd s \, e^{2 \int_{s}^t \rd \tau \, \xi_\bA(\tau)} \rb \lb \frac{d}{c} \, e^{2 \int_{0}^t \rd \tau \, \xi_\bB(\tau)} +  \int_0^t \rd s \, e^{2 \int_{s}^t \rd \tau \, \xi_\bB(\tau)} \rb},
    \end{align}
    with $\xi_{\bA}(\tau)$ $(\xi_{\bB}(\tau))$ being a known function upper bounding the log-norm of $\bA(\tau)$ $(\bB(\tau))$ for any $\tau \in [0,t]$.
    Then we can estimate
     \begin{align}
        \bra{\phi} \bX(t) \ket{\psi} = \bra{\phi}W_{\bA}^\dagger(t,0) \bD W_{\bB }(t,0)\ket{\psi} + \int_0^t \rd s \, \bra{\phi} W^\dagger_{\bA}(t,s) \bC (s) W_{\bB}(t,s) \ket{\psi} 
    \end{align}
    within additive error $\epsilon$ with probability of success at least $2/3$ using
    \begin{align}
         \cO \lb \frac{c \cL_2}{\epsilon} \rb
    \end{align}
    queries to controlled-$U_{\phi}$, controlled-$U_{\psi}$, controlled-$U_{\bC}$, controlled-$U_{\bD}$ and their inverses, 
    \begin{align}
         \cO \lb \frac{c \cL_2}{\epsilon} \times t \mu \times \log(c \cL_2/\epsilon) \times \log \lb t \mu c \cL_2 \log(c \cL_2/\epsilon)/\epsilon \rb \rb
    \end{align}
    queries to controlled-$U_{\bA}$, controlled-$U_{\bB}$ and their inverses, and an additional
    \begin{align}
        \cO \lb \frac{c \cL_2}{\epsilon} \times \lb t \mu + \frac{\mu d}{c} \rb \times \mathrm{polylog} \lb t\mu, c \cL_2/\epsilon, \max_{s \in [0,t]} \left\{ g_t \norm{\bA'(s)}, g_t \norm{\bB'(s)}, \norm{\bC'(s)}/(c \mu) \right\}  \rb \rb
    \end{align}
    primitive gates, where $g_t := \max \{ t/a, t/b, t^2 \}$.
\label{thm:matrix_element_LCHS_time-dep}
\end{theorem}

\begin{proof}
    The proof is similar to the proof of Theorem~\ref{thm:matrix_element_time-dep}.
    In particular, to estimate $\bra{\phi} \bX(t) \ket{\psi}$, we also perform overlap estimation with history states. However, the history states have a different normalization constant compared to the history states in Theorem~\ref{thm:matrix_element_time-dep}. More specifically, let $U_{\hist, \bA}^{(\LCHS)}$ be a unitary such that
    \begin{align}
       U_{\hist, \bA}^{(\LCHS)} \ket{0}_{\mathrm{f}}\ket{0}_{\mathrm{clock}}\ket{0}_{\hist} = \ket 0_{\rm f} \frac{1}{\sqrt{\cN_{\xi_\bA}}}\lb \sum_{m=0}^{R-1} \ket{m}_{\rm clock} W_{\bA} (t,0) \ket{\phi} + \sum_{m=R}^{R+M-1} \ket{m}_{\rm clock} W_{\bA} (t,t_{m+1-R}) \ket{\phi} \rb + \ket{0^\perp} \; ,
    \end{align}
    where the $\rm f$ register flags success, and
    \begin{align}
        \cN_{\xi_A} = \lb \sum_{m=0}^{R-1} e^{2 \int_0^t \rd \tau \, \xi_A(\tau)} + \sum_{m=R}^{M+R-1} e^{2 \int_{t_{m+1-R}}^t \rd \tau \, \xi_A(\tau)} \rb \in \cO \lb R \, e^{2 \int_{0}^t \rd \tau \, \xi_A(\tau)} + a \int_0^t \rd s \, e^{2 \int_{s}^t \rd \tau \, \xi_A(\tau)} \rb,
    \end{align}  
    with $\cN_{\xi_\bA} = \lb \sum_{m=0}^{R-1} e^{2 \int_0^t \rd \tau \, \xi_\bA(\tau)} + \sum_{m=R}^{M+R-1} e^{2 \int_{t_{m+1-R}}^t \rd \tau \, \xi_\bA(\tau)} \rb$, and define $U_{\hist, \bB}^{(\LCHS)}$ similarly, using $\bB$ and $\ket \psi$ instead. These history state unitaries can be constructed using Theorem~\ref{thm:history_state_LCHS_time-dep}.
    Accordingly, we define
    the following two normalized states:
    \begin{align}
        \ket{\phi''_{\hist}} &:= \ket{0}_{\mathcal{I}}  \otimes \lb U_{\hist, \bA}^{(\LCHS)}  \ket{0}_\mathrm{f} \ket{0}_{\mathrm{clock}} \ket{0}_{\hist}\rb \\
        \ket{\mathcal{I} \psi''_{\hist}} &:=  {U}_{\mathcal{I}_{\bC,\bD}}
        \ket{0}_{\mathcal{I}}   \otimes \lb U_{\hist, \bB}^{(\LCHS)}  \ket 0_\mathrm{f} \ket{0}_{\mathrm{clock}} \ket{0}_{\hist}\rb \;,
    \end{align}
    where ${U}_{\mathcal{I}_{\bC,\bD}}$ is a unitary  block-encoding of
    \begin{align}
        \mathcal{I}_{\bC,\bD} \otimes \ketbra{0}{0}_f ,
    \end{align}
    that uses the $\cI$ register of ancillas and normalization constant $\lambda_{\cI_{\bC,\bD}}$.
    This block-encoding can be constructed from the block-encoding of $\mathcal{I}_{\bC,\bD}$ in Lemma~\ref{lem:block-encode_ICD_time-dep}. Note that
    \begin{align}
    \bra{\phi''_{\hist}}\ket{\cI \psi''_\hist} = \frac{1}{\lambda_{\cI_{\bC,\bD}} \sqrt{\cN_{\xi_{\bA}} \cN_{\xi_{\bB}}}} 
    \bra{\phi} \bX(t) \ket{\psi} \;,
    \end{align}
    with
    \begin{align}
        \lambda_{\cI_{\bC,\bD}} \sqrt{\cN_{\xi_{\bA}} \cN_{\xi_{\bB}}} \in \cO \lb c \cL_2 \rb.
    \end{align}
    
    Thus, in order to estimate $\bra{\phi} \bX(t) \ket{\psi}$ within additive error $\epsilon$, we need to perform amplitude estimation within error $\cO \lb \epsilon/\lb  c \cL_2 \rb \rb$. Additionally, we require the errors of the history states and the block-encoding of $\mathcal{I}_{\bC,\bD}$ to satisfy
    \begin{align}
        \epsilon_{\hist} \in \cO \lb \frac{\epsilon}{c \cL_2} \rb, \quad
        \epsilon'_{\cI_{\bC, \bD}}  \in \cO \lb \frac{\epsilon}{\mu \cL_2} \rb.
    \end{align}
    Therefore, we require a total of
    \begin{align}
        \cO \lb \frac{c \cL_2}{\epsilon} \rb
    \end{align}
    queries to controlled versions of $U_{\phi}$, $U_{\psi}$, $U_{\bC}$, $U_{\bD}$ and their inverses,
    \begin{align}
         \cO \lb \frac{c \cL_2}{\epsilon} \times t \mu \times \log(c \cL_2/\epsilon) \times \log \lb t \mu c \cL_2 \log(c \cL_2/\epsilon)/\epsilon \rb \rb
    \end{align}
    queries to controlled versions of $U_{\bA}$ and $U_{\bB}$ and their inverses.
    For the number of additional gates, we have the following contribution from Theorem~\ref{thm:history_state_LCHS_time-dep} for the history state preparations:
    \begin{align}
        \cO \lb \frac{c \cL_2}{\epsilon} \times \lb t\mu + \frac{\mu d}{c}\rb \times \mathrm{polylog} \lb t\mu, c \cL_2/\epsilon, g_t \max_{s \in [0,t]} \left\{ \norm{\bA'(s)}, \norm{\bB'(s)} \right\} \rb \rb.
    \end{align}
    The contribution to the gate complexity of $U_{\cI_{\bC,\bD}}$ can be seen from Lemma~\ref{lem:block-encode_ICD_time-dep}. In particular, we have the following contribution:
    \begin{align}
        \cO \lb \frac{c \cL_2}{\epsilon} \times \log(\cL_2 \max_s \|\bC'(s)\|/(\mu \epsilon)) \rb.
    \end{align}
    The overall additional gate complexity therefore scales like
    \begin{align}
        \cO \lb \frac{c \cL_2}{\epsilon} \times \lb t \mu + \frac{\mu d}{c} \rb \times \mathrm{polylog} \lb t\mu, c \cL_2/\epsilon, \max_{s \in [0,t]} \left\{ g_t \norm{\bA'(s)}, g_t \norm{\bB'(s)}, \norm{\bC'(s)}/(c \mu) \right\}  \rb \rb.
    \end{align}
\end{proof}

In Table~\ref{tab:compare_time-dep} we summarize the cost of estimating $\bra{\phi} \bX(t) \ket{\psi}$ within additive error $\epsilon$ in the time-dependent setting using either quantum linear system solvers or LCHS for the history state preparation. Number of queries to the unitaries also include their inverses.

\begin{table}[t]
    \renewcommand{\arraystretch}{1.7}
    \setlength{\tabcolsep}{6pt}
    \centering
    \begin{tabular}{|c|c|c|}
    \hline
         & Linear systems approach & LCHS approach \\
         \hline
         \# queries to $U_{\phi}$ and $U_{\psi}$ & $\cO \lb \frac{ c \tcL_2}{\epsilon} \times \mu \tcL_1 \times \log (c \tcL_2/\epsilon) \rb$  & $\cO \lb \frac{c \cL_2}{\epsilon} \rb$\\
         \# queries to $U_{\bA}$ and $U_{\bB}$ & $\begin{aligned}[t] & \textstyle \cO \Big( \frac{ c \tcL_2}{\epsilon} \times \mu \tcL_1 \times \log \lb \frac{c \tcL_2}{\epsilon} \rb \\ & \textstyle \quad \times \log \lb \frac{t \mu  c^2 \tcL_2 \tcL_1}{d\epsilon} \rb\Big) \end{aligned}$ & $\begin{aligned}[t] & \textstyle \cO \Big( \frac{c \cL_2}{\epsilon} \times t \mu \times \log \lb \frac{c \cL_2}{\epsilon} \rb \\ & \textstyle \quad \times \log \lb \frac{t \mu c \cL_2 \log(c \cL_2/\epsilon)}{\epsilon} \rb \Big) \end{aligned}$ \\
         \# queries to $U_{\bC}$ and $U_{\bD}$ & $\cO \lb \frac{c \tcL_2}{\epsilon}\rb$ & $\cO \lb \frac{c \cL_2}{\epsilon} \rb$ \\
         \hline
    \end{tabular}
    \caption{Comparison of the cost for estimating $\bra{\phi} \bX(t) \ket{\psi}$ within additive error $\epsilon$ in the time-dependent setting using either quantum linear system solvers or LCHS for the history state preparation; see Theorems~\ref{thm:matrix_element_time-dep} and \ref{thm:matrix_element_LCHS_time-dep} for more details.
    }
\label{tab:compare_time-dep}
\end{table}

 \section{Derivation of quantum Lyapunov equation for non-interacting fermions}
 \label{app:QLE}

A direct application of our quantum algorithm
is that for solving differential Sylvester equations~\cite{behr2019solution},
where a prominent example is the well known (differential) Lyapunov equation. 
While these equations appear in many fields like control theory, applied math, and physics, here we focus on a derivation and an application to simulating open quantum system dynamics.
Later, we will apply this analysis to large free fermion systems interacting with a bath at inverse temperature $\beta \ge 0$. 
In this way, we are able to extend the results in Ref.~\cite{somma2025shadowhamiltoniansimulation} on ``shadow Hamiltonian simulation'' and Refs.~\cite{stroeks2024solvingfreefermionproblems,chen2024quantum}
for thermal states of free fermion systems
to the case of stochastic noise and dissipative dynamics.

Our starting point is the differential Lyapunov equation for quantum systems that are experiencing dissipation and quantum fluctuations.
In general, when the system interacts with a bath, 
the equation of motion for an operator $O(t)$
in the Heisenberg picture is generalized as
\begin{align}
\label{eq:QLE}
    \frac{\rd}{\rd t}O(t) = \ri [H_S, O(t)] + L(O(t)) + \eta(t) \;,
\end{align}
where $H_S$ is the system's Hamiltonian that produces unitary dynamics, $L(.)$ is the dissipation term that we will assume to be linear, and $\eta(t)$ is an operator
that accounts for the quantum noise due to the bath. This differs from the Lindblad equation
in the Heisenberg picture (i.e., the ``adjoint'' Lindblad equation), in that thus far we are tracking the operator and not its expectations.

This noise operator must satisfy certain properties
in order to preserve the commutation relations
between the time-dependent operators $O(t)$; in contrast with the classical case, quantum operators cannot ``decay'' in time.
Moreover, within the context of the fluctuation-dissipation theorem, in order for the system to evolve towards thermal equilibrium, both $L$ and $\eta$ must satisfy 
a consistency condition.

For concreteness, we consider a quantum system of free fermions that is also interacting with a bath 
of free fermions. The following derivation
uses standard assumptions in physics; see Ref.~\cite{purkayastha2022lyapunov} for a related derivation.
Let the system-bath Hamiltonian be $H=H_S + H_B + H_{\rm int}$, where $H_S$ is the system's Hamiltonian, $H_B$ is the bath's Hamiltonian, and $H_{\rm int}$ models the interactions. These are:
\begin{align}
    H_S &:= \sum_{j,k=1}^{N_s} \alpha_{jk} c^\dagger_j c^{}_k\;,\\
    H_B & := \sum_{l=1}^{N_b} \epsilon_l b^\dagger_l b^{}_l \;,\\
    H_{\rm int}& :=\sum_{j=1}^{N_s} \sum_{l=1}^{N_b} \lb \nu_{jl} c^\dagger_j b^{}_l - \overline {\nu_{jl}}  c^{}_j b^{\dagger}_l \rb \;,
\end{align}
where $c^\dagger_j$ and $c^{}_j$ are the fermionic creation and annihilation operators for the system that satisfy the anticommutation relations $\{c^\dagger_j, c^{}_k\}=\delta_{j,k}$ and $\{c^\dagger_j, c^{\dagger}_k\}=\{c^{}_j, c^{}_k\}=0$. The coefficients $\alpha_{jk} \in \mathbb C$
define the $N_s \times N_s$ matrix $\bA$ which is Hermitian.
Similarly, $b^\dagger_l$ and $b^{}_l$ are the fermionic creation and annihilation operators for the bath, where $\epsilon_l \in \mathbb R$ since the Hamiltonian $H_B$ is already diagonal. Also, $\{c^\dagger_j,b^\dagger_l\}= \{c^\dagger_j,b^{}_l\}=\{c^{}_j,b^{}_l\}=0$, since the system and bath fermionic modes are distinct.
Our end goal is to write a differential equation
for the expectations $\langle c^\dagger_j c_k^{}(t)\rangle:=\tr(c^\dagger_j c^{}_k \rho(t))$,
when the initial state is $\rho(0)=\rho$, in the form of a Lyapunov differential equation. These expectations define an $N_s \times N_s$ matrix $\bX(t)$, referred to as the covariance matrix.

Since the algebra of fermionic operators is conserved with quadratic Hamiltonians, it is simple to write the equations of motion for these from Heisenberg's equations:
\begin{align}
    \frac{\rd}{\rd t}c^{}_j(t) &= \ri [H,c^{}_j(t)]= \ri \lb -\sum_{k=1}^{N_s} \alpha_{jk} c^{}_k(t) - \sum_{l=1}^{N_b} \nu_{jl} b_l(t) \rb \; , \\
     \frac{\rd}{\rd t}b^{}_l(t) &=\ri [H,b^{}_l(t)]=\ri \lb -\epsilon_l b^{}_l (t)- \sum_{j=1}^{N_s} \overline{\nu_{jl}} c^{}_j(t) \rb \;.
\end{align}
These can be easily integrated and, in particular,
\begin{align}
    b_l(t) =  e^{-\ri t \epsilon_l} b_l(0)- \ri \sum_{j=1}^{N_s}  \overline{\nu_{jl}} e^{-\ri t \epsilon_l} \int_0^t \rd s \; e^{\ri s \epsilon_l} c^{}_j(s) \;,
\end{align}
which gives
\begin{align}
\label{eq:OQS}
   \frac{\rd}{\rd t}c^{}_j(t) =-\ri  \sum_{k=1}^{N_s} \alpha_{jk} c^{}_k(t)  - \sum_{k=1}^{N_s} \sum_{l=1}^{N_b} \nu_{jl}\overline{\nu_{kl}}  e^{-\ri t \epsilon_l} \int_0^t \rd s \; e^{\ri s \epsilon_l}c^{}_k(s)-\ri \sum_{l=1}^{N_b} \nu_{jl} e^{-\ri t \epsilon_l}b_l(0)\;.
\end{align}
This last equation can also be integrated. It is not readily in the form of Eq.~\eqref{eq:QLE} since the dissipative term
contains an integral over all times.

To this end, we make  a standard approximation in open quantum systems, which is essentially a Markov approximation where the kernel of the integral is proportional to the Dirac delta distribution $\delta(s-t)$. 
This occurs when the bath contains an infinite number of fermionic modes ($N_b \gg N_s$), so that we can approximate
\begin{align}
  \sum_{l=1}^{N_b} \nu_{jl}\overline{\nu_{kl}}  e^{-\ri (t-s) \epsilon_l}     \mapsto 
  \int_{-\infty}^\infty \rd \omega \; s(\omega) V_{j}(\omega)\overline{V_k(\omega)} e^{-\ri (t-s) \omega} \;,
\end{align}
where $s(\omega)$ is the density of states for the bath at frequency (energy) $\omega$,
and $s(\omega) V_{j}(\omega)\overline{V_k(\omega)}$ denotes the coupling. That is, the $V_j(\omega)$'s are the $v_{jl}$'s in the continuum where $\omega \in \mathbb R$.
If we further assume these couplings to be constant,
then
\begin{align}
  \int_{-\infty}^\infty \rd \omega \;  s(\omega) V_{j}(\omega)\overline{V_k(\omega)} e^{-\ri (t-s) \omega} \mapsto   \gamma_{jk}\delta(t-s) \;.   
\end{align}
Here, $ \gamma_{jk}:= 2\pi  s(\omega) V_{j}(\omega)\overline{V_k(\omega)}$ is independent of $\omega$ under the assumption. Nevertheless, under realistic physical assumptions where the system interacts locally with the bath, the coefficients 
$ \gamma_{jk}$ could be sparse. This allows us to generalize the analysis by introducing a Hermitian and positive semidefinite matrix $\bGamma$ of dimension $N_s \times N_s$ that encodes the interactions and the corresponding dissipative coefficients
$\gamma_{jk}$. In the prior example, these are $ \gamma_{jk}:= 2\pi  s(\omega) V_{j}(\omega)\overline{V_k(\omega)}$, 
but the matrix $\bGamma$ can be sparse in general. This occurs, for example, when each fermionic mode -- or a subset of neighboring modes -- interacts with a bath that is localized within the mode's region of space.

The second term in Eq.~\eqref{eq:OQS} can now be replaced by
$  - \sum_{k=1}^{N_s} \gamma_{jk} c_k(t)$, and the equation
takes now the form of Eq.~\eqref{eq:QLE}: 
\begin{align}
L(c^{}_j(t))& = - \sum_{k=1}^{N_s} \gamma_{jk} c_k(t) \; , \\
    \eta_j(t) &= -\ri \int \rd \omega \; \sqrt{s(\omega)} V_j(\omega) e^{-\ri t \omega}b_\omega(0)\;.
    \end{align}
In summary, 
\begin{align}
\label{eq:QLE2}
    \frac{\rd}{\rd t} c_j(t) =-\ri \sum_{k=1}^{N_s} \alpha_{jk} c_k(t) -   \sum_{k=1}^{N_s} \gamma_{jk} c_k(t) + \eta_j(t)\;.
\end{align}

Consider now the Hermitian covariance matrix $\bX(t)$ of entries $\langle c^\dagger_j c^{}_k(t)\rangle$, where the expectations are over an initial system-bath state $\rho(0)=\rho$. (Here $c^\dagger_j c^{}_k(t)\equiv c^\dagger_j(t) c^{}_k(t)$.) We can use Eq.~\eqref{eq:QLE2}
and, using linearity,
\begin{align}
\label{eq:QLE3}
  \frac{\rd}{\rd t}  \langle c^\dagger_j c^{}_k(t)\rangle = \sum_{k'=1}^{N_s}  (-\ri \alpha_{kk'}-  \gamma_{kk'}) \langle c^\dagger_j c^{}_{k'}(t)\rangle+ \sum_{j'=1}^{N_s} (\ri \overline{\alpha_{jj'}}-  \overline{\gamma_{jj'}}) \langle c^\dagger_{j'} c^{}_k(t)\rangle + \langle c^\dagger_j \eta_k(t)\rangle
  + \langle \eta^\dagger_j c^{}_k(t)\rangle \;.
\end{align}
Let $\bB = -\ri \bA - \bGamma$ 
and let $\bC$ be the $N_s \times N_s$ Hermitian matrix of entries 
$\langle c^\dagger_j \eta_k(t)\rangle
  + \langle \eta^\dagger_j c^{}_k(t)\rangle$. Then, in compact form,
  Eq.~\eqref{eq:QLE3} is
  \begin{align}
  \label{eq:QLE4}
      \frac{\rd}{\rd t} \bX(t) = \bB^\dagger \bX(t)
      + \bX(t) \bB + \bC \;.
  \end{align}
  This has the form of the equations that can be solved with our quantum algorithm.\ Moreover, $\bB$ has a non-positive log norm and it is sparse if the interactions in the system are geometrically local. This example generalizes Refs.~\cite{somma2025shadowhamiltoniansimulation,stroeks2024solvingfreefermionproblems} that study the case of unitary dynamics in equilibrium. A similar equation can be obtained from the standard Lindblad equation for $\rho(t)$ after computing the expectations $\langle c^\dagger_j c^{}_k\rangle$.

  It is evident in our derivation 
  that the dissipative and noise terms
  are not independent. 
  That is, the matrices $\bGamma$ and $\bC$ must satisfy a consistency condition. In the important case of thermal equilibrium, this condition is set by the fixed point of $\bX(t)$. Considering the thermal Gibbs state of the system $\rho_\beta = e^{-\beta H_S}/\tr(e^{-\beta H_S})$, the corresponding covariance matrix $\bX_\beta$ is the well-known Fermi-Dirac distribution:
  \begin{align}
    \bX_\beta = \frac 1 {\one + e^{\beta \bA}}\;.  
  \end{align}
  Considering this as the fixed point of Eq.~\eqref{eq:QLE4}, we obtain
  \begin{align}
      0 = \bB^\dagger \bX_\beta
      + \bX_\beta \bB + \bC  \implies \bC =   ( \bX_\beta \bGamma  +    \bGamma \bX_\beta) \;.
  \end{align}

This is a manifestation of the fluctuation-dissipation theorem, for which the dissipation (microscopic) has to be compensated by the noise and fluctuations induced by the bath (macroscopic), to reach thermal equilibrium.

While we have considered free fermions that are number conserving for simplicity, it is simple to extend the analysis to all quadratic fermionic Hamiltonians. In such cases, we are also left with a Lyapunov differential equation of the form of Eq.~\eqref{eq:QLE4}.

\section{Classical Krylov-subspace algorithms for linear matrix differential equations}
 \label{app:classicalalgs}

We describe classical algorithms for computing the entries of $\bX(t)$, in the case where the matrices are sparse. 
While we show the basic ideas that determine the proper asymptotic complexities, we note that 
there is vast literature for classical algorithms and the following techniques
might be specialized to particular problem instances (e.g., by leveraging other problem structures), potentially giving improved results.
See Ref.~\cite{simoncini2016computational}.

The first algorithm we discuss is based on Krylov subspaces~\cite{hached2018computational,behr2019solution} to output
the entry $\bra j \bX(t) \ket k$, where $\ket j$ and $\ket k$
are standard basis states. For this case we assume the matrices to be ${\rm S}$-sparse and time-independent, so that
\begin{align}
    \bX (t) = e^{t \bA^\dagger }\bD e^{t \bB } + \int_0^t \rd s \; 
    e^{(t-s) \bA^\dagger}\bC e^{(t-s) \bB } \;.
\end{align}
We then define two Krylov subspaces:
\begin{align}
\label{eq:Krylovsubspaces}
   K_{\bA}(m):={\rm Span} \{\ket j, \bA  \ket j, \ldots , 
   (\bA )^{m-1} \ket j \} , \; \; 
    K_{\bB}(m):={\rm Span} \{\ket k, \bB \ket k, \ldots , 
   (\bB)^{m-1} \ket k \} \;.
\end{align}
 The dimension $m$ is a parameter that can be adjusted, 
 and ideally $m \ll N$. In practice, 
 one generates a set of $m$ orthogonal vectors
 via a method like Gram-Schmidt (e.g., the Lanczos method): if $\ket{v_i}$
 is the vector at the $i^{\rm th}$ step, we apply for example $\bA$ to generate $\ket{w_{i+1}}=\bA  \ket{v_i}$
 and then subtract all the components aligned with the orthogonal vectors $\ket{v_0},\ldots,\ket{v_i}$ to obtain $\ket{v_{i+1}}$.
 The complexity of this step increases with $i$ and can be dominant. For ${\rm S}$-sparse matrices, the generation of each orthogonal basis
 has cost $\cO(m^2  N + {\rm S} m N)$: the factor ${\rm S}N$ is from a single matrix-vector multiplication, and there are $m$ of them, 
 and the factor $m^2$ is due to performing $\cO(m^2)$ inner products
 for orthogonalization.

 By projecting into these subspaces, we are now able to replace
 \begin{align}
     \bA \mapsto \tilde {\bA}= V^\dagger \bA V \in \mathbb C^{m \times m}, \;  \bB \mapsto \tilde {\bB}= W^\dagger \bB W
     \in \mathbb C^{m \times m} \;,
 \end{align}
 which are matrices of much smaller dimension. Here, $V \in \mathbb C^{N \times m}$ and $W \in \mathbb C^{N \times m}$
 contain the orthogonal vectors of dimension $N$ generated by the Krylov subspaces $K_\bA(m)$ and $K_\bB(m)$, respectively,  in each of their columns. 
 Constructing $\tilde {\bA}$ and 
  $\tilde {\bB}$ does not incur any additional costs since they are a byproduct of the orthogonalization step. That is, 
  for orthogonalization, we already had to apply $\bA$
  to the columns of $V$ and compute the inner products, and the same for $\bB$.
  
  The remaining task then is to compute  $\bra j \bX(t) \ket k$
  by working with the projected problem. One approach
  is to replace each $\bra j e^{s \bA^\dagger}$ by the corresponding 
  $\bra j e^{s \tilde{\bA}^\dagger}$ of dimension $m$, and similarly
  for $e^{s \bB}\ket k$. This has cost $\cO(m^3)$ or better, 
  which is the cost of constructing the exponentials. 
  We can also project, for example, $\bC$ as $\bC \mapsto \tilde{\bC}=V^\dagger \bC W$ also with cost $\cO({\rm S} mN+m^2N)$: the term $smN$ comes from computing $\bC W$, and the term $m^2N$
  comes from acting with $V^\dagger$.
  This projection would allow us to compute the integral by now working with $m \times m$ matrices. In particular, if we can approximate the integral via quadratures using $R$ terms, the additional cost would be $\cO(R m^3)$, where the factor $m^3$
  is from dense matrix-vector multiplication with dimension $m$.
  
 In summary, the cost of this procedure is 
  $\cO({\rm S} m N + m^2 N + R m^3)$. In practice this approach is very efficient if $m$ is fairly small (e.g., much smaller than $N$), and for sparse matrices the dominant term is often $m^2 N$ in the orthogonalization step. For a controlled approximation, it is expected that $m \sim t \max\{\|\bA\|,\|\bB\|\}$, up to mild corrections due to the error, since the problem's complexity increases with the evolution time.  (In a similar manner, it is expected that $R \sim t \max\{\|\bA\|,\|\bB\|\}$ using quadratures.) Hence, using a ``lightcone'' argument, where $N \sim m^{D}$ for a system in a lattice in $D \ge 1$ spatial dimensions, the overall cost of this approach would be
  $\widetilde \cO(m^{2+D})$. Also, the memory requirements are dominated by those of storing the matrices $V$ and $W$, and this is $\cO(mN)=\cO(m^{1+D})$.

  The memory requirement  of this approach might be improved to $\widetilde \cO(N)$  and the runtime improved to $\widetilde \cO(mN)$, but not further, using a more clever way of computing the integral. 
  While we did not find this approach applied directly to our matrix equations in the literature,
  it shares some features of the ``restarted Krylov'' method~\cite{eiermann2006restarted}.
It is then plausible that a linear scaling in the time-space volume
  can be achieved with classical algorithms for this problem, although with some additional overheads due to the integrals.
   This is important at the time of comparing the runtimes of quantum and classical approaches, especially in the case where the expected quantum speedups are polynomial. 
  Here we sketch the idea for such an improved method.
   We can change variables and write
  \begin{align}
     \int_{0}^{t}  \rd s \; e^{(t-s)\bA^\dagger} \bC e^{(t-s)\bB}=
      \int_{0}^{t}  \rd s \; e^{s\bA^\dagger} \bC e^{s\bB}\;.
  \end{align}
We can split this integral as $\int_0^t \mapsto \int_0^r + \int_r^{2r}+ \ldots$, where we can choose, say, 
  $r \sim 1/\max \{\|\bA\|,\|\bB\|\}$. The number of smaller integrals is then $L=t/r\sim t \max \{\|\bA\|,\|\bB\|\}$. 
  These smaller integrals are related with each other and, in particular,
  \begin{align}
  \bra j  \int_{0}^{t}  \rd s \; e^{s\bA^\dagger} \bC e^{s\bB}  \ket k =
  \sum_{l=0}^{L-1} (\bra j e^{lr \bA^\dagger}) \left( \int_0^{r}\rd s\; e^{s\bA^\dagger} \bC e^{s\bB}\right) 
  (e^{lr \bB}\ket k)\;.
  \end{align}

  Using the prior approach to compute the entry of the first integral ($l=0$) we can construct two Krylov subspaces $K_\bA^0$ and $K_\bB^0$ where $m \mapsto m' = \widetilde \cO(1)$. This works
  because we are only evolving for a short time $r$.
  The second integral ($l=1$) is now related to the first one, but we are interested in a different matrix element,
  that corresponding to $\bra j e^{r\bA^\dagger}$ and $e^{r\bB}\ket k$.
  We can then generate the new Krylov subspaces $K_\bA^1$ and $K_\bB^1$ by letting the first vectors be $e^{r\bA} \ket j$ and $e^{r\bB} \ket k$, respectively. 
  These vectors will not be sparse in general. Indeed, we can choose these as a linear combination of the vectors in the prior Krylov subspaces $K_\bA^0$ and $K_\bB^0$
  by using, for example, a truncated Taylor series of the exponential. At this point we can delete the prior Krylov subspace from memory, and only work with the new ones:
  \begin{align}
   K_{\bA}^1&:={\rm Span} \{(e^{r\bA} \ket j), \bA (e^{r\bA} \ket j), \ldots , 
   (\bA)^{m'-1} (e^{r\bA} \ket j) \} , \\ 
    K_{\bB}^1&:={\rm Span} \{(e^{r\bB} \ket k), \bB(e^{r\bB} \ket k), \ldots , 
   (\bB)^{m'-1} (e^{r\bB} \ket k) \} \;,
\end{align}
  where again $m'=\widetilde \cO(1)$. Note that the new subspaces are not necessarily orthogonal to $K_\bA^0$ and $K_\bB^0$.

  We can repeat this approach 
  until the last integral ($l=L-1$). Since $m'=\widetilde \cO(1)$ for all the generated subspaces, the memory requirements are then $\widetilde \cO(N)$.
  Accordingly, the runtime is now $\widetilde \cO(L \times ({\rm S}N))=\widetilde \cO(t \max \{\|\bA\|,\|\bB\|\}{\rm S} N)$,
  being linear in the space-time volume for sparse matrices.
  If we define $m \sim t \max \{\|\bA\|,\|\bB\|\}$, and assuming the relevant volume is $N \sim m^D$ due to the lightcone argument, then the runtime is $\widetilde \cO(m^{1+D})$
  and the memory requirements are $\widetilde \cO(m^D)$, i.e., linear in $N$.

  We note, however, that we would expect the approximation errors
  of this approach to grow as we increase the number of steps. We have not considered this accumulation of errors in the analysis
  but it appears they can be controlled by choosing $m'$ properly.
  Nevertheless, the $L$ Krylov subspaces generated can still be spanned
  by the original ones in Eq.~\eqref{eq:Krylovsubspaces}, since all vectors can still be obtained by acting with $\bA$
  and $\bB$ on $\ket j$ and $\ket k$, respectively. At the same time, we note that these $L$ subspaces generated are not necessarily orthogonal with each other, making it different from the former approach. 
  These observations raise the question of how much more efficient and accurate this improved approach is in practice. 
  See Ref.~\cite{hached2018computational}
  for other numerical improvements based on different Krylov subspaces.

  We also note that related classical algorithms
  might be constructed via the Kernel Polynomial Method (KPM)~\cite{weisse2006kernel}. Rather than generating the Krylov subspaces as with the prior approach, the KPM
  constructs vectors corresponding to the action of Chebyshev polynomials of $\bA$ and $\bB$ acting on $\ket j$ and $\ket k$.
  Just like with the Krylov subspace method, the runtime of KPM is at least linear in the space-time volume (by splitting the integral as discussed) and memory requirements that are least linear in the dimension.

  The cases of time-dependent or non-sparse matrices 
  will be significantly more costly for these classical algorithms, 
  with the potential of much improved asymptotic quantum speedups for such instances.

\section{Query lower bound and proof of Thm.~\ref{thm:lowerbound}}
 \label{app:lowerbound}

Consider our problem where the goal is to solve for
$\bX(t)$ satisfying
\begin{align}
    \frac{\rd}{\rd t}\bX(t) = \bA^\dagger \bX(t)+ \bX(t) \bB + \bC \;.
\end{align}
We have presented a quantum algorithm 
that computes the entries of $\bX(t)$ within error $\epsilon \ge 0$ assuming
access to block-encodings of $\bX(0)=\bD$, $\bA$, $\bB$, $\bC$.  We now prove the lower bound in Thm.~\ref{thm:lowerbound}.

As mentioned, the proof of Thm.~\ref{thm:lowerbound}
is based on a reduction to quantum phase estimation (QPE), and a decision version of that problem. To this end we adapt the lower bounds $\Omega(\frac 1 \delta |\log (1-p)|)$ for QPE within error $\delta$ and success probability $p$ established in 
Ref.~\cite{mande2023tight}; the following result is adapted from Claim 5.2 of
that work.
\begin{claim}[Query lower bound on decision QPE]
\label{claim:lowerbound}
 Let $U=\exp(\ri \theta \ketbra{0})$
 be a diagonal unitary, where either 
  i) $\theta=\delta$ or ii) $\theta \in (2\delta, \pi-2\delta]$,
  for some given $\delta$ satisfying $\pi/16 \ge \delta >0$.
   Then, for success probability $  p  \in (3/4,1)$, any quantum algorithm that solves this decision phase estimation problem requires $\Omega(\frac 1 \delta |\log(1-p)|)$ uses of $U$ and $U^{-1}$.
\end{claim}
\begin{proof}
The main modification with the problem in 
Ref.~\cite{mande2023tight} is the domain for $\theta$,
since that work solves the decision problem 
where either $\theta =0$ or $\theta \notin [-\delta,\delta]\mod 2\pi$, for $\pi/4 \ge \delta >0$.
Nevertheless, the same lower bound applies for two simple reasons.
First, that same proof is easily extended to the case where
either $\theta =\pm \delta$ or $\theta \in (2\delta,\pi-2\delta] \cup [\pi+2\delta,2 \pi-2\delta)$, for $\pi/16 \ge \delta >0$
from a simple redefinition of the domain for $\theta$. That is, the QPE algorithm
has to produce a trigonometric polynomial
where the measure $\mu$ in Thm. 2.7 of that work is such that $\mu \ge 2 \pi-s$, and $s=8 \delta$ in our case (also $s\le \pi/2$). This implies that the number of calls to $U$ and $U^{-1}$
for solving this decision problem is lower bounded as $\Omega(\frac 1 \delta |\log(1-p')|$, for success probability $p' > 1/2$.

Second, we note that a quantum algorithm that solves the problem in this claim, where either $\theta=\delta$ or $\theta \in (2\delta,\pi-2\delta]$, can be used twice for solving the prior decision problem. For example, in the first execution we can determine whether $\theta \notin (2\delta,\pi-2\delta]$ or $\theta \neq \delta$ with probability $p$, 
and in the second execution we can
determine whether 
$\theta \notin [\pi+2\delta,2\pi-2\delta)$
or $\theta \neq -\delta$ also with probability $p$. 
 For this second run we need to transform $\theta \mapsto -\theta$, but this is simply done since we have access to $U$ and $U^{-1}$; that is, we only need to replace $U \leftrightarrow U^{-1}$. 
 
 For success probability $p'$ in solving the former decision problem, it suffices to choose $p=1-(1-p')/2$ in each run of the algorithm for this claim. Hence, that lower bound $\Omega(\frac 1 \delta |\log(1-p')|)$ translates to the 
lower bound $\Omega(\frac 1 \delta |\log(1-p)|)$ for the current claim. (Since $p' \in (1/2,1)$ in Ref.~\cite{mande2023tight},
here we assumed $p \in (3/4,1)$.)

\end{proof}

We now want to use this lower bound in query complexity for QPE
to prove a lower bound in the query complexity for estimating an entry of $\bX(t)$. We do this via a simple problem reduction from decision QPE.

\begin{lemma} 
\label{lem:reduction}
    Consider the instances where $\bA=-\sin \theta \ketbra 0-\sum_{n>0}\ketbra n$
    and assume either i) $\theta=\delta$ or ii) $\theta \in (2\delta , \pi-2\delta]$, for some given $\delta$ satisfying $\pi/16 \ge \delta > 0$. Let $t>0$ and $\cL_\theta=\int_0^t \rd s \|e^{s\bA}\|$. 
    Then, computing $\bra{0}\bX(t)\ket{0}$
    within error at most $\epsilon_\delta=3 \cL_\delta^2 \sin \delta / 100$
    allows one to determine whether $\theta$ belongs to
    case i) or case ii).
\end{lemma}
\begin{proof}
For these instances, the solution satisfies
\begin{align}
    \bX(t) = \int_0^t \rd s \; e^{(t-s)\bA}\;,
\end{align}
and hence
\begin{align}
  \bra{0}  \bX (t) \ket{0}  = \bra{0} \int_0^t \rd s \; e^{(t-s) \bA} \ket{0} 
  & = \frac{1- e^{- t \sin \theta}}{ \sin \theta}.
\end{align}
For this problem, we also have $\cL_\theta=\int_0^t \rd s \; e^{-(t-s) \sin \theta}=\bra 0 \bX(t) \ket 0$.
We then need to show that computing $\bra{0}  \bX (t) \ket{0}$ within error $\epsilon_\delta$ allows us to determine one case or the other.
Note that $\cL_\theta$ decreases monotonically with $\theta$ in $\theta \in (0,\pi/2)$ and $\cL_\theta =\cL_{\pi-\theta}$, it will suffice to establish error bounds for $\theta=\delta$
and $\theta=2\delta \le \pi/8$ only. 
If the gap between those two entries is at least $2 \epsilon_\delta$,
then estimating the entry within $\epsilon_\delta$ will suffice.

\begin{itemize}
    \item Case $t \le 1/\sin (2\delta)$.
We obtain
\begin{align}
   \cL_{2\delta} & = \frac{1- e^{- t \sin (2\delta)}}{ \sin(2 \delta)} \\
      & \le \frac 1 { \sin (2\delta)} \left( t \sin (2\delta) - \frac{(t \sin (2\delta))^2} e\right) \\
      & = t - \frac {t^2 \sin (2\delta)} e  \;,
\end{align}
where we used $t \sin(2\delta) \le 1$ in the second line.
Also,
\begin{align}
   \cL_{\delta} & = \frac{1- e^{- t \sin \delta}}{ \sin  \delta} \\
      & \ge \frac 1 { \sin \delta} \left( t \sin \delta - \frac{(t \sin \delta)^2} 2\right) \\
      & = t - \frac {t^2 \sin \delta} 2  \;,
\end{align}
where we used $t \sin \delta \le 1$ in the second line.
These give the gap
\begin{align}
    \cL_\delta -\cL_{2\delta}& \ge t^2 (\sin(2\delta)/e - \sin(\delta)/2) \\
    & \ge (\cL_\delta)^2 (\sin(2\delta)/e - \sin(\delta)/2)\\
    & = (\cL_\delta)^2 \lb \frac 2 e \sin \delta \cos \delta -\frac 1 2 \sin \delta \rb \\
    & \ge (\cL_\delta)^2 \frac 1 5 \sin \delta \\
    & \ge 2 \epsilon_\delta \;.
\end{align}
We used $\cos \delta \ge \cos (\pi/16)$
so that $2 \cos \delta/e-1/2 >1/5$,
and $t \ge \cL_\delta$.

\item Case $t \ge 1/\sin (2\delta)$.
We will use $2 \sin \delta \ge \sin (2\delta) \ge 2 \cos(\pi/16)\sin(\delta) \ge 1.96 \sin(\delta)$.
We obtain
\begin{align}
     {\cL_{2\delta}}  &= \frac{1- e^{- t \sin (2\delta)}}{   \sin  (2\delta)} \\
    & \le \frac 1 {1.96}\frac{1- e^{- t \sin (2\delta)}}{   \sin  (\delta)}  \\
    & \le \frac 1 {1.96}\frac{1- e^{- 2t \sin (\delta)}}{   \sin  (\delta)} \\
    & = \frac 1 {1.96}\frac{(1- e^{- t \sin (\delta)})(1+ e^{- t \sin (\delta)})}{   \sin  (\delta)}\\
    & = \frac 1 {1.96}(1+ e^{- t \sin (\delta)}) \cL_\delta \\
     & \le \frac 1 {1.96}(1+ e^{- \frac {\sin (\delta)} {\sin(2\delta)} }) \cL_\delta \\
     & \le \frac 1 {1.96}(1+ e^{- \frac {1} {2} }) \cL_\delta \\
     & \le 0.82 \cL_\delta \\
    & \implies \cL_\delta \ge \frac 6 5 \cL_{2\delta}\;.
\end{align}
Next,
\begin{align}
  \cL_{2\delta}   & \ge \frac{1-e^{-1}}{\sin (2\delta)} \\
  & \ge  \frac{1-e^{-1}}{2 \sin (\delta)}\\
  & \ge 0.31 \frac 1 {\sin \delta} \\
  & \ge  0.31 \frac{(1-e^{-t \sin \delta})^2}{\sin \delta}
  \\
  & = 0.31 (\cL_\delta)^2 \sin \delta \;.
\end{align}
Last,
\begin{align}
    \cL_\delta - \cL_{2\delta} &\ge \frac 1 5 \cL_{2\delta} \\
    & \ge \frac{0.31}5 (\cL_\delta)^2 \sin \delta \\
    & \ge 2 \frac 3 {100} (\cL_\delta)^2 \sin \delta \\
    & = 2 \epsilon_\delta \;.
\end{align}

\end{itemize}

It follows that computing  $ \bra{0}  \bX (t) \ket{0}$  within error at most $\epsilon_\delta$ for each instance solves the decision problem.

\end{proof}

Last, we use Claim~\ref{claim:lowerbound} and Lemma~\ref{lem:reduction} to prove Thm.~\ref{thm:lowerbound}.
We do this by considering a decision version of that problem.

 \begin{lemma} 
 \label{lem:lb2}
    Consider the instances where $\bA=-\sin\theta \ketbra 0 - \sum_{n \ne 0}\ketbra n$
    and assume either i) $\theta=\delta$ or ii) $\theta \in (2\delta , \pi-2\delta)$, for some given $\delta$ satisfying $\pi/16 \ge \delta > 0$. Let $t > 0$ and $\cL_\theta = \int_0^t \rd s \; \|e^{s\bA}\|$. 
    Then,  any quantum algorithm that computes $\bra{0}\bX(t)\ket{0}$
    within error at most $\epsilon_\delta = 3\cL_\delta^2 \sin \delta/ 100$
    and success probability $p \in(3/4,1)$
    requires $\Omega (\frac {\cL_\theta^2} {\epsilon_\delta} |\log(1-p)|)$ uses of the block-encoding for $\bA$.
\end{lemma}

\begin{proof}
    A possible block-encoding for $\bA$ could be constructed using the unitaries $U$ and $U^{-1}$ due to a simple LCU:
    \begin{align}
        -\sin \theta \ketbra 0 = \frac {1}{2\ri}(U-U^{-1})\;.
    \end{align}
    Hence, the block-encoding for $\bA$
    uses $U$ and $U^{-1}$ once.
    
    If the quantum algorithm for computing $\bra{0}\bX(t)\ket{0}$ requires $Q$ uses of the block-encoding for $\bA$, it can be executed with at most $2Q$ uses of $U$ and $U^{-1}$. This would allow us to solve the prior estimation problem, and further the decision QPE problem in Claim~\ref{claim:lowerbound}.
    The lower bound $Q = \Omega (\frac 1 \delta |\log(1-p)|)$ then follows.

    Next note
    \begin{align}
        \frac{\cL_\theta^2}{\epsilon_\delta} \le \frac{\cL_\delta^2}{\epsilon_\delta}
        = \frac {100}{3\sin \delta} & \le \frac {100}{3 \delta} \;.
    \end{align}
    Hence, we can alternatively express the lower bound as 
    \begin{align}
        Q = \Omega \left( \frac{\cL_\theta^2}{\epsilon_\delta} |\log (1-p)|\right) \;.
    \end{align}
\end{proof}

We can turn the input parameters around
and use the same bound to establish 
Thm.~\ref{thm:lowerbound}. That is, instead of being given a $\delta >0$,
assume we are given a $t >0$ and an $\epsilon >0$ such that there exists some $\delta \in (0,\pi/16]$, that satisfies the condition $\epsilon=\epsilon_\delta = 3 (\cL_\delta)^2 \sin \delta /100$. Then, if we run the entry estimation algorithm for the inputs $t$ and $\epsilon$, we can then solve the decision problem in Claim~\ref{claim:lowerbound} for this $\delta$. The lower bound is $\Omega((\cL')^2 |\log(1-p)|/\epsilon)$ for each pair of inputs satisfying the condition,
where now $\cL'=\cL_\theta\le \cL_\delta$ and $\epsilon=\epsilon_\delta$.

Note that the condition can always be satisfied if, for example,
\begin{align}
    \epsilon \le \frac t {100} \le\frac 3 {100} (1-e^{-1})^2 t .
\end{align}
The right hand side is the solution for $\sin \delta =1/t$, assuming $ t \ge 6 \ge 1/\sin(\pi/16)$ for $\delta \le \pi/16$.
Further, since $(\cL_\delta)^2 \sin \delta \rightarrow 0$ as $\delta \rightarrow 0$ for fixed $t>0$, then there will be a solution consistent with the bound above.

Finally, we note that Thm.~\ref{thm:lowerbound} is more general in that it does not constrain to specific instances other than $\bA \preceq 0$ and $\|\bA\|\le 1$. Indeed, the lower bound applies to this wider class of instances, since the prior problem in Lemma~\ref{lem:lb2} is an example in this class.
That is, that problem is at least as difficult as the decision QPE problem in the prior lemma, and the same lower bound applies. 

Moreover, since we allow for instances where $\delta \rightarrow 0$, we can replace $\cL'$ in the established bound for $t$. Also, for these instances, $\cL' \ge \cL =\int_0^t \rd s e^{2s \xi_\bA}$, where $\xi_\bA$ is the log-norm of $\bA$. Hence, we can alternatively express the lower bound as $\Omega(\cL t /\epsilon)$ as in Thm.~\ref{thm:lowerbound}.

\section{Preconditioning the quantum linear systems approach}
\label{app:preconditioning}

In this section, we introduce a general technique for improving the complexity of history state preparation via quantum linear systems solvers. While this technique can be applied to our specific quantum algorithm, it also extends to quantum differential equation solvers more broadly, making the result of independent interest.

The cost of preparing normalized history states via quantum linear system solvers depends linearly on the condition number $\kappa$ of the linear system encoding the discretized differential equation. As shown in Appendix~\ref{app:QAlgorithm}, in the time-independent case, the condition number scales like $\mu \tcL_1$ where $\mu = \max \{a, b\}$ and $\tcL_1$ is a known upper bound on $\max_{\bY \in \{\bA, \bB\}} \lb \int_0^t \rd s \; \norm{e^{s \bY }}  + \frac{d}{c}\max_{s \in [0,t]} \norm{ e^{s\bY }} \rb$. A similar bound can also be obtained in the time-dependent case, see Appendix~\ref{app:QAlgorithm_time-dep}.
This bound on the condition number can be exponential in $t$ in which case the cost of preparing the normalized history state also scales exponentially with $t$. To reduce the complexity, we can use a form of preconditioning as explained below. Also, note that in order to estimate the generalized matrix element $\bra \phi \bX(t) \ket \psi$, we do not actually need to prepare normalized history states, it suffices to prepare subnormalized history states. This is also exploited in the LCHS approach as discussed previously.

Preconditioning is especially useful in cases where the solution norm is growing as a function of time. For simplicity, let us consider the time-independent case. The results extend straightforwardly to the time-dependent case though.
Recall that the goal is to prepare (sub-)normalized versions of
\begin{align}
    \ket{\phi_{\hist}} = \sum_{m=0}^{M-1} \ket{m}e^{t\bA \frac{m}{M}}\ket{\phi}  + \sum_{m=M}^{M + R-1} \ket{m} e^{t \bA} \ket{\phi}, \\
    \nonumber
    \ket{\psi_{\hist}} = \sum_{m=0}^{M-1} \ket{m}e^{t \bB \frac{m}{M}}\ket{\psi}  + \sum_{m=M}^{M + R-1} \ket{m} e^{t \bB} \ket{\psi}.
\end{align}
For our purposes, we choose $M = \lceil t \mu \rceil$ and $R \in \cO \lb \frac{\mu d}{c} \rb$.
The key idea now is to shift $\bA$ and $\bB$ by known upper bounds on their respective log-norms and solve the linear systems with the shifted matrices.
More specifically, define $\bA_\xi :=  \bA - \xi_{\bA} \one$ and $\bB_\xi := \bB - \xi_{\bB} \one$ where $\xi_{\bA}$ ($\xi_{\bB}$) is a known upper bound on the log-norm of $\bA$ ($\bB$). Recall that the log-norm of an arbitrary square matrix $\bA$ is the largest eigenvalue of the Hermitian matrix $\lb \bA + \bA^\dagger \rb/2$, which can be negative in principle.
Then we can use the linear systems approach to prepare the following two normalized history states:
\begin{align}
    \ket{\phi_{\hist, \xi}} &:= \frac{1}{\sqrt{\cN_{\bA_\xi}}} \lb \sum_{m=0}^{M-1} \ket{m}e^{t\bA_\xi \frac{m}{M}}\ket{\phi}  + \sum_{m=M}^{M + R-1} \ket{m} e^{t \bA_\xi} \ket{\phi} \rb \;, \\
    \nonumber
    \ket{\psi_{\hist, \xi}} &:= \frac{1}{\sqrt{\cN_{\bB_\xi}}} \lb \sum_{m=0}^{M-1} \ket{m}e^{t \bB_\xi \frac{m}{M}}\ket{\psi}  + \sum_{m=M}^{M + R-1} \ket{m} e^{t \bB_\xi} \ket{\psi} \rb.
\end{align}
By Lemma~\ref{lem:norm_history_state}, we have that
\begin{align}
    \cN_{\bA_\xi} &\in \cO \lb \frac{M}{t} \int_0^t \rd \, s \norm{e^{s \bA_{\xi}}}^2 + R \norm{e^{t \bA_{\xi}}}^2 \rb \subseteq \cO \lb M + R \rb \\
    \cN_{\bB_\xi} &\in \cO \lb \frac{M}{t} \int_0^t \rd \, s \norm{e^{s \bB_{\xi}}}^2 + R \norm{e^{t \bB_{\xi}}}^2 \rb \subseteq \cO \lb M + R \rb.
\end{align}
Furthermore, from the proof of Theorem~\ref{thm:history_state} we see that the condition numbers of the preconditioned linear systems are now also upper bounded by $\cO \lb M + R \rb \subseteq \cO \lb t \mu + \frac{\mu d}{c} \rb$, meaning that the overall query complexity for preparing $\ket{\phi_{\hist, \xi}}$ and $\ket{\phi_{\hist, \xi}}$ within some error $\epsilon_{\hist, \xi}$ scales like
\begin{align}
    \cO \lb \lb t \mu + \frac{\mu d}{c} \rb \times \log (1/\epsilon_{\hist, \xi}) \times \log \lb \lb t \mu + \frac{\mu d}{c} \rb/\epsilon_{\hist, \xi} \rb \rb.
\end{align}
To get the correct weighting in the amplitudes of the history states, i.e.~to remove the factors $e^{-t \frac{m}{M} \xi_{\bA}}$ and $e^{-t \frac{m}{M} \xi_{\bB}}$ from the amplitudes, we then apply block-encodings of the following diagonal matrices to the clock registers:
\begin{align}
    \mathcal{D}_{\bA} &= \mathrm{diag} \lb 1, e^{\frac{t}{M} \xi_{\bA}}, \dots, e^{t \xi_{\bA}} \rb \in \mathbb{C}^{(M+R) \times (M+R)} \\
    \mathcal{D}_{\bB} &= \mathrm{diag} \lb 1, e^{\frac{t}{M} \xi_{\bB}}, \dots, e^{t \xi_{\bB}} \rb \in \mathbb{C}^{(M+R) \times (M+R)}.
\end{align}
These can be implemented using $\mathrm{polylog} \lb (M+R)/\epsilon_{\hist, \xi} \rb$ primitive gates with block-encoding constants scaling like $\cO \lb \max \left\{1, e^{t \xi_{\bA}} \right\} \rb$ and $\cO \lb \max \left\{1, e^{t \xi_{\bB}} \right\} \rb$, respectively.

Following Theorem~\ref{thm:mainformal}, the overall number of required queries to $U_{\phi}$, $U_{\psi}$, $U_{\bA}$, $U_{\bB}$, $U_{\bC}$ and $U_{\bD}$ for estimating the generalized matrix element $\bra{\phi} \bX(t) \ket{\psi}$ within additive error $\epsilon$ via preconditioned linear systems then scales like
\begin{align}
    \widetilde{\cO} \lb  \frac{ \frac{c}{\mu} \times \lb  t\mu + \frac{\mu d}{c} \rb \times \max \left\{1, e^{t \xi_{\bA}} \right\} \times \max \left\{1, e^{t \xi_{\bB}} \right\}}{\epsilon} \times  \lb t\mu + \frac{\mu d}{c} \rb \rb .
\end{align}
This shows that preconditioning does not improve the query complexity in cases where $\bA$ and $\bB$ have non-positive log-norm. However, in cases of exponential growth, where $\norm{e^{s \bA}} \sim e^{s \xi_{\bA}}$ and $\norm{e^{s \bB}} \sim e^{s \xi_{\bB}}$ with $\xi_{\bA}, \xi_{\bB} > 0$, preconditioning allows us to shave off a factor of $\max \left\{ e^{s \xi_{\bA}}, e^{s \xi_{\bB}} \right\}$ from the overall query complexity given in Theorem~\ref{thm:mainformal} and replace it with a factor $\sim t^2$ instead.

More specifically, the factors in the complexity are $M+R$ for the condition numbers for the solvers, another factor of $M+R$ for the normalizations of the states,
$\max \left\{1, e^{t \xi_{\bA}} \right\} \times \max \left\{1, e^{t \xi_{\bB}} \right\}$ for the subnormalizations induced by $\mathcal{D}_{\bA}$ and $\mathcal{D}_{\bA}$, $c/\mu$ for $\cI$, and
$1/\epsilon$ for amplitude estimation.
In comparison, Theorem~\ref{thm:mainformal} has the factors $\tcL_1$ for the condition numbers for the solvers, $\tcL_2$ for the normalizations of the states, and the factors from $\cI$ and amplitude estimation the same.
The magnitude of $\tcL_2$ can be approximated as $\cO( e^{t \xi_{\bA}} e^{t \xi_{\bB}} )$ for positive log-norms, so is comparable to the normalization factors in the preconditioned case. The preconditioned approach has an extra factor of $t$ in the normalization, and another factor of $t$ in the condition number, rather than the exponential condition number from $\tcL_1$ in Theorem~\ref{thm:mainformal}.

\end{document}